\documentclass[12pt]{article}
\usepackage[margin=.9in]{geometry}
\usepackage{amsmath,amssymb,amsfonts,amsthm}
\usepackage{algorithm}
\usepackage{algpseudocode}
\usepackage{booktabs}
\usepackage{enumitem}
\usepackage{float}
\usepackage{placeins}
\usepackage{graphicx}
\usepackage{xcolor}
\definecolor{citationblue}{RGB}{0,70,180}
\usepackage{hyperref}
\usepackage[authoryear,round,longnamesfirst]{natbib}
\hypersetup{colorlinks=true,citecolor=citationblue,linkcolor=black,urlcolor=black,pdftitle={Grouped Semi-Supervised Estimation of Two-Sample Functionals via Single-Index Conditional Imputation},pdfauthor={Hengyao Xu and Tao Tan},pdfsubject={Grouped semi-supervised inference}}
\newcommand{\Var}{\operatorname{Var}}
\newcommand{\Cov}{\operatorname{Cov}}
\newtheorem{theorem}{Theorem}[section]
\newtheorem{proposition}[theorem]{Proposition}
\newtheorem{lemma}[theorem]{Lemma}
\newtheorem{corollary}[theorem]{Corollary}
\newtheorem{assumption}{Assumption}[section]

\newtheorem{remark}{Remark}[section]
\newtheorem{lemmaB}{Lemma}[section]
\newtheorem{theoremB}[lemmaB]{Theorem}
\newtheorem{propositionB}[lemmaB]{Proposition}
\newtheorem{corollaryB}[lemmaB]{Corollary}
\newtheorem{assumptionB}{Assumption}[section]
\newtheorem{remarkB}{Remark}[section]
\newcommand{\UseAppendixBEnvironments}{%
 \let\lemma\lemmaB\let\endlemma\endlemmaB
 \let\theorem\theoremB\let\endtheorem\endtheoremB
 \let\proposition\propositionB\let\endproposition\endpropositionB
 \let\corollary\corollaryB\let\endcorollary\endcorollaryB
 \let\assumption\assumptionB\let\endassumption\endassumptionB
 \let\remark\remarkB\let\endremark\endremarkB}
\floatstyle{ruled}
\restylefloat{algorithm}
\floatname{algorithm}{Algorithm}
\algrenewcommand\algorithmicrequire{\textbf{Require:}}
\algrenewcommand\algorithmicensure{\textbf{Ensure:}}
\algrenewcommand\alglinenumber[1]{\scriptsize #1:}
\renewcommand{\thefootnote}{\fnsymbol{footnote}}
\title{\parbox{0.9\textwidth}{\centering Grouped Semi-Supervised Estimation of Two-Sample Functionals via Single-Index Conditional Imputation\textsuperscript{\(\dagger\)}}}
\author{Hengyao Xu, Tao Tan\textsuperscript{*}\\[0.3em]\normalsize College of Mathematics and Systems Science, Xinjiang University, Urumqi, Xinjiang, China}
\date{}
\begin{document}
\frenchspacing
\maketitle
\footnotetext[1]{Corresponding author. Email: ttan303@163.com.}
\footnotetext[2]{This research was supported by the Natural Science Foundation of Xinjiang Uygur Autonomous Region (No.~2025D14015) and the National Natural Science Foundation of China (No.~72571102).}
\renewcommand{\thefootnote}{\arabic{footnote}}
\begin{abstract}
Estimating two-sample comparison functionals within heterogeneous groups is challenging when outcomes are observed only for a small labelled subset while covariates are widely available. We develop semi-supervised estimation of group-specific two-sample comparison functionals and their prespecified aggregates. The proposed estimator combines single-index conditional imputation from labelled cross-arm pairs with averaging over all available covariate pairs. To handle index-score densities that may vanish at attainable endpoints, we use centred stabilization, with a smooth stabilizer in the profile criterion and a shrinking lower truncation of the kernel denominator (the hard floor) in the point estimator. The outer average retains all covariate pairs. Under a correctly specified single-index conditional-mean model and the stated regularity conditions, we establish consistency and asymptotic normality for bounded comparison functionals. A four-role original-observation influence representation accounts for sample overlap and identifies the variance component that unlabelled covariates can reduce relative to supervised estimation. We further establish conditional validity of observation-level multinomial perturbation resampling with complete refitting and shared labelled counts, distinguishing ideal variance consistency from finite-replicate variance estimation. Simulations show small bias and lower mean squared error than supervised estimation across the point-estimation settings, together with near-nominal coverage in the perturbation experiment. A descriptive NHANES 2011--2018 illustration comparing serum creatinine by hypertension history across age groups yields point estimates similar to their supervised counterparts and smaller semi-supervised standard errors for all reported comparisons.
\end{abstract}
\noindent\textbf{Keywords:} semi-supervised inference; two-sample comparison; group-specific inference; single-index model; perturbation resampling
\par
\noindent\textbf{MSC 2020:} 62D10, 62G08, 62G20.
\section{Introduction}\label{sec:introduction}

\subsection{Motivation and Grouped Comparison}\label{subsec:intro-motivation}

Semi-supervised inference addresses settings in which response information is expensive, delayed, or partially missing, whereas covariates are broadly available from routine records, registries, or digital systems. Broad overviews include \citet{chapelle2006} and \citet{vanengelen2020}; here the focus is statistical inference rather than prediction. Unlabelled covariates can improve inference on marginal functionals even when they do not reveal responses directly: the target is often an outer expectation of an inner conditional object, and the larger covariate sample refines the outer covariate-distribution component \citep{chakrabortty2018,gronsbell2018,zhang2019,song2024}.

Two-sample problems often target functionals of two responses, including mean contrasts, probabilistic comparisons, and threshold-based differences in survival-type probabilities. These are distributional comparisons; causal interpretations require additional identifying assumptions. Covariates help explain how the comparison varies across the observed covariate space in applications ranging from clinical studies to multi-source experiments.

Many applications are inherently grouped or multi-source, with treatment and control arms observed within hospitals, regions, cohorts, studies, platforms, or administrative units. The treatment-control comparison, covariate distribution, and response process may vary across groups. Simple pooling can therefore blur between-group variation and replace a group-aware estimand by a mixture-level target determined by the pooled design.

Aggregation is part of the inferential target; pooling would generally change the estimand. We distinguish the equally weighted group average, Ave, from an aggregate W with prespecified deterministic weights: the former treats each group as one scientific unit, whereas the latter reflects their chosen relative importance. These alternatives answer different scientific questions when group sizes or group-specific effects differ.

\subsection{Related Literature}\label{subsec:related-literature}

Work on semi-supervised inference for marginal parameters uses unlabelled covariates to improve estimation when responses are observed only on a labelled subset; efficient estimation of marginal functionals also connects to the broader semiparametric theory of \citet{bickel1991}. Regression-based constructions include \citet{chakrabortty2018} and \citet{azriel2022}, while \citet{zhang2019} develop a general framework for mean estimation and \citet{zhangbradic2022} study high-dimensional mean inference. \citet{gronsbell2018} consider prediction-performance functionals, and \citet{song2024} formulate a general $M$-estimation theory. More recent work extends semi-supervised inference to distribution-level processes \citep{wen2025distribution}, extreme-value functionals \citep{ahmed2025extreme}, and likelihood-based efficiency analysis for semi-supervised regression models \citep{tian2026etm}.

For semi-supervised two-sample and pairwise functionals, \citet{tan2025} develop an efficient semiparametric procedure for probabilistic comparison, while \citet{zhangpeng2025} study a broader class of comparison functions across two populations. \citet{kim2025semisupervised} develop semi-supervised $U$-statistics with asymptotic normality and efficiency guarantees, including a refined construction for bivariate kernels. The grouped problem considered here additionally keeps group-specific functionals and prespecified cross-group aggregates explicit, while its estimator must accommodate endpoint-vanishing score densities and the fitting--outer overlap created by reusing labelled covariates in the full pairwise average.

Single-index semiparametric estimation and kernel smoothing reduce a high-dimensional pairwise conditional comparison to one-dimensional smoothing. The semiparametric least-squares (SLS) approach of \citet{ichimura1993}, the smoothing theory of \citet{hardle1993}, and the classical kernel estimators of \citet{nadaraya1964} and \citet{watson1964} provide the standard ingredients. In kernel-based semiparametric $M$-estimation, denominator protection and bandwidth choice are part of the estimator itself \citep{DelecroixHristachePatilea2006}, while $U$-process analysis supplies a route from profile criteria to stochastic equicontinuity \citep{Sherman1994Semiparametric}. Here, endpoint-vanishing score densities and cross-arm reuse also enter the analysis.

Partial-mean theory explains why averaging a nonparametric fit can have a regular first-order law even when pointwise estimation is slower \citep{Newey1994PartialMeans,IchimuraLee2010,IchimuraLee2018Corrigendum}. Inverse-density integration \citep{DelyonPortier2016}, generated-covariate expansions \citep{MammenRotheSchienle2012,MammenRotheSchienle2016}, and single-index nuisance adaptivity \citep{Song2014} provide the corresponding tools for denominator and first-stage effects. The resulting partial-mean problem also involves a hard floor, rectangular two-sample dependence, and reuse of labelled covariates in the outer average.

Hoeffding projections and classical $U$-statistic theory organize first-order terms at the level of original observations \citep{Hoeffding1948,Serfling1980}; uniform control of changing kernel classes uses $U$-process, decoupling, and empirical-process theory \citep{NolanPollard1987,ArconesGine1993,Sherman1994,deLaPenaGine1999,vdVWellner1996}. These tools are applied to low-density endpoint geometry, the centred floor, index-overlap patterns, and the original-observation projection representation.

Perturbing an objective function provides a route to inference for fitted estimators \citep{jin2001}, while exchangeably weighted empirical processes and semiparametric bootstrap theory address conditional distributional approximation \citep{PraestgaardWellner1993,ChengHuang2010}. Bootstrap results for $U$- and $V$-statistics, including two-sample weighting, provide background for dependent pair arrays \citep{ArconesGine1992Bootstrap,HuangLiuPeng2023}, and kernel-based bootstrap theory highlights the role of smoothing in resampling validity \citep{CattaneoJansson2018}. The resampling argument retains score standardization, index estimation, floor crossings, and shared labelled counts. Conditional moment arguments \citep{Cheng2015} support the distinction between ideal variance consistency and its finite-replicate approximation.

\subsection{Methodological Challenges and Approach}\label{subsec:intro-approach}

The estimator uses labelled cross-arm outcomes to learn a group-specific comparison surface through a single index and averages it over all available covariate pairs. Group-specific fitting precedes aggregation, so unlabelled covariates inform the outer distribution without changing the target.

A technical difficulty is that the fitted scalar score has compact support and a density that may vanish at attainable endpoints. We therefore use a smooth stabilizer in the profile criterion and a shrinking hard floor in the point estimator, without trimming outer covariate-pair evaluation points. Off the floor, the estimator reduces to ordinary Nadaraya--Watson smoothing, while centring leaves the population target and first-order influence functions unchanged.

Labelled observations are reused across all cross-arm fitting pairs and also enter the full outer covariate sample. Pair counts therefore do not represent independent observations, and fitting and outer averaging are not independent stages. The first-order analysis regroups these contributions into four original-observation roles and identifies, under the correct single-index model, the variance component reducible by additional unlabelled covariates.

Multinomial perturbation assigns counts to original observations, reuses each labelled count through fitting and outer averaging, and refits the complete estimator in every replicate. The conditional theory distinguishes the ideal perturbation variance from its finite-replicate estimate.

\subsection{Contributions and Organization}\label{subsec:intro-contributions}

The paper treats group-specific two-sample functionals as the primary estimands and defines equal-group and fixed weighted aggregates explicitly, separating aggregation from the within-group comparisons. For score densities that may vanish at attainable endpoints, the estimator combines single-index conditional imputation, smooth profile stabilization, and a shrinking hard floor while retaining the full covariate-pair outer average. The first-order theory accounts for cross-arm reuse and labelled--outer overlap through four original-observation projections: two from the full covariate distributions and two from labelled residual variation. Under the correct single-index model, this representation yields the supervised--semi-supervised variance comparison and inference for fixed-$K$ aggregates. The perturbation procedure assigns multinomial counts to original observations and refits the complete estimator, reusing labelled counts in both fitting and outer averaging. The accompanying theory establishes conditional distributional validity and ideal conditional variance consistency; finite-replicate Wald inference additionally requires the stated fourth-moment condition.

The numerical studies assess increasing unlabelled information, heterogeneous grouped targets, a bounded clipped contrast, unequal allocation, and full-refit perturbation inference. The NHANES analysis uses controlled masking to illustrate group-specific and aggregate estimation.

Section~\ref{sec:methodology} introduces the grouped targets and estimator; Section~\ref{sec:theory} develops point-estimation and first-order theory; and Section~\ref{sec:perturbation} establishes perturbation inference. Sections~\ref{sec:simulation}--\ref{sec:discussion} present the numerical evidence, NHANES illustration, and discussion. Appendices~\ref{app:A}--\ref{app:C} collect notation, auxiliary results, and proofs of the main results.

\section{Grouped Semi-Supervised Estimation}\label{sec:methodology}

\subsection{Targets and Data Structure}\label{subsec:data}

Consider $K$ statistically meaningful groups, centres, or sources indexed by $k=1,\ldots,K$. Within each group, let $r=1$ denote the control arm and let $r=2$ denote the treatment arm. For arm $r$ in group $k$, the labelled sample is
\begin{equation*}
\mathcal L_{rk}
=
\{(Y_{rki},\mathbf{X}_{rki}): i=1,\ldots,n_{rk}\},
\end{equation*}
where $Y_{rki}$ is the response and $\mathbf{X}_{rki}$ is a covariate vector. The unlabelled covariate sample is
\begin{equation*}
\mathcal U_{rk}
=
\{\mathbf{X}_{rki}: i=n_{rk}+1,\ldots,M_{rk}\},
\qquad
M_{rk}=n_{rk}+N_{rk}.
\end{equation*}
The full observed covariate sample is
\begin{equation*}
\mathcal E_{rk}
=
\{\mathbf{X}_{rki}: i=1,\ldots,M_{rk}\}.
\end{equation*}
Within each arm and group, the labelled units are nested in the full covariate sample, and labelled and unlabelled covariates are sampled from the same target covariate distribution. Semi-supervision enters separately in the two arms, while covariate distributions and outcome mechanisms may vary across groups. The analysis treats $K$ as fixed, with independent treatment and control samples within groups and independent groups across $k$.

\phantomsection\label{subsec:targets}
For each group $k$, define the target functional
\begin{equation*}
\Delta_k(t)=E\{L(Y_{2k},Y_{1k},t)\},
\qquad
k=1,\ldots,K,
\end{equation*}
where $L$ is a prespecified comparison function and $t$ is a threshold or functional argument when needed. The expectation is taken under the product distribution of independent treatment and control draws within group $k$; thus, for example, $P(Y_{2k}>Y_{1k})$ is an independent two-sample comparison rather than a matched or paired estimand.

\medskip\noindent\emph{Example 1 (Separable average-treatment-effect-type contrast).}
Take
\begin{equation*}
L(Y_2,Y_1)=Y_2-Y_1,
\qquad
\Delta_k=E(Y_{2k})-E(Y_{1k}).
\end{equation*}
This target is separable, with $a(y)=b(y)=y$, and serves as a conceptual benchmark for the arm-specific reduction of the pairwise construction. The formal root-scale theory below concerns bounded comparison responses.

\medskip\noindent\emph{Example 2 (Probabilistic comparison effect).}
Take
\begin{equation}
L(Y_2,Y_1)=\mathbf 1\{Y_2>Y_1\},
\qquad
\Delta_k=P(Y_{2k}>Y_{1k}).
\label{eq:probability-example}
\end{equation}
This target is genuinely pairwise and nonseparable, unlike the case in Remark~\ref{rem:separability}.

\medskip\noindent\emph{Example 3 (Threshold or tail comparison contrast).}
Take
\begin{equation*}
L(Y_2,Y_1,t)=\mathbf 1\{Y_2>t\}-\mathbf 1\{Y_1>t\},
\qquad
\Delta_k(t)=P(Y_{2k}>t)-P(Y_{1k}>t).
\end{equation*}
This target is separable for each prespecified fixed $t$ and illustrates how the framework accommodates fixed-threshold contrasts.

\begin{remark}[Separability]\label{rem:separability}
If the comparison function satisfies
\begin{equation*}
L(Y_2,Y_1)=a(Y_2)-b(Y_1),
\end{equation*}
then the grouped target decomposes into two marginal components,
\[
\Delta_k=E\{a(Y_{2k})\}-E\{b(Y_{1k})\}.
\]
The general pairwise framework includes this case, although arm-specific implementations can avoid explicit cross-arm pair construction.
\end{remark}

Throughout, ``treatment'' and ``control'' label the two samples; $\Delta_k(t)$ is interpreted descriptively unless additional causal identification assumptions are imposed.

\phantomsection\label{subsec:aggregate-targets}
We consider two prespecified aggregate targets. The equal-group functional is
\begin{equation}
\Delta_{\mathrm{Ave}}(t)
=
\frac{1}{K}\sum_{k=1}^K \Delta_k(t),
\label{eq:aggregate-ave}
\end{equation}
which treats each group as one scientific unit.

The weighted functional is
\begin{equation}
\Delta_{\mathrm{W}}(t)
=
\sum_{k=1}^K a_k \Delta_k(t),
\qquad
a_k\ge 0,
\qquad
\sum_{k=1}^K a_k=1.
\label{eq:aggregate-w}
\end{equation}
When $a_k=1/K$, $\Delta_{\mathrm{W}}(t)=\Delta_{\mathrm{Ave}}(t)$; otherwise the $a_k$ are prespecified scientific weights, such as external population proportions, rather than weights determined by the observed group sizes.

\subsection{Single-Index Conditional Imputation}\label{subsec:conditional-imputation}

Define the group-specific conditional comparison function
\begin{equation}
m_k(\mathbf{x}_2,\mathbf{x}_1;t)
=
E\{L(Y_{2k},Y_{1k},t)\mid \mathbf{X}_{2k}=\mathbf{x}_2, \mathbf{X}_{1k}=\mathbf{x}_1\}.
\label{eq:conditional-object}
\end{equation}
By the law of iterated expectation,
\begin{equation*}
\Delta_k(t)=E\{m_k(\mathbf{X}_{2k},\mathbf{X}_{1k};t)\}.
\end{equation*}
The identity separates an inner conditional-comparison fit from an outer covariate-distribution average. Labelled cross-arm pairs identify the inner object, while the full covariate samples estimate the outer expectation; unlabelled covariates refine this distributional average without revealing missing outcomes.

For indicator-valued $L$, such as the probabilistic comparison in \eqref{eq:probability-example}, smoothing is applied to the conditional comparison in \eqref{eq:conditional-object}, not to the indicator itself.

\phantomsection\label{subsec:single-index}
Identification is based on a population profile. The conditional comparison function $m_k(\mathbf{x}_2,\mathbf{x}_1;t)$ generally depends jointly on both arm-specific covariates, making direct nonparametric smoothing in the full pairwise covariate space impractical at moderate dimension. We therefore use the working single-index conditional-mean model
\begin{equation*}
m_k(\mathbf{x}_2,\mathbf{x}_1;t)
=E\!\left\{ L(Y_{2k},Y_{1k},t)\ \middle|\ 
\mathbf{X}_{2k}^{\top}\boldsymbol{\beta}_{2k,0}+\mathbf{X}_{1k}^{\top}\boldsymbol{\beta}_{1k,0}
=\mathbf{x}_2^{\top}\boldsymbol{\beta}_{2k,0}+\mathbf{x}_1^{\top}\boldsymbol{\beta}_{1k,0}\right\}.
\end{equation*}
Here $\boldsymbol{\beta}_{2k,0}$ and $\boldsymbol{\beta}_{1k,0}$ are the treatment- and control-arm index parameters, and
\begin{equation*}
\boldsymbol{\beta}_{k,0}=(\boldsymbol{\beta}_{2k,0}^{\top},\boldsymbol{\beta}_{1k,0}^{\top})^{\top},
\qquad
S_k(\boldsymbol{\beta}_k)=\mathbf{X}_{2k}^{\top}\boldsymbol{\beta}_{2k}+\mathbf{X}_{1k}^{\top}\boldsymbol{\beta}_{1k}.
\end{equation*}
Semiparametric least squares and single-index smoothing build on the classical work of \citet{ichimura1993} and \citet{hardle1993}, respectively. The model reduces the pairwise covariates to a scalar index while allowing the conditional mean and index coefficients to vary across groups and arms. General pairwise targets use the joint index, with simpler arm-specific reductions for separable functionals.

Put $\mathbf{V}_k=(\mathbf{X}_{2k}^{\top},\mathbf{X}_{1k}^{\top})^{\top}$ and
$H_k(t)=Z_k(t)=L(Y_{2k},Y_{1k},t)$. Fix a bounded reference $c_{H,k}(t)$ from the scientific target before observing the data and write
\begin{equation*}
H_{c,k}(t)=H_k(t)-c_{H,k}(t).
\end{equation*}
The probability target uses $c_H=1/2$, while the clipped target
$H=\operatorname{clip}(Y_2-Y_1,-2,2)$ uses $c_H=0$. Thus $c_H$ is a target-specific prespecified reference constant. For every candidate direction $\boldsymbol{\beta}$, let $S_{k,\boldsymbol{\beta}}=\mathbf{V}_k^{\top}\boldsymbol{\beta}$ and define the population profile
\begin{equation*}
m_{k,\boldsymbol{\beta}}(s;t)=E\{Z_k(t)\mid \mathbf{V}_k^{\top}\boldsymbol{\beta}=s\},\qquad
Q_k(\boldsymbol{\beta};t)=E\left[\{Z_k(t)-m_{k,\boldsymbol{\beta}}(\mathbf{V}_k^{\top}\boldsymbol{\beta};t)\}^2\right].
\end{equation*}
Since $E(H_{c,k}\mid S)=m_{k,\boldsymbol{\beta}}(S)-c_{H,k}$, the ideal centred profile criterion is exactly $Q_k(\boldsymbol{\beta};t)$ for every $\boldsymbol{\beta}$. Centring therefore modifies the finite stabilized estimator while leaving the scientific target, ideal population minimizer, and identification problem unchanged.
The profile criterion is an $M$-estimation problem; its consistency analysis combines uniform control of the empirical criterion with standard extremum-estimation and approximate-optimisation arguments \citep{PakesPollard1989,newey1994}.
The normalized parameter space is
\begin{equation*}
\mathcal B_k=\{\boldsymbol{\beta}:\|\boldsymbol{\beta}\|_2=1,\ \text{the first nonzero component of }\boldsymbol{\beta}\text{ is positive}\},
\end{equation*}
and the identified profile direction is $\boldsymbol{\beta}_{k,0}=\arg\min_{\boldsymbol{\beta}\in\mathcal B_k}Q_k(\boldsymbol{\beta};t)$. The normalization fixes both scale and sign, and is consequently part of identification and of the tangent-space expansion used in Section~\ref{sec:theory}.

The definition of $\boldsymbol{\beta}_{k,0}$ by profile minimization is distinct from the correct single-index conditional-mean condition used for the asymptotic linear representation of the feasible estimator:
\begin{equation*}
 E\{Z_k(t)\mid \mathbf{V}_k\}=m_{k,\boldsymbol{\beta}_{k,0}}(\mathbf{V}_k^{\top}\boldsymbol{\beta}_{k,0};t).
\end{equation*}
The latter says that the full conditional mean is correctly represented by the identified scalar index. The asymptotic linear representation below is developed under this correctly specified single-index condition.

\phantomsection\label{subsec:estimation}
For estimation, fix $\boldsymbol{\beta}_k$ and define the labelled-pair index
\begin{equation*}
S_{kij}(\boldsymbol{\beta}_k)
=
\mathbf{X}_{2ki}^{\top}\boldsymbol{\beta}_{2k}+\mathbf{X}_{1kj}^{\top}\boldsymbol{\beta}_{1k}.
\end{equation*}
Let $K$ be the standard Gaussian density and
$K_{h_k}(u)=h_k^{-1}K(u/h_k)$. Define the empirical numerator, denominator,
and centred numerator
\begin{equation*}
\widehat g_{k,\boldsymbol{\beta}}(s;t)=\frac1{n_{2k}n_{1k}}\sum_{i=1}^{n_{2k}}\sum_{j=1}^{n_{1k}}
K_{h_k}\{S_{kij}(\boldsymbol{\beta})-s\}Z_{kij}(t),\quad
\widehat d_{k,\boldsymbol{\beta}}(s)=\frac1{n_{2k}n_{1k}}\sum_{i=1}^{n_{2k}}\sum_{j=1}^{n_{1k}}
K_{h_k}\{S_{kij}(\boldsymbol{\beta})-s\}.
\end{equation*}
\begin{equation*}
 \widehat q_{k,\boldsymbol{\beta},c}(s;t)
 =\widehat g_{k,\boldsymbol{\beta}}(s;t)
  -c_{H,k}(t)\widehat d_{k,\boldsymbol{\beta}}(s).
\end{equation*}
The smooth profile stabilizer is
\begin{equation*}
 \rho_w(d)=w r(d/w),
\end{equation*}
where $r$ is fixed, nondecreasing, and $C^3$, equals $1$ on $[0,1/2]$,
equals the identity on $[2,\infty)$, and is comparable to $1\vee x$.
Thus $\rho_w$ equals $w$ below $w/2$, equals
$d$ above $2w$, and is comparable to $d\vee w$. The profile fit is
\begin{equation}
 \widehat m_{k,\boldsymbol{\beta}}^{\mathrm{prof},c}(s;t)
 =c_{H,k}(t)+
 \frac{\widehat q_{k,\boldsymbol{\beta},c}(s;t)}
 {\rho_{w_k}\{\widehat d_{k,\boldsymbol{\beta}}(s)\}}.
 \label{eq:centred-profile-fit}
\end{equation}
On the region where $\widehat d_{k,\boldsymbol{\beta}}(s)\ge 2w_k$, \eqref{eq:centred-profile-fit}
reduces to the ordinary pairwise Nadaraya--Watson fit \citep{nadaraya1964,watson1964}.
The ensuing full-covariate average is a generated-index partial mean \citep{Newey1994PartialMeans,IchimuraLee2010,IchimuraLee2018Corrigendum,MammenRotheSchienle2012,MammenRotheSchienle2016}.
The empirical profile-SLS criterion and estimator are
\begin{equation*}
\widehat{\boldsymbol{\beta}}_k
=
\arg\min_{\boldsymbol{\beta}_k\in\mathcal B_k}
\frac{1}{n_{2k}n_{1k}}
\sum_{i=1}^{n_{2k}}\sum_{j=1}^{n_{1k}}
\left[
L(Y_{2ki},Y_{1kj},t)
-\widehat m_{k,\boldsymbol{\beta}_k}^{\mathrm{prof},c}
 \{S_{kij}(\boldsymbol{\beta}_k);t\}
\right]^2,
\end{equation*}
where $Z_{kij}(t)=L(Y_{2ki},Y_{1kj},t)$. The criterion uses the leave-in mean squared residual over all labelled pairs; point estimation retains the centred numerator but uses the denominator truncation defined below.

At the estimated direction $\widehat{\boldsymbol{\beta}}_k$, we impose a shrinking lower
truncation on the kernel denominator, replacing $\widehat d$ by
$\widehat d\vee w_k$; we refer to this device as the hard floor:
\begin{equation}
 \widehat m_{k}^{H,c}(s;t)
 =c_{H,k}(t)+
 \frac{\widehat q_{k,\widehat{\boldsymbol{\beta}}_k,c}(s;t)}
 {\widehat d_{k,\widehat{\boldsymbol{\beta}}_k}(s)\vee w_k}.
 \label{eq:centred-hard-fit}
\end{equation}
The smooth profile regularises derivatives with respect to the local index
coordinates, whereas the hard-floor fit is used for point estimation.

\begin{remark}[Centred stabilization]\label{rem:centred-identities}
Whenever the active stabilizer equals $d>0$,
\[
c_H+\frac{g-c_Hd}{d}=\frac{g}{d}.
\]
Hence centring leaves the ordinary Nadaraya--Watson fit unchanged whenever
the floor or stabilizer is inactive; exact score-coordinate equivariance is
established in Lemma~\ref{lem:b6-candidate-scale}.
\end{remark}

\begin{remark}[Pairwise dependence]\label{rem:pairwise-dependence}
The pairwise construction reuses each labelled observation across all opposite-arm pairs, inducing two-sample $U$-statistic-type dependence rather than independent-pair sampling. Accordingly, the theory must account for this dependence in both nuisance fitting and the final stochastic expansion.
\end{remark}

\subsection{Group and Aggregate Estimators}\label{subsec:group-estimator}

The proposed group-level semi-supervised plug-in estimator is
\begin{equation*}
\widehat\Delta_k^c(t)\equiv\widehat\Delta_k(t)
=
\frac{1}{M_{2k}M_{1k}}
\sum_{i=1}^{M_{2k}}
\sum_{j=1}^{M_{1k}}
\widehat m_{k}^{H,c}\left(
\mathbf{X}_{2ki}^{\top}\widehat{\boldsymbol{\beta}}_{2k}
+
\mathbf{X}_{1kj}^{\top}\widehat{\boldsymbol{\beta}}_{1k};
t
\right).
\end{equation*}
Write $\widehat\Delta_k^c(\boldsymbol{\beta};t)$ for the same expression with $\widehat{\boldsymbol{\beta}}_k$ replaced by $\boldsymbol{\beta}$; hence $\widehat\Delta_k^c(t)=\widehat\Delta_k^c(\widehat{\boldsymbol{\beta}}_k;t)$. 
The labelled pairs determine the conditional fit, whereas all covariate units enter the outer average; labelled units therefore contribute to both stages. The infeasible oracle estimator uses the true index in the same centred hard-floor fit:
\begin{equation}
 \widehat\Delta_k^{\mathrm{or},c}(t)=\widehat\Delta_k^c(\boldsymbol{\beta}_{k,0};t).
\label{eq:oracle-estimator}
\end{equation}
The hard-floor construction retains all $M_{2k}M_{1k}$ covariate pairs in the outer average.

\phantomsection\label{subsec:aggregate-estimator}
The corresponding equal-group aggregate estimator is
\begin{equation}
\widehat\Delta_{\mathrm{Ave}}(t)
=
\frac{1}{K}\sum_{k=1}^K \widehat\Delta_k(t),
\label{eq:aggregate-estimator-ave}
\end{equation}
and the weighted aggregate estimator is
\begin{equation}
\widehat\Delta_{\mathrm{W}}(t)
=
\sum_{k=1}^K a_k \widehat\Delta_k(t).
\label{eq:aggregate-estimator-w}
\end{equation}
These estimators target $\Delta_{\mathrm{Ave}}(t)$ and $\Delta_{\mathrm{W}}(t)$ in \eqref{eq:aggregate-ave} and \eqref{eq:aggregate-w}, respectively.

\phantomsection\label{subsec:benchmarks}
For comparison, define the supervised group estimator using labelled data only:
\begin{equation*}
\widehat\Delta_{\mathrm{SUP},k}(t)
=
\frac{1}{n_{2k}n_{1k}}
\sum_{i=1}^{n_{2k}}
\sum_{j=1}^{n_{1k}}
L(Y_{2ki},Y_{1kj},t).
\end{equation*}
The corresponding supervised aggregate estimators are
\begin{equation*}
\widehat\Delta_{\mathrm{sup,Ave}}(t)
=
\frac{1}{K}\sum_{k=1}^K \widehat\Delta_{\mathrm{SUP},k}(t),
\qquad
\widehat\Delta_{\mathrm{sup,W}}(t)
=
\sum_{k=1}^K a_k \widehat\Delta_{\mathrm{SUP},k}(t).
\end{equation*}
They target the same estimands as their semi-supervised counterparts but use labelled data only.

\begin{remark}[Supervised, oracle, and feasible procedures]\label{rem:supervised-oracle-practical}
The supervised estimator uses labelled outcomes only. The oracle semi-supervised estimator uses the true index, labelled pairs to fit the conditional profile, and the full covariate-pair sample for the outer average. The feasible estimator replaces the true index by the profile-SLS estimate while retaining precisely the same fitting and outer construction.
\end{remark}

\phantomsection\label{subsec:tuning-rule}
We use deterministic bandwidth and floor sequences $(\bar h_k)$ and $(\bar w_k)$ satisfying,
uniformly over the fixed groups,
\begin{equation}
 \bar h_k\asymp n_{\min,k}^{-\alpha},\qquad
 n_{\min,k}=\min(n_{1k},n_{2k}),
 \label{eq:bandwidth-rule}
\end{equation}
and
\begin{equation}
 \bar w_k\asymp n_{L,k}^{-\eta},\qquad
 n_{L,k}=(n_{1k}^{-1}+n_{2k}^{-1})^{-1}.
 \label{eq:floor-rule}
\end{equation}
The common exponents $(\alpha,\eta)$ satisfy the admissible rate
conditions stated after Assumption~\ref{ass:E}. Standardization of the score for each candidate index direction is exactly a coordinate
change: the numerator, denominator, and centred numerator all acquire the
same inverse score-scale factor, while the add-back $c_H$ is invariant.
In standardized score coordinates, $h_k=\bar h_k$ and $w_k=\bar w_k$.
In original score coordinates, if $\widehat\sigma_{S,k,\boldsymbol{\beta}}$ is the
empirical pair-score standard deviation for candidate direction $\boldsymbol{\beta}$, then
$h_{k,\boldsymbol{\beta}}=\widehat\sigma_{S,k,\boldsymbol{\beta}}\bar h_k$ and
$w_{k,\boldsymbol{\beta}}=\bar w_k/\widehat\sigma_{S,k,\boldsymbol{\beta}}$.
Consequently the standardized-score and original-score profile and hard-floor fits
are identical pointwise when bandwidth and floor are transformed together.

\phantomsection\label{subsec:implementation-order}
The estimation workflow is
\begin{equation*}
\{\mathcal L_{rk},\mathcal U_{rk}:r=1,2;\,k=1,\ldots,K\}
\longrightarrow
\{\widehat{\boldsymbol{\beta}}_k,\widehat m_{k}^{H,c}\}_{k=1}^K
\longrightarrow
\{\widehat\Delta_k(t)\}_{k=1}^K
\longrightarrow
\{\widehat\Delta_{\mathrm{Ave}}(t),\widehat\Delta_{\mathrm{W}}(t)\}.
\end{equation*}
Algorithm~\ref{alg:proposed-estimator} summarises the resulting implementation.

\begin{algorithm}[H]
\caption{Grouped Single-Index Conditional-Imputation Estimator}
\label{alg:proposed-estimator}
\begin{algorithmic}[1]
\Require Labelled samples $\mathcal L_{rk}$, full covariate samples $\mathcal E_{rk}$, bounded comparison function $L$, target reference $c_H$, Gaussian kernel $K$, prespecified scales $(\bar h_k,\bar w_k)$, and aggregate weights $a_k$
\Ensure Group-level estimates $\{\widehat\Delta_k(t)\}_{k=1}^K$ together with $\widehat\Delta_{\mathrm{Ave}}(t)$ and $\widehat\Delta_{\mathrm{W}}(t)$
\For{$k=1,\ldots,K$}
\State \textbf{Step 1:} Construct $H_{kij}=L(Y_{2ki},Y_{1kj},t)$ and $H_{c,kij}=H_{kij}-c_{H,k}$ for every labelled cross-arm pair
\State \textbf{Step 2:} For each candidate direction, standardize its labelled pair scores, apply the prespecified standardized bandwidth, and form $\widehat d$, $\widehat g$, and $\widehat q_c=\widehat g-c_H\widehat d$
\State \textbf{Step 3:} Estimate $\widehat{\boldsymbol{\beta}}_k$ by the leave-in empirical profile-SLS criterion using the smooth centred fit in \eqref{eq:centred-profile-fit}
\State \textbf{Step 4:} Refit at $\widehat{\boldsymbol{\beta}}_k$ with the centred hard floor in \eqref{eq:centred-hard-fit} and average over all $M_{2k}M_{1k}$ full-covariate pairs
\EndFor
\State \textbf{Step 5:} Compute the equal-group and weighted aggregate estimators in \eqref{eq:aggregate-estimator-ave} and \eqref{eq:aggregate-estimator-w}
\end{algorithmic}
\end{algorithm}

\begin{remark}[Relation to recent semi-supervised comparison methods]\label{rem:recent-comparison-methods}
Recent semi-supervised procedures for pairwise targets provide close reference points for the estimators above \citep{tan2025,zhangpeng2025,kim2025semisupervised}. For the grouped low-density setting considered here, the present construction has three distinguishing features. First, conditional fitting is performed within group before any prespecified cross-group aggregation, so the group-specific scientific targets are not replaced by a pooled mixture target. Second, all $M_{2k}M_{1k}$ outer covariate pairs are retained, with endpoint-vanishing score densities handled through denominator stabilization rather than deletion of low-density evaluation pairs. Third, the original-observation formulation carries labelled fitting--outer reuse explicitly into both the variance decomposition and the full-refit perturbation scheme developed in Sections~\ref{sec:theory}--\ref{sec:perturbation}.
\end{remark}

\section{Asymptotic Properties}\label{sec:theory}

This section develops the point-estimation theory for the estimator in
Section~\ref{sec:methodology}. The analysis identifies the index, derives the
four original-observation contributions to the group estimator, and aggregates
the resulting group expansions. Because the score density may vanish at
attainable endpoints, denominator control is integrated against the evaluation
law rather than imposed through a uniform lower bound. The first-order
decomposition uses two-sample Hoeffding projections
\citep{Hoeffding1948,Serfling1980}, while uniform pair-process bounds rely on
$U$-process, decoupling, and empirical-process methods
\citep{NolanPollard1987,ArconesGine1993,Sherman1994,deLaPenaGine1999,vdVWellner1996}.
Throughout, the comparison argument $t$ and the number of groups $K$ are fixed,
and the first-order analysis is indexed by the original observations rather
than by cross-arm pairs.
Suppressing the group index in a single-group argument, write
$\mathcal B=\mathcal B_k$.
For a generic observed unit, write $O_{rk}=(Y_{rk},\mathbf X_{rk})$;
$S_k=S_k(\boldsymbol{\beta}_{k,0})$ denotes the true-index score, and
$S_k^X$ denotes the same score evaluated at covariate arguments in the outer
projections.

\subsection{Assumptions and Index Estimation}

For group $k$, put
\[
 n_{L,k}=(n_{2k}^{-1}+n_{1k}^{-1})^{-1},\qquad
 n_{\min,k}=\min(n_{1k},n_{2k}),
\]
and
\[
 \nu_k^{-1}=n_{2k}^{-1}+n_{1k}^{-1}
             +M_{2k}^{-1}+M_{1k}^{-1}.
\]
Here $n_{L,k}$ is the harmonic labelled-arm scale, $n_{\min,k}$ is the
smaller labelled-arm size, and $\nu_k$ is the effective sample-size scale
for the four-role fluctuation.

\begin{assumption}[Sampling, nesting, and boundedness]\label{ass:A}
Within each fixed group, the two arms are independent, observations are iid
within arm, and groups are mutually independent. In arm $r$, observations
$1,\ldots,n_{rk}$ are labelled and their covariates are the first $n_{rk}$
entries of the full array of size $M_{rk}\ge n_{rk}$. The labelled arm sizes
are comparable, $n_{rk},M_{rk}\to\infty$, and $n_{rk}/M_{rk}$ has a limit in
$[0,1]$. The mark
$H_k=L(Y_{2k},Y_{1k},t)$ and deterministic reference $c_{H,k}$ are bounded,
and $\mathbf{V}_k=(\mathbf{X}_{2k}^\top,\mathbf{X}_{1k}^\top)^\top$ has compact support. Set
$H_{c,k}=H_k-c_{H,k}$ and $C_{c,k}=\|H_{c,k}\|_\infty$.
\end{assumption}

\begin{assumption}[Single-index identification]\label{ass:C}
The parameter space $\mathcal B$ is a compact sign-anchored unit sphere, and a unique
$\boldsymbol{\beta}_{k,0}$ satisfies
\[
 E(H_k\mid \mathbf{V}_k)=m_k(\mathbf{V}_k^\top\boldsymbol{\beta}_{k,0}),
\]
where $m_k$ is twice continuously differentiable on the score-support interior,
with the required one-sided endpoint derivatives. The ideal criterion
\[
 Q_k(\boldsymbol{\beta})=E[\{H_k-E(H_k\mid \mathbf{V}_k^\top\boldsymbol{\beta})\}^2]
\]
is uniquely minimized at $\boldsymbol{\beta}_{k,0}$ and uniformly separated from it outside
every neighbourhood. The point $\boldsymbol{\beta}_{k,0}$ is interior to the local
sign-anchored parameterisation, and the local profile Hessian defined below is
positive definite.
\end{assumption}

For local analysis choose $\mathbf{R}_k$ with $\mathbf{R}_k^\top\boldsymbol{\beta}_{k,0}=0$ and write
\[
 \boldsymbol{\beta}_k(\boldsymbol{\vartheta})=
 \frac{\boldsymbol{\beta}_{k,0}+\mathbf{R}_k\boldsymbol{\vartheta}}
      {\|\boldsymbol{\beta}_{k,0}+\mathbf{R}_k\boldsymbol{\vartheta}\|}.
\]
Here $\boldsymbol{\vartheta}$ is the local tangent coordinate and
$\mathbf R_k$ is a basis for the tangent space at
$\boldsymbol{\beta}_{k,0}$.
The anchor coordinate is nonzero, and no other zero coordinate of
$\boldsymbol{\beta}_{k,0}$ is fixed by this parametrisation. Let $\mathbf{d}_k(\mathbf{V}_k)$ be the derivative
with respect to $\boldsymbol{\vartheta}$ of
the population conditional profile and set
\[
 \mathbf{H}_k=E\{\mathbf{d}_k(\mathbf{V}_k)\mathbf{d}_k(\mathbf{V}_k)^\top\}.
\]

\begin{assumption}[Low-density score and profile regularity]\label{ass:E}
\textnormal{(a) Low-density score geometry.}
The true-index score has compact support $[a_k,b_k]$, an interior-positive
continuous density, and endpoint order $p_{0,k}-1$, where
$p_{0,k}=p_{1k}+p_{2k}$ and $p_{rk}\ge2$. On each finite sign/zero stratum of
the fixed-dimensional local tangent parametrisation, its moving support
$[a_{k,\boldsymbol{\vartheta}},b_{k,\boldsymbol{\vartheta}}]$ has three-times continuously differentiable
endpoints and
\[
 f_{k,\boldsymbol{\vartheta}}(a_{k,\boldsymbol{\vartheta}}+r)=r^{p_k(\boldsymbol{\vartheta})-1}c_{-,k,\boldsymbol{\vartheta}}(r),\qquad
 f_{k,\boldsymbol{\vartheta}}(b_{k,\boldsymbol{\vartheta}}-r)=r^{p_k(\boldsymbol{\vartheta})-1}c_{+,k,\boldsymbol{\vartheta}}(r),
\]
with constant active contribution count; the arm-score densities have
one-sided orders $p_{rk}-1$. The coefficient functions and moving supports
satisfy the uniform smoothness and nondegeneracy conditions in
\hyperref[app:B:assumption-technical]{Appendix~\ref*{app:B}}.
\par\medskip
\textnormal{(b) Kernel-process regularity.}
The Gaussian-kernel classes required for the level, derivative, deletion,
hard-floor, and overlap expansions satisfy the uniform VC-type entropy, moment,
projection, and stochastic-equicontinuity conditions stated at the beginning
of Appendix~\ref*{app:B}.
\par\medskip
\textnormal{(c) Global profile regularity.}
For $\sigma_{\boldsymbol{\beta}}^2=\Var(\mathbf{X}_2^\top\boldsymbol{\beta}_2)+\Var(\mathbf{X}_1^\top\boldsymbol{\beta}_1)$,
$\inf_{\boldsymbol{\beta}\in\mathcal B}\sigma_{\boldsymbol{\beta}}\ge c_\sigma>0$. The finite-$(h,w)$
population profile criterion approximates the ideal criterion uniformly on
$\mathcal B$, and the sample-standardized empirical criterion obeys the
corresponding global uniform law under the conditions in Appendix~\ref{app:B}.
\end{assumption}

Parts (a)--(c) separate the three roles of Assumption~\ref{ass:E}. Part (a)
permits the score density to vanish at attainable endpoints and quantifies the
resulting low-density geometry; part (b) supplies the uniform empirical-process
control needed for the kernel, deletion, and overlap expansions; and part (c)
provides the global profile approximation used for consistency. In particular,
the theory does not impose a uniform positive lower bound on the score density.

For later use, let $p_{*,k}$ be the largest endpoint order over the finite
local parameter strata and set
\[
 \gamma_{0,k}=\frac{p_{0,k}}{p_{0,k}-1},\qquad
 \gamma_{*,k}=\frac{p_{*,k}}{p_{*,k}-1}.
\]
Because $K$ is fixed and the tuning exponents are common across groups, put
\[
 p_* = \max_{1\le k\le K}p_{*,k},\qquad
 \gamma_* = \frac{p_*}{p_*-1}=\min_{1\le k\le K}\gamma_{*,k}.
\]
The theorem-level admissible rate region is
\begin{equation*}\tag{R}\label{eq:admissible-rate-region}
 \frac14<\alpha<\frac25,
 \qquad
 \frac{1}{2\gamma_*}<\eta<
 \min\left\{\frac12,\,1-\frac{3\alpha}{2}\right\}.
\end{equation*}
The global worst-case exponent controls the floor bias in every group; the
remaining group- and stratum-specific overlap conditions are implied by this
region under the endpoint orders in Assumption~\ref{ass:E}.
Throughout Section~\ref{sec:theory}, the scales in
\eqref{eq:bandwidth-rule}--\eqref{eq:floor-rule} use common exponents in
this region.

Centring leaves the ideal criterion unchanged: for every candidate $\boldsymbol{\beta}$,
\[
 q_{k,\boldsymbol{\beta},c}(s)
 =d_{k,\boldsymbol{\beta}}(s)\{m_{k,\boldsymbol{\beta}}(s)-c_{H,k}\},
\]
and hence
\begin{equation*}
 Q_{k,c}(\boldsymbol{\beta})
 =E[\{H_{c,k}-E(H_{c,k}\mid \mathbf{V}_k^\top\boldsymbol{\beta})\}^2]
 =Q_k(\boldsymbol{\beta}).
\end{equation*}
The finite-$(h,w)$ effects are controlled in Appendix~\ref{app:B}.
We first establish consistency and the root-scale expansion of the profile
direction, since all subsequent feasible-estimator arguments require control
of the generated index.

\begin{proposition}[Profile-SLS consistency and first-order expansion]\label{prop:profile-consistency}
Under Assumptions~\ref{ass:A}--\ref{ass:E}, the following hold for the
centred smooth profile criterion.
\emph{(i) Consistency.} Any exact global minimizer satisfies
\[
 \widehat{\boldsymbol{\beta}}_k\ \xrightarrow{p}\ \boldsymbol{\beta}_{k,0}.
\]
The same conclusion holds for a numerical solution satisfying the stated
objective-gap or projected-score condition; see \citet{PakesPollard1989} and
\citet{newey1994}.

\emph{(ii) First-order expansion.}
Under Assumptions~\ref{ass:A}--\ref{ass:E}, let
$\epsilon_k=H_k-m_k(S_{k0})$ and
\[
 \boldsymbol{\zeta}_{2k}(O_{2k})=E\{\epsilon_k\mathbf{d}_k(\mathbf{V}_k)\mid O_{2k}\},\qquad
 \boldsymbol{\zeta}_{1k}(O_{1k})=E\{\epsilon_k\mathbf{d}_k(\mathbf{V}_k)\mid O_{1k}\}.
\]
Then
\begin{equation}\label{eq:beta-expansion}
 \widehat{\boldsymbol{\beta}}_k-\boldsymbol{\beta}_{k,0}
 =\mathbf{R}_k\mathbf{H}_k^{-1}\left\{
 \frac1{n_{2k}}\sum_{i=1}^{n_{2k}}\boldsymbol{\zeta}_{2k}(O_{2ki})
 +\frac1{n_{1k}}\sum_{j=1}^{n_{1k}}\boldsymbol{\zeta}_{1k}(O_{1kj})
 \right\}+o_p(n_{L,k}^{-1/2}).
\end{equation}
Hence the tangent coordinates are asymptotically normal, with Hessian limit
$2\mathbf{H}_k$. The first-order terms are invariant to centring because
\[
 H_{c,k}-\{m_k(S_{k0})-c_{H,k}\}=H_k-m_k(S_{k0}).
\]
\end{proposition}

\begin{remark}[Effective sample size and zero coordinates]\label{rem:sparse-chart}
Although the profile criterion uses $n_{2k}n_{1k}$ cross-arm pairs, the
expansion in \eqref{eq:beta-expansion} is governed by the original-arm
harmonic scale $n_{L,k}^{-1/2}$, reflecting repeated use of the same
observations across pairs. No sparsity condition is imposed on $\boldsymbol{\beta}_{k,0}$;
zero coordinates remain admissible tangent directions under the local
sign-anchored parametrisation.
\end{remark}

With the profile direction controlled, consistency of the hard-floor group
estimator follows under Assumptions~\ref{ass:A}--\ref{ass:E}, before the
stronger target-transfer conditions needed for root-scale inference are
imposed.

\begin{theorem}[Consistency of the group and aggregate estimators]\label{thm:consistency}
Under Assumptions~\ref{ass:A}--\ref{ass:E},
\[
 \widehat\Delta_k\xrightarrow{p}\Delta_k
\]
for every fixed group. For fixed $K$ and deterministic bounded weights,
\[
 \widehat\Delta_{\mathrm{Ave}}\xrightarrow{p}\Delta_{\mathrm{Ave}},
 \qquad
 \widehat\Delta_{\mathrm{W}}\xrightarrow{p}\Delta_{\mathrm{W}}.
\]
\end{theorem}

\subsection{Group-Level Asymptotic Theory}

Write $S_k$ for the true-index score in the following population
functionals; the covariate-only notation $S_k^X$ is used for their outer
projections.
For the true-index score write
\[
 d_{k,h}(s)=E\{K_{h_k}(S_k-s)\},\qquad
 g_{k,h}(s)=E\{K_{h_k}(S_k-s)H_k\},\qquad
 m_{k,h}=g_{k,h}/d_{k,h}.
\]
Let $\Delta^c_{k,h,w}$ be the centred hard-floor population functional
\begin{equation*}
 \Delta^c_{k,h,w}
 =E\left\{c_{H,k}+\frac{g_{k,h}(S_k)-c_{H,k}d_{k,h}(S_k)}
 {d_{k,h}(S_k)\vee w_k}\right\}.
\end{equation*}

\begin{assumption}[Target-transfer smoothness]\label{ass:D}
For $p_{0,k}\ge4$, the zero extension of $f_k$ is absolutely continuous,
$f_k(a_k)=f_k(b_k)=0$, and $f_k$ is continuously differentiable on
$(a_k,b_k)$ with
\[
 |f_k'(a_k+r)|+|f_k'(b_k-r)|\le C r^{p_{0,k}-2}
\]
for small $r>0$. The function $m_k$ has a $C^2([a_k,b_k])$ version with
one-sided endpoint derivatives and
$\|m_k'\|_\infty+\|m_k''\|_\infty<\infty$.
\end{assumption}

The endpoint mass bound and Gaussian smoothing yield
\begin{equation*}
 |\Delta^c_{k,h,w}-\Delta_k|=O(w_k^{\gamma_{0,k}}+h_k^2)
 =o(\nu_k^{-1/2}).
\end{equation*}
Partial-mean and inverse-density integration arguments
\citep{Newey1994PartialMeans,DelyonPortier2016} inform the expansion,
together with generated-covariate adaptivity
\citep{MammenRotheSchienle2012,MammenRotheSchienle2016,Song2014}. The proof
has three layers. The low-density and two-sample bounds in Appendix~\ref{app:B}
linearise the centred hard-floor ratio and control nonlinear and
fitting--evaluation overlap remainders at the true index. Proposition~\ref{prop:oracle-decomp}
organizes the leading fluctuation into four original-observation projections,
two from the full covariate laws and two from labelled residual variation.
Index- and score-scale stability shows that replacing the true direction and
population score scale by their fitted counterparts is negligible at the root
scale, while target transfer replaces the finite-$(h,w)$ reference functional by the
scientific target. Regrouping by original unit yields \eqref{eq:group-var},
the CLT, and the supervised--semi-supervised variance comparison.

The limiting outer projections are
\[
 \phi_{X2,k}(\mathbf{x}_2)=E\{m_k(S_k^X(\mathbf{x}_2,\mathbf{X}_{1k}))\}-\Delta_k,
\]
\[
 \phi_{X1,k}(\mathbf{x}_1)=E\{m_k(S_k^X(\mathbf{X}_{2k},\mathbf{x}_1))\}-\Delta_k.
\]
With $\epsilon_k=H_k-m_k(S_k)$, the labelled residual projections are
\[
 \phi_{Y2,k}(o_2)=E(\epsilon_k\mid O_{2k}=o_2),\qquad
 \phi_{Y1,k}(o_1)=E(\epsilon_k\mid O_{1k}=o_1).
\]
Centring leaves both projection pairs unchanged because
$H_c-(m-c_H)=H-m$ and
$(m-c_H)-\{E(H)-c_H\}=m-E(H)$.

For the fixed-scale decomposition below, let
$\widehat\Delta_{k,h,w}^{\mathrm{or},c,\mathrm{ref}}$ denote the
population-scale reference statistic evaluated at $\boldsymbol{\beta}_{k,0}$ in standardized score
coordinates with fixed finite-$(h,w)$. Unlike the oracle estimator in
Section~\ref{sec:methodology}, equation~\eqref{eq:oracle-estimator}, it uses
population rather than candidate-dependent sample score scaling.

\begin{proposition}[Exact four-role oracle decomposition]\label{prop:oracle-decomp}
Let $A_{c,k}=\widehat q_{k,c}-q_{k,c}$,
$B_k=\widehat d_k-d_k$, $T_w(d)=d\vee w$, and
$J_{k,h,w}(s)=\mathbf 1\{d_{k,h}(s)>w_k\}$. The exact centred ratio identity is
\begin{equation*}
 \frac{\widehat q_{k,c}}{T_w(\widehat d_k)}
 -\frac{q_{k,c}}{T_w(d_k)}
 =\frac{A_{c,k}}{T_w(d_k)}
 -J_{k,h,w}\frac{q_{k,c}B_k}{d_k^2}
 +R_{k,h,w}^{ND,c},
\end{equation*}
where the three-region formula for $R_{k,h,w}^{ND,c}$ is given in
Appendix~\ref{sec:ld-ratio-decomp}. After outer averaging and the two-arm
Hoeffding decomposition,
\begin{align}\label{eq:exact-four-role}
 \widehat\Delta_{k,h,w}^{\mathrm{or},c,\mathrm{ref}}
 -\Delta^c_{k,h,w}
  ={}&M_{2k}^{-1}\sum_{a=1}^{M_{2k}}\phi_{X2,k,h,w}(\mathbf{X}_{2ka})
 +M_{1k}^{-1}\sum_{b=1}^{M_{1k}}\phi_{X1,k,h,w}(\mathbf{X}_{1kb})\notag\\
 &+n_{2k}^{-1}\sum_{i=1}^{n_{2k}}\phi_{Y2,k,h,w}(O_{2ki})
 +n_{1k}^{-1}\sum_{j=1}^{n_{1k}}\phi_{Y1,k,h,w}(O_{1kj})\notag\\
 &+R_{X,k}^d+R_{Y,k}^d+R_{4,k}^{\mathrm{lin}}
 +R_{4,k}^{ND,c}.
\end{align}
 Centring changes only the numerator mark; the denominator geometry and
overlap orders are unchanged.
\end{proposition}

\begin{assumption}[Variance regularity]\label{ass:F}
The variance of the four-role original-unit leading term is bounded above and below by
constant multiples of $\nu_k^{-1}$, and its summands satisfy the uniformly
bounded third-moment condition required for the group Lyapunov CLT.
\end{assumption}

\begin{theorem}[Asymptotic linearity, normality, and variance comparison]\label{thm:oracle-adapt-practical}
Under Assumptions~\ref{ass:A}--\ref{ass:D}, the following hold for every
fixed group in parts (i)--(iii). Parts (iv)--(v) additionally use the
conditions in Assumption~\ref{ass:F}.

\emph{(i) Oracle asymptotic linear representation (ALR).}
The oracle estimator admits the four-role original-observation expansion:
\begin{align}\label{eq:oracle-alr}
 \widehat\Delta_k^{\mathrm{or},c}-\Delta_k
 ={}&M_{2k}^{-1}\sum_{a=1}^{M_{2k}}\phi_{X2,k}(\mathbf{X}_{2ka})
 +M_{1k}^{-1}\sum_{b=1}^{M_{1k}}\phi_{X1,k}(\mathbf{X}_{1kb})\notag\\
 &+n_{2k}^{-1}\sum_{i=1}^{n_{2k}}\phi_{Y2,k}(O_{2ki})
 +n_{1k}^{-1}\sum_{j=1}^{n_{1k}}\phi_{Y1,k}(O_{1kj})
 +o_p(\nu_k^{-1/2}).
\end{align}

\emph{(ii) Empirical beta adaptivity.}
Estimating the single-index direction has no first-order effect on the group estimator:
\begin{equation}\label{eq:beta-adaptivity}
 \widehat\Delta_k^c(\widehat{\boldsymbol{\beta}}_k)
 -\widehat\Delta_k^c(\boldsymbol{\beta}_{k,0})
 =o_p(\nu_k^{-1/2}).
\end{equation}

\emph{(iii) Feasible scientific-target ALR.}
Consequently, the feasible estimator inherits the same first-order representation at the scientific target:
\begin{align}\label{eq:practical-alr}
 \widehat\Delta_k^c-\Delta_k
 ={}&M_{2k}^{-1}\sum_{a=1}^{M_{2k}}\phi_{X2,k}(\mathbf{X}_{2ka})
 +M_{1k}^{-1}\sum_{b=1}^{M_{1k}}\phi_{X1,k}(\mathbf{X}_{1kb})\notag\\
 &+n_{2k}^{-1}\sum_{i=1}^{n_{2k}}\phi_{Y2,k}(O_{2ki})
 +n_{1k}^{-1}\sum_{j=1}^{n_{1k}}\phi_{Y1,k}(O_{1kj})
 +o_p(\nu_k^{-1/2}).
\end{align}
Centring leaves the first-order influence representation and rate unchanged.
\emph{(iv) Group-level normality and variance.} The asymptotic variance is
\begin{equation}\label{eq:group-var}
 \sigma_{\mathrm{SSL},k}^2
 =\sum_{r=1}^2\left\{
 \frac{\Var(\phi_{Xr,k})}{M_{rk}}
 +\frac{\Var(\phi_{Yr,k})}{n_{rk}}
 +\frac{2}{M_{rk}}\Cov(\phi_{Xr,k},\phi_{Yr,k})
 \right\}.
\end{equation}
Under the correct model,
\[
 E\{\phi_{Yr,k}(O_{rk})\mid \mathbf{X}_{rk}\}=0,
 \qquad \Cov(\phi_{Xr,k},\phi_{Yr,k})=0,
\]
and
\[
 \frac{\widehat\Delta_k^c-\Delta_k}
 {\sigma_{\mathrm{SSL},k}}\Rightarrow N(0,1).
\]
\emph{(v) Supervised/SSL variance comparison.} For the supervised pairwise
estimator,
\begin{equation}\label{eq:group-variance-gain}
 \sigma_{\mathrm{SUP},k}^2-\sigma_{\mathrm{SSL},k}^2
 =\sum_{r=1}^2\left(\frac1{n_{rk}}-\frac1{M_{rk}}\right)
 \Var\{\phi_{Xr,k}(\mathbf{X}_{rk})\}\ge0.
\end{equation}
\end{theorem}

\begin{remark}[Interpretation of the four roles]
The $X$-projections describe fluctuation in the two full-covariate empirical
distributions, whereas the $Y$-projections describe labelled residual
fluctuation. Because each labelled unit contributes to both roles, variance
must be formed after regrouping contributions at the original-observation
level.
\end{remark}

To express the leading term in terms of independent original observations,
define, for $a\le n_{rk}$,
\[
 \eta_{rka,k}=M_{rk}^{-1}\phi_{Xr,k}(\mathbf{X}_{rka})
              +n_{rk}^{-1}\phi_{Yr,k}(O_{rka}),
\]
and for $a>n_{rk}$ set
$\eta_{rka,k}=M_{rk}^{-1}\phi_{Xr,k}(\mathbf{X}_{rka})$.
Let $L_{n,k}=\sum_{r=1}^2\sum_{a=1}^{M_{rk}}\eta_{rka,k}$.

\begin{remark}[Design overlap and labelled fractions]
The coefficient $2/M_{rk}$ in \eqref{eq:group-var} reflects the overlap of
the first $n_{rk}$ labelled units between the full-covariate and labelled
empirical averages. Additional covariates reduce the outer $\phi_X$ component
but not the labelled-residual $\phi_Y$ component, so the normalizing scale is
determined by original observations rather than by the $n_{1k}n_{2k}$ or
$M_{1k}M_{2k}$ pair counts.
\end{remark}

\subsection{\texorpdfstring{Fixed-$K$ Aggregation and Inference}{Fixed-K Aggregation and Inference}}\label{sec:aggregation}

For deterministic $a_{k,n}$ define
\[
 \Delta_{a,n}=\sum_{k=1}^Ka_{k,n}\Delta_k,\qquad
 \widehat\Delta_{a,n}^c=\sum_{k=1}^Ka_{k,n}\widehat\Delta_k^c,
\]
and
\[
 \sigma_{\mathrm{SSL},a,n}^2
 =\sum_{k=1}^Ka_{k,n}^2\sigma_{\mathrm{SSL},k}^2.
\]

Define the supervised aggregate variance by
\[
 \sigma_{\mathrm{SUP},a,n}^2
 =\sum_{k=1}^Ka_{k,n}^2\sigma_{\mathrm{SUP},k}^2.
\]

Here $K$ is fixed and the deterministic weights are uniformly bounded and not
all zero. For proportional unequal allocation, indexed by $m$, suppose that
for every arm and group,
\[
 \frac{n_{rk,m}}{m}\to a_{rk},\qquad
 \frac{M_{rk,m}}{m}\to b_{rk},\qquad 0<a_{rk},b_{rk}<\infty.
\]
Then $n_{L,k,m}$, $n_{\min,k,m}$, and $\nu_{k,m}$ are all of order $m$.

\begin{corollary}[Fixed-\(K\) joint and aggregate inference]\label{cor:aggregate-clt}
Under Assumptions~\ref{ass:A}--\ref{ass:F}:

\emph{(i)} the standardized group vector converges jointly to
$N_K(0,\mathbf{I}_K)$;

\emph{(ii)} For deterministic weights $a_{k,n}$, the aggregate estimator
inherits the group-level asymptotic linear representation:
\[
 \widehat\Delta_{a,n}^c-\Delta_{a,n}
 =\sum_{k=1}^Ka_{k,n}L_{n,k}
  +o_p(\sigma_{\mathrm{SSL},a,n}),
\]
Consequently,
\[
 \frac{\widehat\Delta_{a,n}^c-\Delta_{a,n}}
 {\sigma_{\mathrm{SSL},a,n}}
 \Rightarrow N(0,1);
\]

\emph{(iii)} Under the stated correct-specification condition, the groupwise
variance gains aggregate with the squared deterministic weights:
\[
 \sigma_{\mathrm{SUP},a,n}^2-\sigma_{\mathrm{SSL},a,n}^2
 =\sum_{k=1}^K a_{k,n}^2
 \left(\sigma_{\mathrm{SUP},k}^2-\sigma_{\mathrm{SSL},k}^2\right)\ge0.
\]
\end{corollary}

The equal-group average and prespecified deterministic weighted targets are
special cases. The same conclusions hold under the proportional unequal-allocation
regime above, with the corresponding group-specific scales.

\section{Perturbation-Based Inference}\label{sec:perturbation}

This section develops feasible inference by perturbing original observations,
preserving labelled--full-covariate nesting, and fully refitting the estimator
rather than resampling cross-arm pairs. The construction relates to minimand
perturbation and exchangeably weighted bootstrap methods
\citep{jin2001,PraestgaardWellner1993,ChengHuang2010}, and to related
semi-supervised comparison procedures \citep{tan2025,zhangpeng2025}. The
perturbation proof follows the four-role structure: conditional product and
hard-floor expansions identify the perturbed first-order projections; the
profile and index bounds show that refitting contributes only at higher order;
and reuse of each labelled
multinomial count in fitting and the labelled outer component preserves the
labelled outer--residual covariance. Replacing feasible projections by their
population limits yields \eqref{eq:ch4-leading-term}, the Gaussian
approximation, and ideal variance consistency.

Throughout this section, $\mathcal D$ denotes the observed data, $E^*$,
$\Var^*$, and $\mathcal L^*(\cdot\mid\mathcal D)$ denote
expectation, variance, and law under the conditional perturbation distribution,
and $\delta_x$ denotes a unit point mass at $x$. The distance $d_{BL}$ is the
bounded-Lipschitz metric.

\begin{assumption}[Perturbation regularity]\label{ass:perturbation}
The VC-type/entropy, deleted-array, random-evaluation, and two-sample
canonical-process regularity conditions from Section~\ref{sec:theory} admit
conditional multinomial analogues uniformly over the local classes used in
Appendix~\ref{app:C}. Let $p_B$ denote the available conditional
moment/envelope exponent and assume
\[ 
p_B(1/2-\eta)>1.
\]
\end{assumption}

Assumption~\ref{ass:perturbation} is the conditional counterpart of the process
regularity used in Section~\ref{sec:theory}. It ensures that the derivative,
deletion, overlap, and floor-crossing controls remain valid under multinomial
weighting, while the condition $p_B(1/2-\eta)>1$ supplies the tail control
needed to make the conditional hard-floor remainder negligible at the root
scale.

The corresponding conditional bounds are stated in Appendix~\ref{app:C}.

Numerical minimization uses the same objective-gap or projected-score condition
and local-basin requirement as in Section~\ref{sec:theory}.

\subsection{Multinomial Perturbation and Full Refitting}\label{subsec:perturbation-motivation}

For group \(k\) and arm \(r\), generate independently
\[
 (W_{rk1}^{L*},\ldots,W_{rkn_{rk}}^{L*})
 \sim\operatorname{Multinomial}
 \left(n_{rk};n_{rk}^{-1},\ldots,n_{rk}^{-1}\right)
\]
for the labelled observations, and
\[
 (W_{rk1}^{U*},\ldots,W_{rkN_{rk}}^{U*})
 \sim\operatorname{Multinomial}
 \left(N_{rk};N_{rk}^{-1},\ldots,N_{rk}^{-1}\right)
\]
for the unlabelled covariates. These four arm--stratum vectors are
conditionally independent. Write
\(\xi_{rki}^{L*}=W_{rki}^{L*}-1\) and
\(\xi_{rkj}^{U*}=W_{rkj}^{U*}-1\).

The same labelled count \(W_{rki}^{L*}\) is used in both parts of the
replicate: (i) the labelled nuisance, profile, and direction fit, and
(ii) the labelled component of the full outer covariate empirical measure,
\[
 Q_{rk}^*=
 \frac{1}{M_{rk}}\sum_{i=1}^{n_{rk}}W_{rki}^{L*}\delta_{\mathbf{X}_{rki}}
 +\frac{1}{M_{rk}}\sum_{j=1}^{N_{rk}}W_{rkj}^{U*}\delta_{\mathbf{X}_{rkj}^U}.
\]
Resampling is therefore performed at the original-observation level rather than
at the cross-arm-pair level. Each replicate fully refits the score
standardization, profile/index estimation, hard-floor conditional fit, and full
outer average. The target, centred stabilization, and prescribed tuning
sequences are unchanged.

Let \(\widehat\Delta_k^{*(b)}\) denote the resulting full-refit estimate in
replicate \(b\), and define
\[
 \overline{\widehat\Delta_k^*}
 =\frac{1}{B}\sum_{b=1}^B\widehat\Delta_k^{*(b)},\qquad
 \widehat V_{k,B}
 =\frac{1}{B-1}\sum_{b=1}^B
 \left(\widehat\Delta_k^{*(b)}
       -\overline{\widehat\Delta_k^*}\right)^2.
\]
\subsection{Conditional Validity}\label{subsec:perturbation-variance}

Let \(\sigma_{k,n}^2\) be the leading sampling variance in
\eqref{eq:group-var}. Theorem~\ref{thm:ch4-validity} uses the conditions in
Section~\ref{sec:theory} and Assumption~\ref{ass:perturbation}. Conditional
weak convergence is stated in the bounded-Lipschitz metric.

\begin{theorem}[Multinomial perturbation validity and ideal variance consistency]\label{thm:ch4-validity}
For a fixed group \(k\), under the conditions of Section~\ref{sec:theory},
Assumption~\ref{ass:perturbation}, and the numerical-minimization condition of
Section~\ref{subsec:perturbation-motivation} when applicable, the following hold.
\emph{(i) Conditional asymptotic linear representation.}
Conditionally on the observed data, the full-refit perturbation admits
\[
 \widehat\Delta_k^*-\widehat\Delta_k
 =L_{k,n}^*+R_{k,n}^*,
 \qquad
 E^*\{\nu_k(R_{k,n}^*)^2\mid\mathcal D\}\ \xrightarrow{p}\ 0,
\]
where
\begin{align}
 L_{k,n}^*=\sum_{r=1}^2\Bigg[&
 \sum_{i=1}^{n_{rk}}\xi_{rki}^{L*}
 \left\{
 \frac{\phi_{Xr,k}(\mathbf{X}_{rki})}{M_{rk}}
 +\frac{\phi_{Yr,k}(O_{rki})}{n_{rk}}
 \right\}\notag\\
 &+\sum_{j=1}^{N_{rk}}\xi_{rkj}^{U*}
 \frac{\phi_{Xr,k}(\mathbf{X}_{rkj}^U)}{M_{rk}}
 \Bigg].\label{eq:ch4-leading-term}
\end{align}
\emph{(ii) Conditional Gaussian approximation.}
After normalization by the leading sampling standard deviation, its conditional law is asymptotically standard normal:
\[
 d_{BL}\!\left\{
 \mathcal L^*\!\left(
 \frac{\widehat\Delta_k^*-\widehat\Delta_k}{\sigma_{k,n}}
 \,\middle|\,\mathcal D\right),\,N(0,1)\right\}\xrightarrow{p}0.
\]
\emph{(iii) Conditional variance consistency.}
The ideal conditional perturbation variance consistently estimates the leading sampling variance:
\[
 V_{k,n}^*=\Var^*(\widehat\Delta_k^*\mid\mathcal D),
\]
\[
 \frac{V_{k,n}^*}{\sigma_{k,n}^2}\xrightarrow{p}1,
\qquad
 \sigma_{k,n}^2
 =\sum_{r=1}^2\left\{
 \frac{\Var(\phi_{Xr,k})}{M_{rk}}
 +\frac{\Var(\phi_{Yr,k})}{n_{rk}}
 +\frac{2\Cov(\phi_{Xr,k},\phi_{Yr,k})}{M_{rk}}
 \right\}.
\]
\end{theorem}

\begin{remark}[Why observation-level full refitting is required]
Theorem~\ref{thm:ch4-validity} shows that perturbation validity depends on reproducing the original-observation dependence structure, not merely on resampling the $n_{2k}n_{1k}$ labelled pairs. Because a labelled unit enters both conditional fitting and the outer covariate empirical measure, its perturbation count must be reused in both roles; this reuse generates the covariance term in \(\sigma_{k,n}^2\). Full refitting additionally propagates score standardization and index estimation through each replicate. Under the correct-model orthogonality of Section~\ref{sec:theory}, the overlap covariance vanishes asymptotically, but preserving it in the resampling construction is necessary before that simplification is invoked.
\end{remark}

\subsection{Finite-Replicate Variance Estimation and Wald Inference}\label{subsec:perturbation-validity}

Theorem~\ref{thm:ch4-validity} concerns the ideal conditional perturbation
variance. Consistency of its finite-\(B_n\) estimate additionally requires the
normalized fourth-moment condition below \citep{Cheng2015}.

\begin{corollary}[Finite-\(B_n\) variance and Wald inference]\label{cor:ch4-finiteB-wald}
Suppose the conditions of Theorem~\ref{thm:ch4-validity} hold. If
\(B_n\to\infty\) and, with
\[
 Z_{k,n}^*=\widehat\Delta_k^*-E^*(\widehat\Delta_k^*\mid\mathcal D),
 \qquad
 V_{k,n}^*=E^*\{(Z_{k,n}^*)^2\mid\mathcal D\},
\]
the conditional fourth-moment condition
\[
 \frac{E^*\{(Z_{k,n}^*)^4\mid\mathcal D\}}{(V_{k,n}^*)^2}=O_p(1)
\]
holds, then
the finite-replicate variance consistently estimates both the ideal perturbation variance and the leading sampling variance:
\[
 \frac{\widehat V_{k,B_n}}{V_{k,n}^*}\xrightarrow{p}1,
 \qquad
 \frac{\widehat V_{k,B_n}}{\sigma_{k,n}^2}\xrightarrow{p}1.
\]
The same conclusion holds in the iterated limit that first sends
\(n\to\infty\) and then \(B\to\infty\). Hence, with
\(\widehat{SE}_{k,B_n}=\sqrt{\widehat V_{k,B_n}}\),
\[
 CI_{k,1-\alpha}=
 \left[
 \widehat\Delta_k-z_{1-\alpha/2}\widehat{SE}_{k,B_n},\
 \widehat\Delta_k+z_{1-\alpha/2}\widehat{SE}_{k,B_n}
 \right]
\]
satisfies
\[
 P\{\Delta_k\in CI_{k,1-\alpha}\}\longrightarrow1-\alpha.
\]
\end{corollary}

\subsection{\texorpdfstring{Fixed-$K$ Aggregate Inference}{Fixed-K Aggregate Inference}}\label{subsec:perturbation-grouped}

For the deterministic fixed-\(K\) weights of Section~\ref{sec:aggregation},
use the aggregate estimator defined there and set
\[
 \widehat V_{a,B_n}
 =\sum_{k=1}^K a_{k,n}^2\widehat V_{k,B_n}.
\]

\begin{corollary}[Fixed-\(K\) aggregate perturbation inference]\label{cor:ch4-aggregation}
Under independent groups and the conditions of
Corollary~\ref{cor:ch4-finiteB-wald} for each group,
the weighted perturbation variance is consistent for the aggregate sampling variance:
\[
\frac{\widehat V_{a,B_n}}{\sigma_{\mathrm{SSL},a,n}^2}
\xrightarrow{p}1.
\]
Consequently, the corresponding aggregate Wald interval, with \(\widehat{SE}_{a,B_n}=\sqrt{\widehat V_{a,B_n}}\), is
\[
 CI_{a,1-\alpha}=
 \left[
 \widehat\Delta_{a,n}^c-z_{1-\alpha/2}\widehat{SE}_{a,B_n},\
 \widehat\Delta_{a,n}^c+z_{1-\alpha/2}\widehat{SE}_{a,B_n}
 \right].
\]
It has asymptotic coverage \(1-\alpha\):
\[
 P\{\Delta_{a,n}\in CI_{a,1-\alpha}\}\longrightarrow1-\alpha.
\]
\end{corollary}

The equal-group average and any prespecified deterministic weighted target are
obtained as special cases.

\section{Simulation Studies}\label{sec:simulation}
\subsection{Simulation Designs}\label{subsec:simulation-system}
We considered three groups and two independent arms within each group. For each arm and group, the four covariates were generated independently from the coordinatewise truncated-normal law $\mathrm{TN}(0,1;[-2.5,2.5])$; the first two coordinates were active and the last two were noise. With
\[
 \boldsymbol{\gamma}_1=2^{-1/2}(1,-1,0,0)^\top,\qquad
 \boldsymbol{\gamma}_2=2^{-1/2}(1,1,0,0)^\top,
\]
labelled outcomes followed
\[
 Y_{1k}=\boldsymbol{\gamma}_1^\top \mathbf{X}_{1k}+\varepsilon_{1k},\qquad
 Y_{2k}=\delta_k+\boldsymbol{\gamma}_2^\top \mathbf{X}_{2k}+\varepsilon_{2k},
\]
where the arm errors were independent standard Gaussian variables. Unlabelled units contributed covariates only.

Unless stated otherwise, $n$ and $N$ are the labelled and unlabelled sample sizes per arm and group. The regular probability design is the reference setting for isolating the role of unlabelled covariates. We keep the labelled size fixed and vary $N$, so the changes in RMSE and MSE reduction describe the contribution of the outer covariate-distribution average. It uses $\Delta_k=P(Y_{2k}>Y_{1k})$, with $\delta_k=0$, $n=100$, and $N\in\{100,200,500\}$; hence $N/n=1,2,$ and $5$, and $\Delta_k=0.5$.

The heterogeneous (HET) design changes the group-specific comparison probabilities while retaining the same estimation problem. It therefore examines whether the grouped procedure preserves distinct group targets before aggregation, rather than only performing well when all group truths coincide. We set $(\delta_1,\delta_2,\delta_3)=(-d_H,0,d_H)$, with $d_H\approx0.7553$ calibrated to yield probability truths $(0.35,0.50,0.65)$.

The CLIP design changes the bounded comparison response while keeping the regular probability data-generating process. It targets $E\{\operatorname{clip}(Y_{2k}-Y_{1k},-2,2)\}$, with clipping limits $(-2,2)$ and reference $c_H=0$, and thus assesses the estimator beyond the probabilistic comparison target. The unequal-allocation (UAL) design keeps this data-generating process but varies the information available by group, using labelled groupwise sample sizes $(60,100,140)$ and unlabelled groupwise sample sizes $(120,200,280)$ per arm for G1--G3.

With pair order $(\mathbf{X}_2^\top,\mathbf{X}_1^\top)^\top$, the corresponding true direction is
\[
 \boldsymbol{\beta}_0=(\boldsymbol{\gamma}_2^\top,-\boldsymbol{\gamma}_1^\top)^\top/\sqrt{2}
 =(1/2,1/2,0,0,-1/2,1/2,0,0)^\top.
\]
Here the true-index and conservative endpoint orders are $p_0=4$ and $p_*=8$. For the numerical studies, we use the admissible interior choice
\[
 (\alpha,\eta)=\left(\frac13,\frac{15}{32}\right),
\]
with
\[
 \bar h_k=n_{\min,k}^{-1/3},\qquad
 \bar w_k=K(0)n_{L,k}^{-15/32},\qquad
 K(0)=(2\pi)^{-1/2}.
\]
The numerical implementation uses the $C^3$ connector
\[
 r(x)=
 \begin{cases}
 1,&x\le1/2,\\
 p_7\{(x-1/2)/(3/2)\},&1/2<x<2,\\
 x,&x\ge2,
 \end{cases}
 \qquad
 p_7(u)=1+\frac{25}{2}u^4-\frac{51}{2}u^5+19u^6-5u^7.
\]
Appendix~\ref{subsec:app-rate-conditions} verifies admissibility for this geometry.

We report G1--G3, $\mathrm{Ave}=(\Delta_1+\Delta_2+\Delta_3)/3$, and $W=0.2\Delta_1+0.3\Delta_2+0.5\Delta_3$. SUP and SSL denote the supervised and semi-supervised estimators, respectively. Across $R=500$ independent Monte Carlo replications, we report SSL bias, empirical standard deviation (EmpSD), and RMSE, together with $\mathrm{RE}=\mathrm{RMSE}_{\mathrm{SUP}}/\mathrm{RMSE}_{\mathrm{SSL}}$ and
\[
 \text{MSE reduction (\%)}=100\left(1-\frac{\mathrm{MSE}_{\mathrm{SSL}}}{\mathrm{MSE}_{\mathrm{SUP}}}\right).
\]
Bias is the signed Monte Carlo mean error, EmpSD measures sampling variation across replications, and RMSE combines the two. Thus $\mathrm{RE}>1$ and a positive MSE reduction favour SSL relative to SUP. The aggregate rows are evaluated at their prespecified targets, so their performance reflects both groupwise estimation and the chosen aggregation rule.
\subsection{Point-Estimation Performance}\label{subsec:simulation-estimation}
The regular probability design isolates the role of additional covariate information. SSL bias remained small, and SSL RMSE was lower than SUP RMSE for every reported group and aggregate target. As $N/n$ increased from $1$ to $5$, MSE reductions rose from approximately $17$--$20\%$ to $28$--$36\%$ (Figure~\ref{fig:ch5-mse-reduction}); Table~\ref{tab:ch5-regular} gives the exact values. This pattern is the empirical counterpart of the variance decomposition: additional unlabelled covariates refine the outer covariate-distribution component, while labelled-sample uncertainty remains.
\begin{figure}[H]
\centering
\includegraphics[width=.97\linewidth]{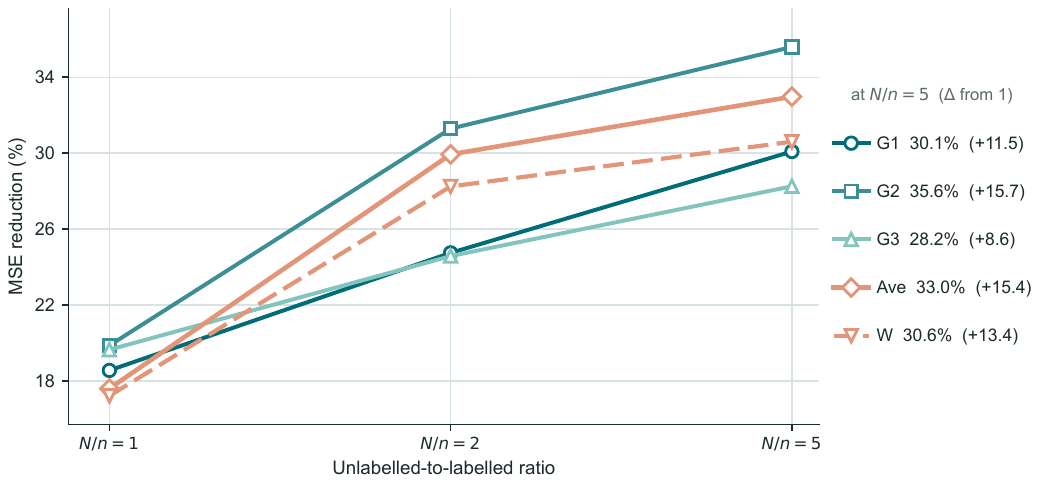}
\caption{MSE reduction of SSL relative to SUP as the unlabelled-to-labelled ratio $N/n$ increases in the regular probability setting ($n=100$ per arm and group). Endpoint labels give the reduction at $N/n=5$, with the increase from $N/n=1$ in parentheses.}
\label{fig:ch5-mse-reduction}
\end{figure}
Table~\ref{tab:ch5-regular} gives the numerical counterpart of Figure~\ref{fig:ch5-mse-reduction}. The small biases and the close alignment of EmpSD and RMSE show that the reported improvements are primarily precision differences rather than corrections of a large systematic error. The aggregate rows have smaller sampling variation than the individual group rows, while the distinction between Ave and W records the scientific choice of how group targets are combined.

The progression across ratios also shows how the gain accumulates. For every reported target, the increase in MSE reduction from $N/n=2$ to $N/n=5$ is smaller than that from $N/n=1$ to $N/n=2$, despite the larger absolute increase in $N$. This attenuation is consistent with increasingly precise estimation of the outer covariate distribution: once that component has been reduced, uncertainty from learning the conditional comparison with labelled outcomes remains.

The regular design uses the same data-generating mechanism and sample sizes for G1--G3. Their curves show a common monotone improvement, with modest differences across the Monte Carlo runs. Ave has slightly smaller RMSE than W, consistent with equal weighting of three independent, equally informative groups. W places half of its weight on G3 and consequently gives that group's sampling variation a larger contribution to the aggregate variance. The difference between Ave and W thus reflects their prespecified scientific weights.
\begin{table}[H]
\centering
\caption{Point-estimation performance under the regular probability setting at three unlabelled-to-labelled ratios.}
\label{tab:ch5-regular}
\small
\setlength{\tabcolsep}{5.2pt}
\begin{tabular}{@{}lrrrrr@{}}
\toprule
Unit & Bias & EmpSD & RMSE & RE & MSE reduction (\%) \\
\midrule
\multicolumn{6}{l}{\textit{Panel A: }$N/n=1$} \\
G1 & -0.0017 & 0.0377 & 0.0377 & 1.108 & 18.6 \\
G2 & -0.0007 & 0.0364 & 0.0363 & 1.117 & 19.9 \\
G3 & 0.0012 & 0.0366 & 0.0365 & 1.116 & 19.6 \\
Ave & -0.0004 & 0.0216 & 0.0216 & 1.102 & 17.6 \\
W & 0.0000 & 0.0231 & 0.0231 & 1.099 & 17.2 \\
\addlinespace[2pt]
\multicolumn{6}{l}{\textit{Panel B: }$N/n=2$} \\
G1 & 0.0015 & 0.0373 & 0.0373 & 1.153 & 24.7 \\
G2 & 0.0012 & 0.0338 & 0.0338 & 1.206 & 31.3 \\
G3 & 0.0009 & 0.0344 & 0.0344 & 1.151 & 24.6 \\
Ave & 0.0012 & 0.0205 & 0.0205 & 1.195 & 29.9 \\
W & 0.0011 & 0.0215 & 0.0215 & 1.180 & 28.2 \\
\addlinespace[2pt]
\multicolumn{6}{l}{\textit{Panel C: }$N/n=5$} \\
G1 & -0.0008 & 0.0341 & 0.0341 & 1.196 & 30.1 \\
G2 & 0.0023 & 0.0327 & 0.0328 & 1.246 & 35.6 \\
G3 & 0.0008 & 0.0345 & 0.0345 & 1.181 & 28.2 \\
Ave & 0.0008 & 0.0187 & 0.0187 & 1.221 & 33.0 \\
W & 0.0009 & 0.0202 & 0.0202 & 1.200 & 30.6 \\
\bottomrule
\end{tabular}
\vspace{0.35em}
\begin{minipage}{0.97\linewidth}
\footnotesize\textit{Note.} Results use $R=500$ Monte Carlo replications, with $n=100$ labelled observations per arm and group. Bias, EmpSD, and RMSE refer to SSL; RE $=\mathrm{RMSE}_{\mathrm{SUP}}/\mathrm{RMSE}_{\mathrm{SSL}}$. Ave uses equal group weights and W uses $(0.2,0.3,0.5)$.
\end{minipage}

\end{table}

The three nonbaseline settings examine different departures from the reference design. Figure~\ref{fig:ch5-nonbaseline} displays their common pattern, and Table~\ref{tab:ch5-nonbaseline} reports the corresponding numerical results under common accuracy and precision measures.

In HET, the group truths differ by construction, so the signed-error panels assess groupwise accuracy while the RMSE panels assess precision around different targets. The SSL errors remain small and SSL RMSE is lower than SUP RMSE across the group-specific and aggregate rows. The result shows that the precision gain does not depend on the groups sharing a common comparison probability.

CLIP evaluates a bounded contrast on a different numerical scale from a probability. Its RMSE values should therefore be compared between SUP and SSL within the CLIP setting rather than across settings. Within that comparison, the SSL estimator again has smaller RMSE for all reported targets, showing that the gain extends beyond the probabilistic comparison functional.

UAL makes the information imbalance explicit: G1, G2, and G3 receive $(n,N)=(60,120)$, $(100,200)$, and $(140,280)$, respectively. The groupwise precision follows these different information levels, while the Ave and W rows continue to estimate their prespecified aggregate targets. In particular, W is determined by the scientific weights $(0.2,0.3,0.5)$ rather than by the observed group sample sizes.
\begin{figure}[!t]
\centering
\includegraphics[width=.98\linewidth]{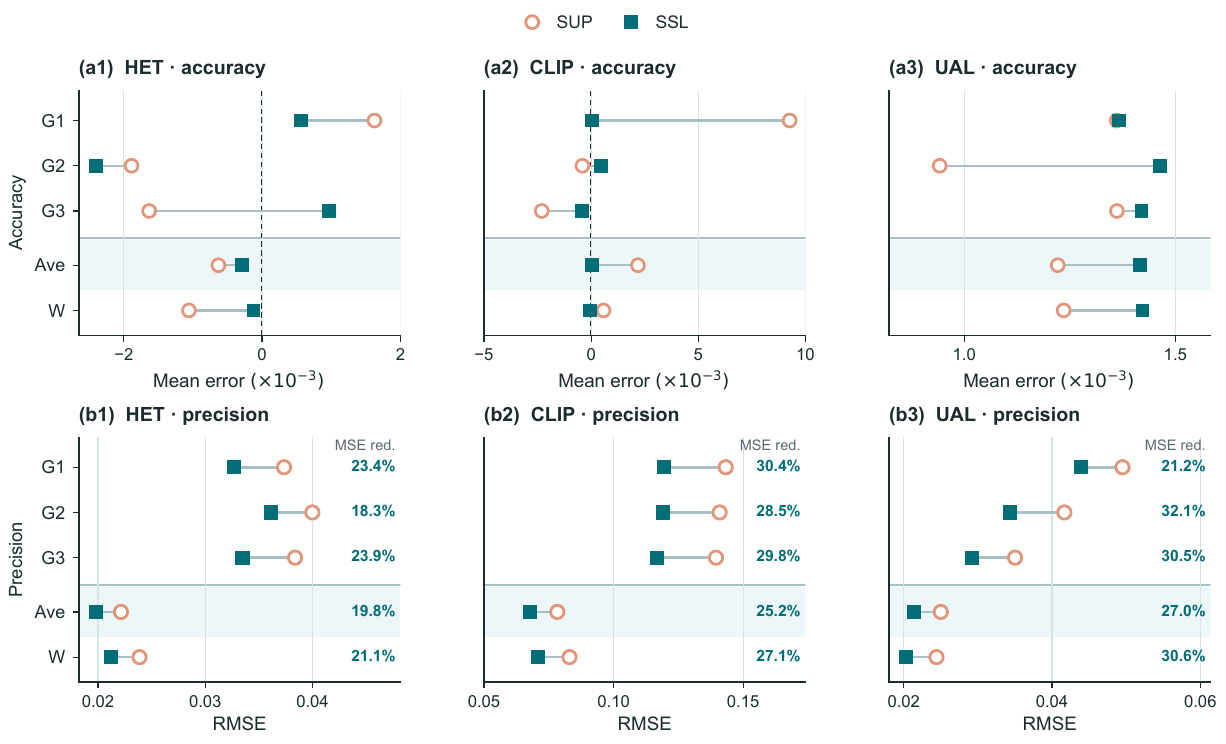}
\caption{Accuracy and precision of SUP and SSL under the heterogeneous probability (HET), clipped-contrast (CLIP), and unequal-allocation (UAL) settings. The top row shows signed Monte Carlo mean errors, and the bottom row compares RMSE; row-end labels report the corresponding SSL MSE reduction relative to SUP.}
\label{fig:ch5-nonbaseline}
\end{figure}
Table~\ref{tab:ch5-nonbaseline} reports the same bias, EmpSD, RMSE, and relative-efficiency summaries used in the regular design. Across its three panels, SSL combines small signed errors with lower RMSE for the group-specific and aggregate targets. The magnitude of the gain varies with target scale and group allocation, as reflected in the setting-specific reductions reported below.

In HET, the groupwise MSE reductions range from $18.3\%$ to $23.9\%$ across truths spanning $0.35$ to $0.65$. The improvement therefore occurs both at the central comparison probability and at the two asymmetric group targets. In CLIP, the corresponding reductions range from $28.5\%$ to $30.4\%$. Within each functional, these percentages express the precision gain relative to its supervised benchmark. Together with the small signed errors, they show that the gains extend across the distinct target values and response scales represented in these designs.

UAL links groupwise precision to the available information. G1, with the smallest labelled sample, has the largest SSL RMSE ($0.0439$), whereas G3, with the largest sample, has the smallest ($0.0292$). Under the prespecified weights, W assigns less weight to G1 and more to G3 than Ave, yielding RMSEs of $0.0203$ and $0.0214$, respectively. Here the scientific weights favour the more precisely estimated groups, reversing the Ave/W ordering in the regular design. The comparison illustrates how allocation and aggregation weights jointly determine aggregate precision.
\begin{table}[H]
\centering
\caption{Point-estimation performance across the heterogeneous probability (HET), clipped-contrast (CLIP), and unequal-allocation (UAL) settings.}
\label{tab:ch5-nonbaseline}
\small
\setlength{\tabcolsep}{3.4pt}
\begin{tabular}{@{}lrrrrrrrr@{}}
\toprule
Unit & Truth & $n/N$ & Bias & EmpSD & RMSE & RE & MSE reduction (\%) \\
\midrule
\multicolumn{8}{l}{\textit{Panel A: HET} ($n=100$, $N=200$)} \\
G1 & 0.350 & 100/200 & 0.0006 & 0.0327 & 0.0327 & 1.143 & 23.4 \\
G2 & 0.500 & 100/200 & -0.0024 & 0.0361 & 0.0361 & 1.106 & 18.3 \\
G3 & 0.650 & 100/200 & 0.0010 & 0.0335 & 0.0335 & 1.146 & 23.9 \\
Ave & 0.500 & 100/200 & -0.0003 & 0.0198 & 0.0198 & 1.116 & 19.8 \\
W & 0.545 & 100/200 & -0.0001 & 0.0212 & 0.0212 & 1.126 & 21.1 \\
\addlinespace[2pt]
\multicolumn{8}{l}{\textit{Panel B: CLIP} ($n=100$, $N=200$)} \\
G1 & 0 & 100/200 & 0.0001 & 0.1196 & 0.1195 & 1.199 & 30.4 \\
G2 & 0 & 100/200 & 0.0005 & 0.1193 & 0.1192 & 1.182 & 28.5 \\
G3 & 0 & 100/200 & -0.0004 & 0.1170 & 0.1169 & 1.194 & 29.8 \\
Ave & 0 & 100/200 & 0.0000 & 0.0678 & 0.0677 & 1.157 & 25.2 \\
W & 0 & 100/200 & -0.0001 & 0.0709 & 0.0709 & 1.171 & 27.1 \\
\addlinespace[2pt]
\multicolumn{8}{l}{\textit{Panel C: UAL} ($n/N=60/120$, $100/200$, and $140/280$ for G1--G3)} \\
G1 & 0.500 & 60/120 & 0.0014 & 0.0440 & 0.0439 & 1.126 & 21.2 \\
G2 & 0.500 & 100/200 & 0.0015 & 0.0343 & 0.0343 & 1.214 & 32.1 \\
G3 & 0.500 & 140/280 & 0.0014 & 0.0292 & 0.0292 & 1.200 & 30.5 \\
Ave & 0.500 & --- & 0.0014 & 0.0213 & 0.0214 & 1.171 & 27.0 \\
W & 0.500 & --- & 0.0014 & 0.0203 & 0.0203 & 1.200 & 30.6 \\
\bottomrule
\end{tabular}
\vspace{0.35em}
\begin{minipage}{0.97\linewidth}
\footnotesize\textit{Note.} Results use $R=500$ Monte Carlo replications. Bias, EmpSD, and RMSE refer to SSL; RE $=\mathrm{RMSE}_{\mathrm{SUP}}/\mathrm{RMSE}_{\mathrm{SSL}}$. In the UAL panel, $n/N$ gives the labelled/unlabelled sample sizes per arm. Ave uses equal group weights and W uses $(0.2,0.3,0.5)$.
\end{minipage}

\end{table}

\FloatBarrier

\subsection{Perturbation-Based Inference}\label{subsec:simulation-inference}
The perturbation experiment uses the regular probability setting with $n=50$ labelled and $N=100$ unlabelled observations per arm and group, $R=500$ Monte Carlo replications, and $B=200$ perturbation replicates per dataset. Each perturbation replicate recomputes the score, profile fit, hard-floor estimate, and full outer average, so the exercise evaluates the complete refitting procedure rather than a fixed-fit variance calculation. We report Monte Carlo bias, EmpSD, mean perturbation standard error, the SE/SD ratio, empirical coverage of nominal $95\%$ intervals, and mean interval length.

Monte Carlo bias is small relative to sampling variability. Mean perturbation standard errors closely track empirical sampling variability, with SE/SD ranging from $0.996$ to $1.045$ (Figure~\ref{fig:ch5-inference-summary} and Table~\ref{tab:ch5-inference}). This calibration is visible for both group-specific and aggregate targets and is consistent with the original-observation construction used in the perturbation theory.

The small biases and near-unity SE/SD ratios indicate accurate centring and variance calibration under complete refitting. The largest absolute bias is $0.0015$, compared with groupwise EmpSD values near $0.05$. For Ave and W, EmpSD decreases to $0.0298$ and $0.0314$, and the mean perturbation standard errors follow this reduction closely. The resulting aggregate intervals are shorter while maintaining coverage comparable to the groupwise intervals. This agreement across group and aggregate targets is consistent with the original-observation variance representation, in which the same labelled observations contribute to both fitting and outer averaging.
\begin{figure}[H]
\centering
\includegraphics[width=\textwidth]{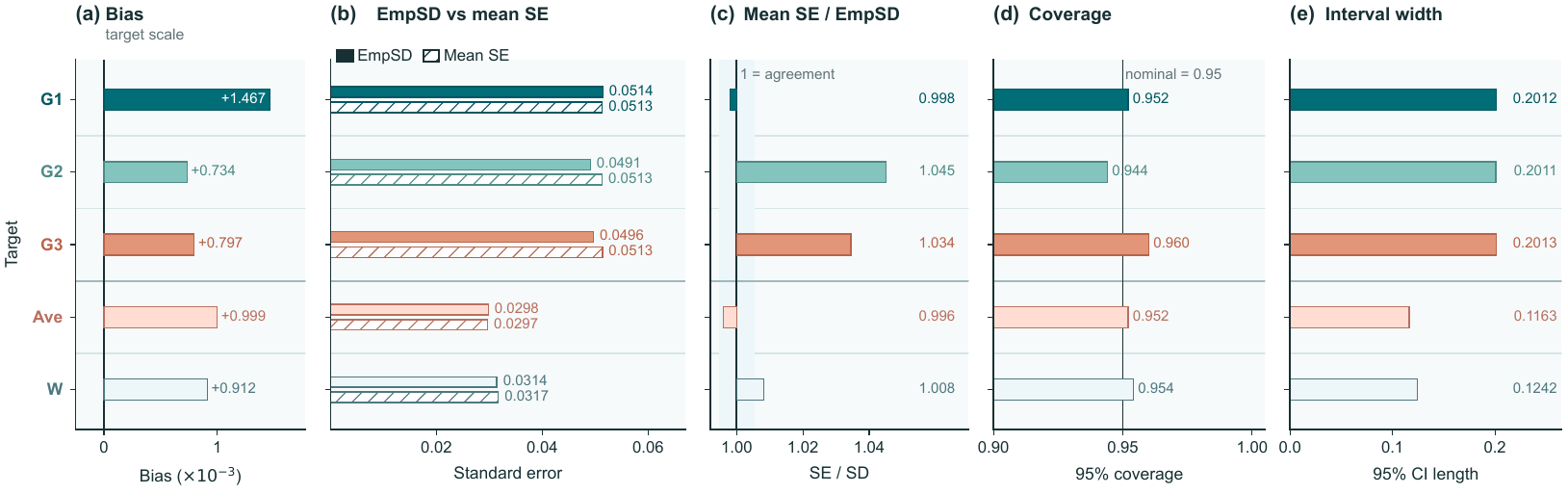}
\caption{Finite-sample performance of full-refit multinomial perturbation inference in the regular probability setting ($n=50$, $N=100$, $R=500$, $B=200$). Panel (a) shows Monte Carlo bias; (b) compares empirical standard deviation (EmpSD) with the mean perturbation standard error; (c) reports Mean SE/EmpSD; (d) reports empirical coverage of nominal 95\% Wald intervals; and (e) reports mean 95\% confidence-interval length. Reference lines in (c) and (d) mark 1 and 0.95, respectively.}
\label{fig:ch5-inference-summary}
\end{figure}
\begin{table}[H]
\centering
\caption{Finite-sample performance of multinomial perturbation inference in the regular probability setting.}
\label{tab:ch5-inference}
\small
\setlength{\tabcolsep}{5.0pt}
\begin{tabular}{@{}lrrrrrr@{}}
\toprule
Unit & Bias & EmpSD & Mean SE & SE/SD & CP & 95\% CI Len. \\
\midrule
G1 & 0.0015 & 0.0514 & 0.0513 & 0.998 & 0.952 & 0.2012 \\
G2 & 0.0007 & 0.0491 & 0.0513 & 1.045 & 0.944 & 0.2011 \\
G3 & 0.0008 & 0.0496 & 0.0513 & 1.034 & 0.960 & 0.2013 \\
Ave & 0.0010 & 0.0298 & 0.0297 & 0.996 & 0.952 & 0.1163 \\
W & 0.0009 & 0.0314 & 0.0317 & 1.008 & 0.954 & 0.1242 \\
\bottomrule
\end{tabular}
\vspace{0.35em}
\begin{minipage}{0.97\linewidth}
\footnotesize\textit{Note.} $n=50$ and $N=100$ per arm and group; $R=500$ and $B=200$. Mean SE is the mean perturbation standard error, CP is empirical coverage, and 95\% CI Len. is the mean confidence-interval length.
\end{minipage}

\end{table}

The empirical coverage ranges from $0.944$ to $0.960$ around the nominal $0.95$ level. Aggregate targets have shorter intervals, consistent with their smaller sampling variability, while their coverage remains close to the groupwise values. The regular-design experiment thus supports finite-sample calibration of the full-refit perturbation procedure for both group-specific and aggregate inference.
\FloatBarrier

\section{Real Data Application}\label{sec:realdata}
We use NHANES 2011--2018 to illustrate grouped semi-supervised comparison in a real covariate distribution. The analysis is descriptive rather than causal and concerns the analysed sample rather than survey-weighted national inference.

The outcome is serum creatinine (LBXSCR, mg/dL). Arm 1 comprises adults without a reported history of clinician-diagnosed hypertension (BPQ020=2), and arm 2 those with such a history (BPQ020=1). We consider $\Delta_k=P(Y_{2k}>Y_{1k})$ within G1 (20--39 years), G2 (40--59 years), and G3 (60+ years), using age, sex, body mass index, and income-to-poverty ratio as covariates.

Within the complete-case analysis set, controlled masking retained $n=50$ labelled observations per arm and treated outcomes for an additional $N=100$ observations per arm as unobserved. This masking design defines the semi-supervised observation pattern. Ave is the equal-group average, whereas W uses the prespecified weights $(0.2,0.3,0.5)$; W is unrelated to the NHANES survey weights. SUP standard errors use the corresponding two-sample U-statistic H\'{a}jek plug-in variance estimator, whereas SSL standard errors use full-refit multinomial perturbation with $B=200$.

SUP and SSL point estimates are similar for all five quantities (Table~\ref{tab:nhanes-creatinine-hypertension-efficiency} and Figure~\ref{fig:nhanes-forest-precision}), while SSL has the smaller standard error in every comparison. The SUP-to-SSL standard-error ratios range from $1.08$ to $1.12$. Consistent with the variance mechanism in Section~\ref{sec:theory}, additional covariate information improves estimation of the outer covariate-distribution component without changing the target. All five SSL $95\%$ confidence intervals include $0.5$.
\begin{table}[!ht]
\centering
\small
\setlength{\tabcolsep}{6pt}
\renewcommand{\arraystretch}{1.20}
\caption{Supervised and semi-supervised estimates of the serum-creatinine comparison probability in the NHANES 2011--2018 illustration.}
\label{tab:nhanes-creatinine-hypertension-efficiency}
\resizebox{\linewidth}{!}{%
\begin{tabular}{@{}lccccc@{}}
\toprule
Unit & SUP estimate (95\% CI) & SSL estimate (95\% CI) &
$SE_{\mathrm{SUP}}$ & $SE_{\mathrm{SSL}}$ & $SE_{\mathrm{SUP}}/SE_{\mathrm{SSL}}$ \\
\midrule
G1 (20--39) &
0.483 [0.368, 0.598] &
0.483 [0.377, 0.589] &
0.0586 & 0.0542 & \textbf{1.08} \\
G2 (40--59) &
0.586 [0.474, 0.698] &
0.593 [0.493, 0.693] &
0.0573 & 0.0511 & \textbf{1.12} \\
G3 (60+) &
0.503 [0.388, 0.618] &
0.515 [0.409, 0.620] &
0.0586 & 0.0538 & \textbf{1.09} \\
\addlinespace
\textbf{Ave} &
0.524 [0.458, 0.590] &
0.530 [0.470, 0.590] &
0.0336 & 0.0306 & \textbf{1.10} \\
\textbf{W} &
0.524 [0.453, 0.594] &
0.532 [0.468, 0.596] &
0.0360 & 0.0328 & \textbf{1.10} \\
\bottomrule
\end{tabular}%
}
\parbox{0.98\linewidth}{\footnotesize\emph{Note.} The controlled-masking analysis uses $n=50$ labelled and $N=100$ unlabelled observations per arm. The final column is $SE_{\mathrm{SUP}}/SE_{\mathrm{SSL}}$; values above one indicate smaller SSL standard errors. SSL inference uses $B=200$ full-refit multinomial perturbations.}
\end{table}

Table~\ref{tab:nhanes-creatinine-hypertension-efficiency} separates changes in the centre of the estimate from changes in precision. Across the five quantities, the largest absolute difference between the SUP and SSL point estimates is $0.012$, small relative to their standard errors. Every SSL standard error is smaller, so the additional covariates chiefly sharpen the estimated comparison probabilities. The resulting intervals provide greater precision around broadly similar point estimates.

The group-specific SSL estimates are $0.483$, $0.593$, and $0.515$ for G1--G3, respectively. G2 has the largest point estimate, while G1 and G3 are closer to $0.5$; the overlapping intervals and their inclusion of $0.5$ make this ordering descriptive. The precision gain is distributed across the age groups: $SE_{\mathrm{SUP}}/SE_{\mathrm{SSL}}$ ranges only from $1.08$ to $1.12$, with the largest ratio in G2.

Aggregation gives a complementary summary. Ave and W yield similar SSL estimates, $0.530$ and $0.532$, and smaller standard errors than the individual group estimates. Their standard-error ratios are both $1.10$, so the semi-supervised gain persists under equal-group and prespecified weighted aggregation. The close Ave and W estimates indicate that the two prespecified aggregation schemes give similar summaries in this analysed sample.
\begin{figure}[!htbp]
\centering
\includegraphics[width=.97\linewidth]{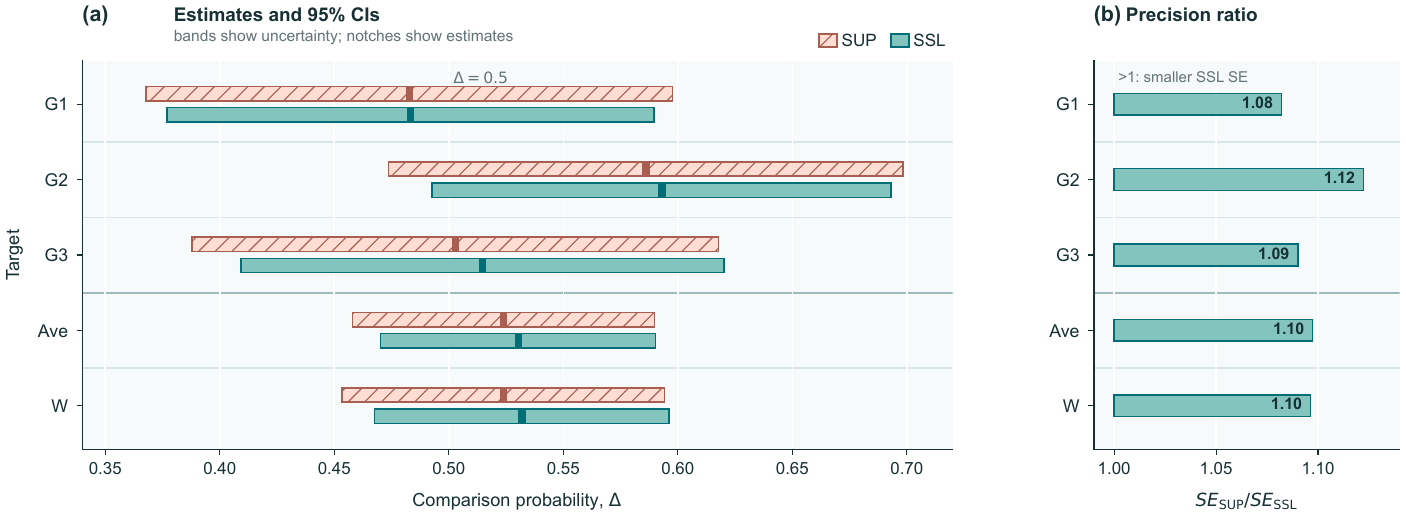}
\caption{NHANES 2011--2018 controlled-masking illustration. Panel (a) shows SUP and SSL estimates with 95\% confidence intervals; the dashed line marks $\Delta=0.5$. Panel (b) reports $SE_{\mathrm{SUP}}/SE_{\mathrm{SSL}}$, with values above one indicating smaller SSL standard errors.}
\label{fig:nhanes-forest-precision}
\end{figure}
\FloatBarrier

\section{Conclusions and Discussion}\label{sec:discussion}

This paper develops semi-supervised inference for group-specific two-sample comparison functionals and their prespecified aggregates. Treating the group-specific functionals as the primary estimands keeps heterogeneity visible, while equal-group and deterministic weighted summaries retain explicit scientific interpretations rather than being induced by pooling.

The main theoretical contribution is an original-observation four-role representation. Labelled outcomes determine the conditional-comparison component, whereas labelled and unlabelled covariates contribute to the outer covariate-distribution component, with their overlap retained in the first-order variance. Under correct single-index specification, additional covariates reduce the outer component but cannot remove uncertainty from learning the conditional comparison with a finite labelled sample. The centred low-density construction preserves the target and full outer averaging without discarding attainable evaluation regions.

The empirical results reflect this decomposition. Across the point-estimation studies, increasing unlabelled information yields progressively larger precision gains while bias remains small, and the pattern persists for heterogeneous group targets, the bounded CLIP functional, and unequal allocation. Full-refit multinomial perturbation preserves the dependence between labelled fitting and outer averaging; its standard errors track empirical sampling variability and coverage remains near the nominal level. In the NHANES analysis, SUP and SSL point estimates are similar, with smaller SSL standard errors throughout.

The scope of the present theory is defined by a bounded comparison response with a correctly specified single-index conditional mean, a fixed number of mutually independent groups and independent arms, prespecified deterministic aggregation weights, and nested labelled and full covariate samples drawn from the same arm-specific distributions. The current first-order analysis therefore does not describe the effects of index-model approximation, covariate shift between labelled and unlabelled samples, or dependence from paired or complex sampling.

These boundaries point to extensions that retain the grouped target structure. Broader outcome classes and more flexible conditional estimators would permit a direct study of approximation error beyond the working single-index model, while transport or reweighting methods could accommodate external unlabelled samples drawn from a different covariate distribution. Extending the original-unit variance accounting to growing collections of groups, estimated aggregation weights, and dependent sampling would broaden the framework to paired and survey designs. More broadly, the analysis clarifies how unlabelled covariate information can sharpen grouped two-sample inference while leaving the prespecified scientific target unchanged.

\section*{Acknowledgement}

This research was supported by the Natural Science Foundation of
Xinjiang Uygur Autonomous Region (No.~2025D14015) and the National
Natural Science Foundation of China (No.~72571102).

\clearpage
\appendix

\makeatletter
\begingroup
\def\@seccntformat#1{Appendix~\csname the#1\endcsname.\quad}
\section{Notation and Conventions}\label{app:A}
\endgroup
\makeatother

\subsection{Empirical Operators and Score Notation}

Suppress the group index $k$. Let $O_{ra}=(Y_{ra},\mathbf{X}_{ra})$,
$\mathbb P_{r,n}$ be the empirical measure of the $n_r$ labelled units,
and $\mathbb Q_{r,M}$ the covariate empirical measure of all $M_r$ units:
\[
 \mathbb P_{r,n}f=\frac1{n_r}\sum_{a=1}^{n_r}f(O_{ra}),\qquad
 \mathbb Q_{r,M}f=\frac1{M_r}\sum_{a=1}^{M_r}f(\mathbf{X}_{ra}).
\]
Their two-arm product operators are
\[
 \mathbb P_{21}f
 =\frac1{n_2n_1}\sum_{i=1}^{n_2}\sum_{j=1}^{n_1}f(O_{2i},O_{1j}),
 \qquad
 \mathbb Q_M f
 =\frac1{M_2M_1}\sum_{a=1}^{M_2}\sum_{b=1}^{M_1}f(\mathbf{X}_{2a},\mathbf{X}_{1b}).
\]
Let $P=P_2P_1$ denote expectation under an independent labelled
treatment--control pair, and let $Q=Q_2Q_1$ denote expectation under an
independent covariate pair.
The first $n_r$ entries of each full covariate array are the labelled
covariates, so $\mathbb P_{21}$ and $\mathbb Q_M$ overlap.

Let $\mathbf{V}_{ij}=(\mathbf{X}_{2i}^\top,\mathbf{X}_{1j}^\top)^\top$,
$H_{ij}=L(Y_{2i},Y_{1j},t)$, $H_{c,ij}=H_{ij}-c_H$, and
$S_{ij}(\boldsymbol{\beta})=\mathbf{V}_{ij}^\top\boldsymbol{\beta}$, and write
$\mathbf{V}^q=(\mathbf{X}_2^{q\top},\mathbf{X}_1^{q\top})^\top$ for an independent symbolic
covariate pair. Near the true direction, use the local tangent parametrisation
\[
 \boldsymbol{\beta}(\boldsymbol{\vartheta})=\frac{\boldsymbol{\beta}_0+\mathbf{R}\boldsymbol{\vartheta}}
 {\|\boldsymbol{\beta}_0+\mathbf{R}\boldsymbol{\vartheta}\|},
 \qquad \mathbf{R}^\top \mathbf{R}=\mathbf{I},\qquad \mathbf{R}^\top\boldsymbol{\beta}_0=0.
\]
Zero coordinates of $\boldsymbol{\beta}_0$ do not reduce the tangent dimension.

For $K_h(u)=h^{-1}K(u/h)$ define
\[
 \begin{aligned}
 d_{\boldsymbol{\beta}}(s)&=P K_h(S_{\boldsymbol{\beta}}-s),&
 g_{\boldsymbol{\beta}}(s)&=P\{K_h(S_{\boldsymbol{\beta}}-s)H\},\\
 q_{\boldsymbol{\beta},c}(s)&=g_{\boldsymbol{\beta}}(s)-c_Hd_{\boldsymbol{\beta}}(s)
 =P\{K_h(S_{\boldsymbol{\beta}}-s)H_c\}.
 \end{aligned}
\]
Hats denote replacement of $P$ by $\mathbb P_{21}$.

\subsection{Projection and Overlap Notation}

The limiting roles $\phi_{X2},\phi_{X1},\phi_{Y2},\phi_{Y1}$ correspond,
respectively, to the two outer full-covariate empirical distributions and
the two labelled residual projections. The overlap labels
$L_{00},L_{10},L_{01},L_{11}$ denote, respectively, no shared fitting unit,
a shared treatment unit, a shared control unit, and both units shared with
the outer pair; this notation retains the leave-in estimator.

\begingroup
\UseAppendixBEnvironments
\makeatletter
\begingroup
\def\@seccntformat#1{Appendix~\csname the#1\endcsname.\quad}
\section{Auxiliary Lemmas}\label{app:B}
\endgroup
\makeatother

Throughout Appendix~B, group indices are suppressed when no confusion arises;
we use the notation of Appendix~\ref{app:A}, the tuning sequences in
\eqref{eq:bandwidth-rule}--\eqref{eq:floor-rule}, and the assumptions of
Section~\ref{sec:theory} as indicated.

\phantomsection\label{app:B:assumption-technical}
Throughout Appendix~B, Assumption~\ref{ass:E} is used in the following
uniform form. On each moving-support stratum, $c_{\pm,k,\boldsymbol{\vartheta}}$ are uniformly bounded
and bounded away from zero. The conditional moment-density maps and their
level, local-coordinate, bandwidth, and score-scale derivatives have uniformly smooth
versions with bounded endpoint derivatives and endpoint-compatible residual
orders; the moving denominator satisfies the local Lipschitz bound required
for floor crossing.

The level and derivative kernel/deletion families, including row, column, and
row--column intersection deletions, hard-floor and overlap indices, and
once- and twice-deleted arrays, form a finite-dimensional or VC-type family
with entropy factor $\ell_n=1+\log(n_L/h)$, uniformly over bounded positive
score-scale perturbations and shrinking bandwidth neighbourhoods. The
required primitive second moments, armwise random-evaluation $L_2$
first-projection envelopes, and canonical-pair and overlap integrability
bounds hold uniformly. The two-arm Hoeffding/VC maximal inequality,
including its exponential-tail form, and the required stochastic
equicontinuity also hold uniformly. The global numerator, denominator, and leave-in
profile-criterion classes are Glivenko--Cantelli, supplying the uniform law
for Proposition~\ref{prop:profile-consistency}.

\phantomsection\label{subsec:app-rate-conditions}
Suppressing the group index in local expressions, on a local stratum with
endpoint order $p$, write
\[
 \gamma=\frac{p}{p-1},\qquad
 \kappa_r=\frac{p_r-1}{p-1},\qquad
 \kappa_{12}=\frac{p-2}{p-1}=2-\gamma.
\]
For the common fixed-$K$ rates, let $p_* = \max_{1\le k\le K}p_{*,k}$ and
$\gamma_*=p_*/(p_*-1)$; let $\kappa_{r,*}$ and $\kappa_{12,*}$ denote the
corresponding worst-case overlap exponents. All $\asymp$ statements are
uniform over the fixed groups and allow deterministic multiplicative
constants bounded away from zero and infinity. For
\[
h\asymp n_{\min}^{-\alpha},\qquad
w\asymp n_L^{-\eta},\qquad b_n=n_L^{-1/2},
\]
the rate region \textup{(R)} in Section~\ref{sec:theory} implies
\[
 \sqrt{n_L}h^2\to0,\qquad
 \sqrt{n_L}w^{\gamma_*}\to0,\qquad
 \frac{b_n}{h}\to0,\qquad \frac{b_n}{w}\to0,
\]
\[
 \eta\kappa_{12,*}<\frac12,\qquad
 1-\alpha-\eta\kappa_{r,*}>0\quad(r=1,2),\qquad
 3-2\alpha-\eta\kappa_{12,*}>0,
\]
and $w/h^{p-1}\to\infty$ uniformly over the relevant local strata.

For the Section~\ref{sec:simulation} design, the true-index and conservative
endpoint orders are $p_0=4$ and $p_*=8$, hence
$\gamma_0=4/3$ and $\gamma_*=8/7$.
For the numerical choice $(\alpha,\eta)=(1/3,15/32)$,
\[
 \frac14<\frac13<\frac25,
 \qquad
 \frac{1}{2\gamma_*}=\frac7{16}<\frac{15}{32}<\frac12
 =1-\frac{3(1/3)}2.
\]
Hence the Section~\ref{sec:simulation} choice lies strictly inside
\textup{(R)}.
\subsection{Low-Density and Two-Sample Bounds}\label{app:B:oracle-algebra}

\phantomsection\label{sec:ld-ratio-decomp}
For the hard-floor linearisation at a fixed evaluation score $s$, write
\[
 A_{c,h}(s)=\widehat q_{c,h}(s)-q_{c,h}(s)
            =A_h(s)-c_HB_h(s),\qquad
 B_h(s)=\widehat d_h(s)-d_h(s),
\]
$T_w(x)=x\vee w$, and $J_{h,w}(s)=\mathbf 1\{d_h(s)>w\}$. Define
\begin{align}
 H^c_{h,w}(u_2,u_1;s)
 ={}&\frac{K_h\{S(u_2,u_1)-s\}H_c(u_2,u_1)-q_{c,h}(s)}{T_w\{d_h(s)\}}
 \notag\\
 &-J_{h,w}(s)\frac{q_{c,h}(s)
 [K_h\{S(u_2,u_1)-s\}-d_h(s)]}{d_h(s)^2}.
 \label{eq:ld-H-definition}
\end{align}
Then $PH^c_{h,w}(\cdot;s)=0$ for every $s$. Equivalently,
\[
 H^c_{h,w}(u;s)=
 \begin{cases}
 \displaystyle
 \frac{K_h(S_u-s)\{H_{c,u}-(m_h(s)-c_H)\}}{d_h(s)},&d_h(s)>w,\\[1.2ex]
 \displaystyle
 \frac{K_h(S_u-s)H_{c,u}-q_{c,h}(s)}{w},&d_h(s)\le w.
 \end{cases}
\]
\begin{lemma}[Low-Density Consequences]\label{lem:b1-low-density}
Under the low-density score-geometry clause of Assumption~\ref{ass:E}, there
are constants $0<c<C<\infty$ and $h_0,r_0>0$ such that the following hold.
\begin{enumerate}[label=(\roman*),leftmargin=*]
\item \emph{Smoothed endpoint order.} For $0<h<h_0$ and $s\in[a,b]$ with
$r(s):=\min(s-a,b-s)\le r_0$,
\[
 c\{r(s)+h\}^{p-1}\le d_h(s)\le C\{r(s)+h\}^{p-1}.
\]
On the remaining compact interior, $d_h$ is bounded above and away from
zero, and, uniformly for sufficiently small $h$,
\[
 f(s)\le C d_h(s),\qquad s\in[a,b].
\]

\item \emph{Low-density mass and margin.} Uniformly for $u\downarrow0$ with
$u/h^{p-1}\to\infty$,
\[
 P\{d_h(S)\le u\}\le Cu^\gamma,
 \qquad \gamma=\frac{p}{p-1},
\]
and, for $0<v\le u/2$,
\[
 P\{|d_h(S)-u|\le v\}\le Cv u^{\gamma-1}.
\]

\item \emph{Integrated inverse moments.} With $D_h=d_h(S)$, for every fixed
$q>0$ and $w/h^{p-1}\to\infty$,
\[
 E\{(D_h\vee w)^{-q}\}\le C_q
 \begin{cases}
 1,&q<\gamma,\\
 1+\log(1/w),&q=\gamma,\\
 w^{\gamma-q},&q>\gamma.
 \end{cases}
\]

\item \emph{Population floor bias.} Let
\[
 m_h(s)=g_h(s)/d_h(s),\qquad
 m^c_{h,w}(s)=c_H+
 \{g_h(s)-c_Hd_h(s)\}/\{d_h(s)\vee w\},
\]
and define $\Delta_h=Qm_h(S^q)$ and
$\Delta^c_{h,w}=Qm^c_{h,w}(S^q)$. Then
\[
 |\Delta^c_{h,w}-\Delta_h|\le Cw^\gamma.
\]
If $\nu\asymp n_L$ under comparable labelled arm sizes and
$w\asymp n_L^{-\eta}$ with $\eta\gamma>1/2$, then
\[
 \sqrt{\nu}\,|\Delta^c_{h,w}-\Delta_h|\to0.
\]

\item \emph{Integrated hard-floor kernel bound.} Define
\[
 G_{h,w}(u_2,u_1)
 =QH^c_{h,w}(u_2,u_1;\mathbf{X}_2^q,\mathbf{X}_1^q).
\]
Then
\[
 \sup_{u_2,u_1}|G_{h,w}(u_2,u_1)|\le C
\]
uniformly over the admissible $h,w$ regime.
\end{enumerate}

\begin{proof}
At the lower endpoint, write $s=a+r$ and
$f(a+x)=x^{p-1}c_0(x)$, where $c_0$ is bounded above and away from zero. Then
\[
 d_h(a+r)=\int_0^{b-a}h^{-1}K\{(x-r)/h\}x^{p-1}c_0(x)\,dx.
\]
The bound $x^{p-1}\le C\{r^{p-1}+|x-r|^{p-1}\}$ and the finite Gaussian
$(p-1)$st moment give the upper estimate. For the lower estimate, integration
over $[r,r+h]$ for $r\ge h$ and $[h,2h]$ for $r<h$, with kernel positivity on
the rescaled intervals, gives $cr^{p-1}$ and $ch^{p-1}$. The symmetric
upper-endpoint argument, compact-interior positivity, and the
approximate-identity property yield
$d_h(a+r)\asymp(r+h)^{p-1}$ at both endpoints and the interior bounds;
$f(a+r)\asymp r^{p-1}$ then gives $f(s)\le Cd_h(s)$.

For $u/h^{p-1}\to\infty$, the endpoint order gives
$d_h(S)\le u\Rightarrow r(S)\lesssim u^{1/(p-1)}$, and hence
\[
 P\{d_h(S)\le u\}
 \le C\int_0^{Cu^{1/(p-1)}}r^{p-1}\,dr
 \le Cu^{p/(p-1)}.
\]
On the shrinking endpoint region,
\[
 |d_h'(a+r)|\asymp(r+h)^{p-2},
\]
and similarly at the upper endpoint. At $d_h\asymp u$, a denominator band of
width $2v$ therefore has endpoint-distance length at most
$Cv u^{-(p-2)/(p-1)}$; since the score density is of order $u$,
\[
 P\{|d_h(S)-u|\le v\}\le Cv u^{1/(p-1)}
 =Cv u^{\gamma-1},\qquad 0<v\le u/2.
\]
For $F_h(u)=P(D_h\le u)$, the boundary term is
$w^{-q}F_h(w)\le Cw^{\gamma-q}$, and integration by parts gives
\[
 E\{D_h^{-q}\mathbf 1\{D_h>w\}\}
 \le C+q\int_w^{C_0}u^{-q-1}F_h(u)\,du
 \le C+Cq\int_w^{C_0}u^{\gamma-q-1}\,du.
\]
Evaluating this integral gives (iii). For the floor bias,
\[
 m^c_{h,w}(s)-m_h(s)
 =\left(1-\frac{d_h(s)}w\right)\{c_H-m_h(s)\}
 \mathbf 1\{d_h(s)<w\}.
\]
The bracketed factor is uniformly bounded, so the mass bound gives
$|\Delta^c_{h,w}-\Delta_h|=O(w^\gamma)$. With $\nu\asymp n_L$ and
$w\asymp n_L^{-\eta}$, this is root-scale negligible when
$\eta\gamma>1/2$.

For the integrated kernel bound, when $d_h(s)>w$,
\[
 K_h(S_u-s)\{H_{c,u}-(m_h(s)-c_H)\}/d_h(s)
 =K_h(S_u-s)\{H_u-m_h(s)\}/d_h(s).
\]
The bound $f(s)/d_h(s)\le C$ controls its evaluation-score integral by
$C\int K_h(S_u-s)\,ds$. On $d_h(s)\le w$,
$f(s)\le Cd_h(s)\le Cw$ and $|q_{c,h}(s)|\le Cd_h(s)\le Cw$ control the
second piece, proving the uniform bound for $G_{h,w}$.
\end{proof}
\end{lemma}

\begin{lemma}[Hard-Floor Expansion and Integrated Remainder]\label{lem:b2-floor-expansion}
\begin{enumerate}[label=(\roman*),leftmargin=*]
\item \emph{Process bound.} Let
\(\ell_n=1+\log(n_L/h)\) and set
\[
 a_n=\sqrt{\ell_n/n_2}+\sqrt{\ell_n/n_1}
     +\sqrt{\ell_n/(n_2n_1h)}.
\]
Then
\[
 \|\widehat q_{c,h}-q_{c,h}\|_\infty
 \vee\|\widehat d_h-d_h\|_\infty=O_p(a_n),
\]
with the same bound for the finite deleted collection, and $a_n=o(w)$ under
\textup{(R)}.

\item \emph{Exact hard-floor expansion and pointwise remainder.} For every
evaluation score $s$,
\begin{equation}
 \widehat m^c_{h,w}(s)-m^c_{h,w}(s)
 = (\mathbb P_{21}-P)H^c_{h,w}(\cdot;s)+R^c_{h,w}(s),
 \label{eq:ld-exact-floor-identity}
\end{equation}
where
\begin{align}
 R^c_{h,w}(s)
 ={}&A_{c,h}(s)\left[
 \frac1{T_w\{d_h(s)+B_h(s)\}}
 -\frac1{T_w\{d_h(s)\}}
 \right]\notag\\
 &+q_{c,h}(s)\left[
 \frac1{T_w\{d_h(s)+B_h(s)\}}
 -\frac1{T_w\{d_h(s)\}}
 +J_{h,w}(s)\frac{B_h(s)}{d_h(s)^2}
 \right].
 \label{eq:ld-exact-remainder}
\end{align}
Let
\[
 \mathcal E_n(\delta)=
 \{\|A_{c,h}\|_\infty\vee\|B_h\|_\infty\le\delta\},
 \qquad 0<\delta\le w/4.
\]
On $\mathcal E_n(\delta)$, with $d=d_h(s)$,
\[
 |R^c_{h,w}(s)|\le
 C\delta^2d^{-2}\mathbf 1\{d>w+2\delta\}
 +C\delta w^{-1}\mathbf 1\{|d-w|\le2\delta\},
\]
and $R^c_{h,w}(s)=0$ when $d<w-2\delta$. The identity remains valid
across floor crossings.

\item \emph{Integrated remainder.} Under the standing low-density and
kernel-process conditions,
\[
 \mathbb Q_M|R^c_{h,w}(S^q)|
 =O_p\{a_n^2w^{\gamma-2}\}.
\]
Under the comparable-arm sampling conditions and \textup{(R)},
\[
 \sqrt\nu\,\mathbb Q_M R^c_{h,w}(S^q)=o_p(1).
\]
\end{enumerate}

\begin{proof}
\noindent
The two-arm Hoeffding decomposition and the primitive VC/envelope maximal
inequality stated at the beginning of Appendix~B give part (i), including the
finite deleted collection. Substituting
$\widehat q_{c,h}=q_{c,h}+A_{c,h}$ and
$\widehat d_h=d_h+B_h$, and adding and subtracting
$A_{c,h}/T_w(d_h)-J_{h,w}q_{c,h}B_h/d_h^2$, gives
\eqref{eq:ld-exact-floor-identity}--\eqref{eq:ld-exact-remainder}. The
centring identities
\[
 P\{K_hH_c-q_{c,h}\}=0,\qquad P\{K_h-d_h\}=0
\]
identify the added term as $(\mathbb P_{21}-P)H^c_{h,w}$.

\medskip
\noindent
On $\mathcal E_n(\delta)$, write $d=d_h(s)$, $A_c=A_{c,h}(s)$, and
$B=B_h(s)$. If $d<w-2\delta$, then $d+B<w$, both denominators in
\eqref{eq:ld-exact-remainder} equal $w$, and $J_{h,w}=0$, so $R^c=0$.
If $d>w+2\delta$, both denominators are unfloored and
\[
 R^c=-\frac{A_cB}{d(d+B)}
       +\frac{q_{c,h}B^2}{d^2(d+B)}.
\]
Since $|q_{c,h}|\le Cd$ and $|B|\le\delta\le d/2$,
$|R^c|\le C\delta^2d^{-2}$. On $|d-w|\le2\delta$, the map
$x\mapsto1/(x\vee w)$ is $w^{-2}$-Lipschitz; with
$|q_{c,h}(s)|\le Cd\le Cw$ and $\delta\le w/4$, this gives
$|R^c|\le C\delta/w$.

\medskip
\noindent
Choose $\delta_n=Ca_n$ with $C$ sufficiently large. Since $a_n=o(w)$,
$\delta_n\le w/4$ eventually, and the exponential-tail version of the
primitive process bound makes
$P\{\mathcal E_n(\delta_n)^c\}$ negligible at the polynomial orders needed
here. On $\mathcal E_n(\delta_n)$, Lemma~\ref{lem:b1-low-density}(iii)
and (ii), respectively, give
\begin{align*}
 Q|R^c_{h,w}(S)|
 &\le C\delta_n^2E\{D_h^{-2}\mathbf 1\{D_h>w\}\}
 +C\frac{\delta_n}{w}P\{|D_h-w|\le2\delta_n\}\\
 &\le C\delta_n^2w^{\gamma-2}
 +C\frac{\delta_n}{w}\delta_nw^{\gamma-1}
 \le C\delta_n^2w^{\gamma-2}.
\end{align*}
The first term uses the inverse-moment bound and the second the margin bound
in Lemma~\ref{lem:b1-low-density}. The right-hand side is a deterministic
function of the evaluation score, so equality of the evaluation marginals and
Markov's inequality transfer the same $O_p$ order to
$\mathbb Q_M|R^c_{h,w}(S^q)|$, regardless of fitting--evaluation index
overlap. On $\mathcal E_n(\delta_n)^c$, bounded marks and polynomial
$h^{-1},w^{-1}$ envelopes are dominated by the same exponential tail. Since
$a_n^2=O(\ell_n/n_L)$, the root-scale conclusion follows from
\textup{(R)}.
\end{proof}
\end{lemma}

\phantomsection\label{sec:ld-pair-dependence}
To account for dependence between fitting and evaluation pairs, let
$H^c_{h,w}$ be defined by \eqref{eq:ld-H-definition}. A fitting pair and an
evaluation pair may be distinct, share the treatment unit, share the control
unit, or share both units.
\begin{lemma}[Two-Sample Projection, Index-Overlap, and Canonical Bounds]\label{lem:b3-collision-canonical}
\hfill\break
\begin{enumerate}[label=(\roman*),leftmargin=*]
\item \emph{Moments.} Under Assumption~\ref{ass:E},
\begin{align*}
 E (H^c_{h,w})^2&\le C h^{-1},&&L_{00},\\
 E (H^c_{h,w})^2&\le C h^{-1}w^{-\kappa_2},&&L_{10},\\
 E (H^c_{h,w})^2&\le C h^{-1}w^{-\kappa_1},&&L_{01},\\
 E (H^c_{h,w})^2&\le C h^{-2}w^{-\kappa_{12}},&&L_{11}.
\end{align*}
Here $L_{10}$ shares the arm-2 unit and $L_{01}$ shares the arm-1 unit.
\[
 \kappa_2=\frac{p_2-1}{p-1},\qquad
 \kappa_1=\frac{p_1-1}{p-1},\qquad
 \kappa_{12}=\frac{p-2}{p-1}=2-\gamma.
\]
The class $L_{00}$ has no shared unit, whereas $L_{11}$ shares both units.
Moreover, $G_{h,w}=QH^c_{h,w}$ is uniformly bounded by
Lemma~\ref{lem:b1-low-density}(v).

\item \emph{Hoeffding decomposition.} For any square-integrable pair
kernel $r(u_2,u_1)$,
\begin{align*}
 \mathbb P_{2,n_2}\mathbb P_{1,n_1}r-Pr
 ={}&(\mathbb P_{2,n_2}-P_2)P_1r
 +(\mathbb P_{1,n_1}-P_1)P_2r\\
 &+(\mathbb P_{2,n_2}-P_2)(\mathbb P_{1,n_1}-P_1)r.
\end{align*}
The final term is canonical in both arm roles and has variance
\[
\frac{\|\Pi_{21}r\|_2^2}{n_2n_1}.
\]
The identity remains valid when the labelled observations are nested in a
larger outer covariate sample; nesting changes cross-covariances with outer
terms but does not change the algebra.

\item \emph{Mixed remainder.} Let
$R_{4,h,w}^{\mathrm{lin}}$ be the sum of the nine terms containing at
least one fitting empirical operator and at least one outer empirical
operator in the exact four-component expansion. Then
\begin{align}
 E\{(R_{4,h,w}^{\mathrm{lin}})^2\}
 \le C r_{4,h,w}^2,
 \label{eq:ld-r4-bound}
\end{align}
where
\begin{align*}
r_{4,h,w}^2={}&
 \frac1h\left(
 \frac1{n_2M_2}+\frac1{n_2M_1}
   +\frac1{n_1M_2}+\frac1{n_1M_1}\right)\\
&+\frac{w^{-\kappa_2}}{M_2^2h}
 +\frac{w^{-\kappa_1}}{M_1^2h}
 +\frac{w^{-\kappa_{12}}}{M_2^2M_1^2h^2}
 +\left(\frac1{M_2}+\frac1{M_1}\right)^2.
\end{align*}
Under the standing sample-size conditions and \textup{(R)},
\[
 R_{4,h,w}^{\mathrm{lin}}=o_{L_2}(\nu^{-1/2}).
\]

\item \emph{Canonical pair remainders.} Let $R_X^d$ be the canonical
two-arm outer remainder generated by the bounded hard-floor population kernel
$m^c_{h,w}\{S^q(\mathbf{X}_2,\mathbf{X}_1)\}$, and let $R_Y^d$ be the canonical labelled
remainder generated by $G_{h,w}$. Then
\[
 E(R_X^d)^2=O\{(M_2M_1)^{-1}\},
 \qquad
 E(R_Y^d)^2=O\{(n_2n_1)^{-1}\},
\]
and both are $o_{L_2}(\nu^{-1/2})$.
\end{enumerate}

\begin{proof}
\noindent
By Lemma~\ref{lem:b1-low-density}(v), $G_{h,w}$ is uniformly bounded.
Boundedness of $H_c$ and $|q_{c,h}|\le Cd_h$ reduce the four overlap
configurations, up to bounded additive terms, to squared-kernel integrals
divided by $(d_h(S^q)\vee w)^2$. For distinct fitting and evaluation roles,
\[
 E\{K_h(S-S^q)^2\mid S^q=s\}
 \le Ch^{-1}d_{ch}(s)\le Ch^{-1}d_h(s),
\]
so, with $D_h=d_h(S^q)$,
\[
 E\{(H^c_{h,w})^2\}
 \le Ch^{-1}E\left\{\frac{D_h}{(D_h\vee w)^2}\right\}+C
 =O(h^{-1})
\]
by Lemma~\ref{lem:b1-low-density}(iii). The displayed two-arm decomposition
is the expansion of the product empirical operator; its last term is
canonical in both arm roles and has the stated variance.
For $L_{10}$, the fitting and evaluation arm-2 units coincide. Near a lower
endpoint, let $x$ and $y$ be the arm-2 and arm-1 endpoint distances. Their
density orders are $x^{p_2-1}$ and $y^{p_1-1}$; the nonshared-arm squared
kernel contributes $h^{-1}(y+h)^{p_1-1}$, while the full denominator is of
order $(x+y+h)^{p-1}$. With $r_w=w^{1/(p-1)}$, the simplex coordinates
$x=rt$, $y=r(1-t)$ give
\begin{align*}
 E\{(H^c_{h,w})^2\}
 &\le Ch^{-1}\left(
 w^{-2}\int_0^{r_w}r^{2p_1+p_2-2}\,dr
 +\int_{r_w}^{r_0}r^{-p_2}\,dr+1\right)\\
 &\le Ch^{-1}w^{-(p_2-1)/(p-1)}.
\end{align*}
The upper endpoint is analogous and the compact interior contributes
$O(h^{-1})$. Interchanging the arms gives the $L_{01}$ order. For $L_{11}$,
$K_h(0)^2=O(h^{-2})$, and Lemma~\ref{lem:b1-low-density}(iii) with $q=2$
gives
\[
 E\{(H^c_{h,w})^2\}
 \le Ch^{-2}E\{(D_h\vee w)^{-2}\}
 \le Ch^{-2}w^{\gamma-2}
 =Ch^{-2}w^{-\kappa_{12}}.
\]

\medskip
\noindent
Partition the nine mixed fitting--outer terms by the overlap classes
$L_{00},L_{10},L_{01},L_{11}$ defined in Appendix~\ref{app:A}. The
distinct-role terms after the two-arm Hoeffding projection give the first line
of \eqref{eq:ld-r4-bound}. The terms with one or both arm indices shared use
the $L_{10}$, $L_{01}$, and $L_{11}$ bounds together with their outer-operator
coefficients. These coefficients have orders one, $M_2^{-1}$, $M_1^{-1}$,
and $(M_2M_1)^{-1}$ for the four overlap classes, respectively; squaring them
and applying the corresponding overlap bounds gives the next three terms. Exact
index-deleted normalizers differ from their population counterparts by
$O(M_2^{-1}+M_1^{-1})$, giving the final term. After the relevant deletion,
canonical components supported on distinct original-unit roles are orthogonal.
With $M_r\ge n_r\asymp n_L$ and $\nu\asymp n_L$, the strict inequalities
\[
 1-\alpha>0,\qquad
 1-\alpha-\eta\kappa_r>0\quad(r=1,2),\qquad
 3-2\alpha-\eta\kappa_{12}>0
\]
imply $r_{4,h,w}^2=o(\nu^{-1})$ and hence
$R_{4,h,w}^{\mathrm{lin}}=o_{L_2}(\nu^{-1/2})$.

\medskip
\noindent
The population kernel defining $R_X^d$ is bounded because $H_c$ is bounded
and $|q_{c,h}|\le Cd_h$; Lemma~\ref{lem:b1-low-density}(v) bounds
$G_{h,w}$.
Applying the canonical variance formula from part (ii) yields
\[
 E(R_X^d)^2=O\{(M_2M_1)^{-1}\},\qquad
 E(R_Y^d)^2=O\{(n_2n_1)^{-1}\},
\]
which are both $o_{L_2}(\nu^{-1/2})$ under the standing sample-size
conditions.
\end{proof}
\end{lemma}

\subsection{Profile and Index Estimation}\label{app:B:profile-index}
\phantomsection\label{sec:idx-profile-estimator}
We work in population-standardized score coordinates. Using the local tangent
parametrisation and pair notation of Appendix~\ref{app:A}, set
\(S_{\boldsymbol{\vartheta}}(\mathbf{V})=\mathbf{V}^\top\boldsymbol{\beta}(\boldsymbol{\vartheta})\).

For each $\boldsymbol{\vartheta}$, let
\[
 \mu_{\boldsymbol{\vartheta}}=E\{S_{\boldsymbol{\vartheta}}(\mathbf{V})\},\qquad
 \sigma_{\boldsymbol{\vartheta}}^2=\Var\{S_{\boldsymbol{\vartheta}}(\mathbf{V})\},\qquad
 T_{\boldsymbol{\vartheta}}(\mathbf{V})=\frac{S_{\boldsymbol{\vartheta}}(\mathbf{V})-\mu_{\boldsymbol{\vartheta}}}{\sigma_{\boldsymbol{\vartheta}}},
\]
with $\bar h\asymp n_{\min}^{-\alpha}$ and
\[
 K_{\boldsymbol{\vartheta}}(u,v)
 =K_{\bar h}\{T_{\boldsymbol{\vartheta}}(u)-T_{\boldsymbol{\vartheta}}(v)\}.
\]
Lemma~\ref{lem:b6-candidate-scale}(i)
gives the equivalent raw-score scales
$\widetilde h_{\boldsymbol{\vartheta}}=\sigma_{\boldsymbol{\vartheta}}\bar h$ and
$\widetilde w_{\boldsymbol{\vartheta}}=\bar w/\sigma_{\boldsymbol{\vartheta}}$; bounded positive local
rescaling preserves the endpoint orders and rates. Population and empirical
numerator, denominator, and centred-numerator operators use $K_{\boldsymbol{\vartheta}}$
and, respectively, $P$ and $\mathbb P_{21}$.

For $w=\bar w\asymp n_L^{-\eta}$, let
$\rho_w(d)=w r(d/w)$, where $r$ is the stabilizer in
Section~\ref{sec:methodology}.
Then
\[
 c(d\vee w)\le\rho_w(d)\le C(d\vee w),
\]
\[
 0\le\rho_w'(d)\le C,
 \quad |\rho_w''(d)|\le Cw^{-1}\mathbf 1\{w/2<d<2w\},
\]
\[
 |\rho_w'''(d)|\le Cw^{-2}\mathbf 1\{w/2<d<2w\}.
\]
The centred profile population and empirical fits are
\[
 m^{\rho,c}_{\boldsymbol{\vartheta}}(v)=c_H+\frac{q_{\boldsymbol{\vartheta},c}(v)}
 {\rho_w\{d_{\boldsymbol{\vartheta}}(v)\}},\qquad
 \widehat m^{\rho,c}_{\boldsymbol{\vartheta}}(v)=c_H+
 \frac{\widehat q_{\boldsymbol{\vartheta},c}(v)}
 {\rho_w\{\widehat d_{\boldsymbol{\vartheta}}(v)\}},
\]
where $\widehat q_{\boldsymbol{\vartheta},c}=\widehat g_{\boldsymbol{\vartheta}}-c_H\widehat
d_{\boldsymbol{\vartheta}}$. The population fit equals the ordinary Nadaraya--Watson
ratio when $d_{\boldsymbol{\vartheta}}(v)\ge2w$.
The leave-in empirical profile criterion is
\[
 \widehat Q^{\rho,c}(\boldsymbol{\vartheta})
 =\mathbb P_{21}
 \left[\{H-\widehat m^{\rho,c}_{\boldsymbol{\vartheta}}(\mathbf{V})\}^2\right],
\]
and $\widehat{\boldsymbol{\vartheta}}$ is a sign-anchored local/global minimizer satisfying
the consistency condition stated below.

For a fixed evaluation point $v$, write
\[
 N_{\boldsymbol{\vartheta}}=\widehat g_{\boldsymbol{\vartheta}}(v),\quad
 D_{\boldsymbol{\vartheta}}=\widehat d_{\boldsymbol{\vartheta}}(v),\quad
 Q_{\boldsymbol{\vartheta}}=N_{\boldsymbol{\vartheta}}-c_HD_{\boldsymbol{\vartheta}},\quad
 R_{\boldsymbol{\vartheta}}=\rho_w(D_{\boldsymbol{\vartheta}}).
\]
Let dots denote derivatives with respect to $\boldsymbol{\vartheta}$. In population-standardized
coordinates $\bar h$ is fixed with respect to $\boldsymbol{\vartheta}$. Then
\[
 \dot{\widehat m}^{\rho,c}_{\boldsymbol{\vartheta}}
 =\frac{\dot Q_{\boldsymbol{\vartheta}}}{R_{\boldsymbol{\vartheta}}}
 -\frac{Q_{\boldsymbol{\vartheta}}\rho_w'(D_{\boldsymbol{\vartheta}})\dot D_{\boldsymbol{\vartheta}}}
 {R_{\boldsymbol{\vartheta}}^2}.
\]
Here, exactly,
\[
 \dot Q_{\boldsymbol{\vartheta}}=\dot N_{\boldsymbol{\vartheta}}-c_H\dot D_{\boldsymbol{\vartheta}},
 \qquad
 \ddot Q_{\boldsymbol{\vartheta}}=\ddot N_{\boldsymbol{\vartheta}}-c_H\ddot D_{\boldsymbol{\vartheta}}.
\]
Define
\[
 \dot R_{\boldsymbol{\vartheta}}=\rho_w'(D_{\boldsymbol{\vartheta}})\dot D_{\boldsymbol{\vartheta}},
\]
\[
 \ddot R_{\boldsymbol{\vartheta}}
 =\rho_w''(D_{\boldsymbol{\vartheta}})
   \dot D_{\boldsymbol{\vartheta}}\dot D_{\boldsymbol{\vartheta}}^\top
  +\rho_w'(D_{\boldsymbol{\vartheta}})\ddot D_{\boldsymbol{\vartheta}}.
\]
The exact Hessian is
\begin{align*}
 \ddot{\widehat m}^{\rho,c}_{\boldsymbol{\vartheta}}
 ={}&\frac{\ddot Q_{\boldsymbol{\vartheta}}}{R_{\boldsymbol{\vartheta}}}
 -\frac{\dot Q_{\boldsymbol{\vartheta}}\dot R_{\boldsymbol{\vartheta}}^\top
       +\dot R_{\boldsymbol{\vartheta}}\dot Q_{\boldsymbol{\vartheta}}^\top}
       {R_{\boldsymbol{\vartheta}}^2}
 -\frac{Q_{\boldsymbol{\vartheta}}\ddot R_{\boldsymbol{\vartheta}}}{R_{\boldsymbol{\vartheta}}^2}
 +\frac{2Q_{\boldsymbol{\vartheta}}
       \dot R_{\boldsymbol{\vartheta}}\dot R_{\boldsymbol{\vartheta}}^\top}
       {R_{\boldsymbol{\vartheta}}^3}.
\end{align*}
Consequently
\[
 \dot{\widehat Q}^{\rho,c}(\boldsymbol{\vartheta})
 =-2\mathbb P_{21}
 \left[\{H-\widehat m^{\rho,c}_{\boldsymbol{\vartheta}}(\mathbf{V})\}
 \dot{\widehat m}^{\rho,c}_{\boldsymbol{\vartheta}}(\mathbf{V})\right],
\]
\begin{align*}
 \ddot{\widehat Q}^{\rho,c}(\boldsymbol{\vartheta})
 ={}&2\mathbb P_{21}
 [\dot{\widehat m}^{\rho,c}_{\boldsymbol{\vartheta}}(\mathbf{V})
  \dot{\widehat m}^{\rho,c}_{\boldsymbol{\vartheta}}(\mathbf{V})^\top]\\
 &-2\mathbb P_{21}
 [\{H-\widehat m^{\rho,c}_{\boldsymbol{\vartheta}}(\mathbf{V})\}
  \ddot{\widehat m}^{\rho,c}_{\boldsymbol{\vartheta}}(\mathbf{V})].
\end{align*}
\phantomsection\label{sec:idx-density-cancel}
For $q=0,1,2$, let $G_{q,\boldsymbol{\vartheta},h}(U;v)$ denote a primitive level or
local-coordinate derivative kernel block, including score and bandwidth derivatives,
of derivative order $q$.
Let
\[
 D_{\boldsymbol{\vartheta},h}(v)=P K_{\boldsymbol{\vartheta}}(U,v).
\]
\begin{lemma}[Uniform Profile Approximation and Derivative Bounds]\label{lem:b4-profile-control}
Under the low-density score-geometry and kernel-process clauses of
Assumption~\ref{ass:E}, the following statements hold uniformly on the
fixed local neighbourhood.
\begin{enumerate}[label=(\roman*)]
\item For $q=0,1,2$, uniformly in $\boldsymbol{\vartheta}$ and $v$,
\[
 P\{\|G_{q,\boldsymbol{\vartheta},h}(U;v)\|^2\}
 \le C h^{-(2q+1)}D_{\boldsymbol{\vartheta},h}(v),
\]
$f_{\boldsymbol{\vartheta}}(s)\le C D_{\boldsymbol{\vartheta},h}(s)$, and
\[
 Q\left[
 \frac{P\{\|G_{q,\boldsymbol{\vartheta},h}(U;\mathbf{V}^q)\|^2\}}
 {\rho_w\{D_{\boldsymbol{\vartheta},h}(\mathbf{V}^q)\}^2}
 \right]
 \le C h^{-(2q+1)}.
\]
The same conclusion holds with $D\vee w$ in place of $\rho_w(D)$.

\item All transition-derivative terms arising in the first two local-coordinate
derivatives and local Hessian equicontinuity, including the required
$\rho_w''$ and $\rho_w'''$ factors, are uniformly bounded on their support.
Their random-evaluation $L_r$ norms are $O(w^{\gamma_*/r})$,
$1\le r<\infty$.

\item At $\boldsymbol{\vartheta}=0$,
\begin{align*}
 \|m^{\rho,c}_0(\mathbf{V})-m(S_0)\|_{L_2}
 &=O(h^2+w^{\gamma_0/2}),\\
 \|\dot m^{\rho,c}_0(\mathbf{V})-\mathbf{d}(\mathbf{V})\|_{L_2}
 &=O(h+w^{\gamma_0/2}),\\
 \|\ddot m^{\rho,c}_0(\mathbf{V})\|_{L_2}&=O(1).
\end{align*}
Moreover,
\[
 |Q^{\rho,c}_{h,w}(0)-Q(0)|=O(h^2+w^{\gamma_0}),
\]
and
\[
 \|\ddot Q^{\rho,c}_{h,w}(0)-2\mathbf{H}\|=o(1).
\]

\item With $\ell_n=1+\log(n_L/h)$, define, for $q=0,1,2$,
\[
 \kappa^\rho_{q,n}
 =h+w^{\gamma_*/2}
 +\sqrt{\ell_n/n_2}+\sqrt{\ell_n/n_1}
 +\sqrt{\ell_n/(n_2n_1h^{2q+1})},
\]
and, for $q=0,1$,
\[
 \lambda^\rho_{q,n}
 =n_2^{-1}+n_1^{-1}
 +\frac1{n_2\sqrt{n_1h^{2q+1}}}
 +\frac1{n_1\sqrt{n_2h^{2q+1}}}
 +\frac1{n_2n_1h^{q+1}}.
\]
Then, for $q=0,1,2$,
\[
 \sup_{\boldsymbol{\vartheta}\in\Theta_\delta}
 \|\partial_{\boldsymbol{\vartheta}}^q\widehat m^{\rho,c}_{\boldsymbol{\vartheta}}
 -\partial_{\boldsymbol{\vartheta}}^qm^{\rho,c}_{\boldsymbol{\vartheta}}\|_{L_2(P_{21})}
 =O_p(\kappa^\rho_{q,n}).
\]
The same rate holds for the finite collection of once- and twice-deleted
fits. Additional deletion of one treatment row and one control column
changes derivative order $q=0,1$ by
$O_p(\lambda^\rho_{q,n})$ in random-evaluation $L_2$.
\end{enumerate}

\begin{proof}
Each Gaussian score/bandwidth derivative has form
\(h^{-(q+1)}P_q((S_{\boldsymbol{\vartheta}}(U)-S_{\boldsymbol{\vartheta}}(v))/h)K(\cdot)\).
Squaring and changing variables yields
\(P\{\|G_{q,\boldsymbol{\vartheta},h}(U;v)\|^2\}
\le C h^{-(2q+1)}D_{\boldsymbol{\vartheta},h}(v)\); the endpoint calculation in
Lemma~\ref{lem:b1-low-density} supplies \(f_{\boldsymbol{\vartheta}}\le C D_{\boldsymbol{\vartheta},h}\).
Since $\rho_w(D)\ge c(D\vee w)$ and $f_{\boldsymbol{\vartheta}}\le CD$,
\[
 Q\left[
 \frac{P\|G_q\|^2}{\rho_w(D)^2}
 \right]
 \le Ch^{-(2q+1)}
 \int_{\mathcal S}
 \frac{D(s)f_{\boldsymbol{\vartheta}}(s)}{(D(s)\vee w)^2}\,ds
 \le Ch^{-(2q+1)},
\]
because the moving score support has uniformly bounded length.

The derivatives $\rho_w''$ and $\rho_w'''$ are supported on
$w/2<d<2w$. The endpoint geometry in Lemma~\ref{lem:b1-low-density} gives,
on this band,
\[
 q_c=O(w),\quad \partial_{\boldsymbol{\vartheta}}^j q_c=O(w),\quad
 \partial_{\boldsymbol{\vartheta}}^j d=O(w),\qquad j=1,2.
\]
Every quotient term in the exact first and second derivative formulas is
therefore $O(1)$. For example,
\[
 \frac{q_c\rho_w''(d)\dot d\dot d^\top}{\rho_w(d)^2}
 =O\left(\frac{w\cdot w^{-1}\cdot w^2}{w^2}\right)=O(1).
\]
The band has probability $O(w^{\gamma_*})$, proving (ii). On $d>2w$,
moving-index kernel expansions give level and first-derivative biases
$O(h^2)$ and $O(h)$; on $d\le2w$, endpoint compatibility and (ii) give
uniformly bounded derivatives on a set of probability $O(w^{\gamma_0})$.
Splitting these regions yields the $L_2$ bounds in (iii).

For $H=m(S_0)+\epsilon$, $E(\epsilon\mid \mathbf{V})=0$ removes the residual cross
term, leaving the integrated squared approximation and smoothing bias.
Differentiation gives
\[
 \ddot Q^{\rho,c}_{h,w}(0)
 =2E(\dot m^{\rho,c}_0\dot m_0^{\rho,c\top})
 -2E\{(H-m_0^{\rho,c})\ddot m_0^{\rho,c}\}.
\]
The first term tends to $2E(\mathbf{d}\mathbf{d}^\top)=2\mathbf{H}$; conditioning removes the
$\epsilon$ component of the second, and Cauchy--Schwarz with the bounds in
part~(iii) makes its remainder $o(1)$.

For (iv), differentiate the primitive blocks and apply the two-arm Hoeffding
decomposition. The entropy and first-projection envelopes give
$\sqrt{\ell_n/n_2}+\sqrt{\ell_n/n_1}$, while density cancellation,
decoupling, and the canonical maximal inequality give
$\sqrt{\ell_n/(n_2n_1h^{2q+1})}$.

The transition bound prevents the smooth ratio map from introducing a
power of $w^{-1}$ after random-evaluation integration. The population
approximation contributes $h+w^{\gamma_*/2}$, yielding
$\kappa^\rho_{q,n}$.

For deletion, decompose the removed row, column, and their intersection.
Density cancellation and the maximal inequality for the deleted classes give
\[
 n_2^{-1}+\{n_2\sqrt{n_1h^{2q+1}}\}^{-1},
 \qquad
 n_1^{-1}+\{n_1\sqrt{n_2h^{2q+1}}\}^{-1},
\]
and the intersection contributes
$(n_2n_1h^{q+1})^{-1}$. Summation gives
$\lambda^\rho_{q,n}$.
\end{proof}
\end{lemma}

\phantomsection\label{sec:profile-sample-scale-bridge}
To incorporate the candidate score scale, for $a\in[1/2,2]$ set
\[
 T_{\boldsymbol{\vartheta},a}(\mathbf{V})=\frac{S_{\boldsymbol{\vartheta}}(\mathbf{V})-\mu_{\boldsymbol{\vartheta}}}
 {a\sigma_{\boldsymbol{\vartheta}}},\qquad
 K_{\boldsymbol{\vartheta},a}(u,v)
 =K_{\bar h}\{T_{\boldsymbol{\vartheta},a}(u)-T_{\boldsymbol{\vartheta},a}(v)\}.
\]
Replacing $T_{\boldsymbol{\vartheta}}$ by $T_{\boldsymbol{\vartheta},a}$ defines
$m^{\rho,c}_{\boldsymbol{\vartheta},a}$ and $\widehat m^{\rho,c}_{\boldsymbol{\vartheta},a}$. Let
\[
 F_n(\boldsymbol{\vartheta},a)=\mathbb P_{21}
 [\{H-\widehat m^{\rho,c}_{\boldsymbol{\vartheta},a}(\mathbf{V})\}^2],\qquad
 Q^{\rho,c}(\boldsymbol{\vartheta},a)=P
 [\{H-m^{\rho,c}_{\boldsymbol{\vartheta},a}(\mathbf{V})\}^2].
\]
The reference criterion corresponds to $a=1$; sample standardization uses
$a=a_{\boldsymbol{\vartheta}}=\widehat\sigma_{\boldsymbol{\vartheta}}/\sigma_{\boldsymbol{\vartheta}}$.

\begin{lemma}[Uniform Local Curvature]\label{lem:b5-curvature-transfer}
Under the low-density score-geometry and kernel-process clauses of
Assumption~\ref{ass:E}, on a sufficiently small
fixed local neighbourhood,
\[
 \sup_{a\in[1/2,2]}\sup_{\|\boldsymbol{\vartheta}\|\le\delta}
 \|F_{n,\boldsymbol{\vartheta}\boldsymbol{\vartheta}}(\boldsymbol{\vartheta},a)
   -Q^{\rho,c}_{\boldsymbol{\vartheta}\boldsymbol{\vartheta}}(\boldsymbol{\vartheta},a)\|=o_p(1),
\]
and
\[
 \sup_{a\in[1/2,2]}
 \|Q^{\rho,c}_{\boldsymbol{\vartheta}\boldsymbol{\vartheta}}(0,a)-2\mathbf{H}\|=o(1).
\]
Consequently, after reducing $\delta$,
\[
 \inf_{a\in[1/2,2]}\inf_{\|\boldsymbol{\vartheta}\|\le\delta}
 \lambda_{\min}\{F_{n,\boldsymbol{\vartheta}\boldsymbol{\vartheta}}(\boldsymbol{\vartheta},a)\}
 \ge c_{\mathrm{curv}}>0
\]
with probability tending to one.

\begin{proof}
Use the exact empirical partial Hessian in
Section~\ref{sec:idx-profile-estimator}, holding $a$ fixed.
Lemma~\ref{lem:b4-profile-control} controls the first-derivative products
and residual--second-derivative term through its $q=0,1,2$ bounds; in
particular, $\kappa^\rho_{2,n}=o(1)$. Its transition bound handles terms
containing $\rho_w''$ or $\rho_w'''$, and the uniform law handles the
remaining derivative classes.
The scalar $a$ adds one compact finite-dimensional coordinate and preserves
the derivative-class moment and entropy orders; the $q=2$ canonical term is
$o(1)$ under \eqref{eq:admissible-rate-region}. Lemma~\ref{lem:b4-profile-control}
is uniform over bounded positive score rescalings, so
$\sup_a\|Q^{\rho,c}_{\boldsymbol{\vartheta}\boldsymbol{\vartheta}}(0,a)-2\mathbf{H}\|=o(1)$.
Positive definiteness of $\mathbf{H}$ and population-Hessian
equicontinuity yield uniform positive curvature on a sufficiently small
local neighbourhood; uniform empirical convergence gives the final display.
\end{proof}
\end{lemma}
\begin{lemma}[Sample Standardization and Profile-Criterion Transfer]\label{lem:b6-candidate-scale}
\begin{enumerate}[label=(\roman*)]
\item For a candidate $\boldsymbol{\beta}$, let $\widehat\mu_{\boldsymbol{\beta}}$ and
$\widehat\sigma_{\boldsymbol{\beta}}>0$ be the labelled pair-score mean and standard
deviation, and set
\[
 z=\frac{s-\widehat\mu_{\boldsymbol{\beta}}}{\widehat\sigma_{\boldsymbol{\beta}}},\qquad
 Z=\frac{S-\widehat\mu_{\boldsymbol{\beta}}}{\widehat\sigma_{\boldsymbol{\beta}}},\qquad
 h_{\boldsymbol{\beta}}=\widehat\sigma_{\boldsymbol{\beta}}\bar h.
\]
For the Gaussian kernel,
\[
 K_{h_{\boldsymbol{\beta}}}(S-s)
 =\widehat\sigma_{\boldsymbol{\beta}}^{-1}K_{\bar h}(Z-z),
\]
and hence
\[
 \bigl(\widehat d^S_{\boldsymbol{\beta},h_{\boldsymbol{\beta}}},
       \widehat g^S_{\boldsymbol{\beta},h_{\boldsymbol{\beta}}},
       \widehat q^S_{\boldsymbol{\beta},c,h_{\boldsymbol{\beta}}}\bigr)(s)
 =\widehat\sigma_{\boldsymbol{\beta}}^{-1}
 \bigl(\widehat d^{\mathrm{std}}_{\boldsymbol{\beta},\bar h},
       \widehat g^{\mathrm{std}}_{\boldsymbol{\beta},\bar h},
       \widehat q^{\mathrm{std}}_{\boldsymbol{\beta},c,\bar h}\bigr)(z).
\]
With $w^S_{\boldsymbol{\beta}}=\bar w/\widehat\sigma_{\boldsymbol{\beta}}$ and
\[
 \rho^S_{\boldsymbol{\beta}}(d)=\widehat\sigma_{\boldsymbol{\beta}}^{-1}
 \rho_{\bar w}(\widehat\sigma_{\boldsymbol{\beta}} d),
\]
the pointwise coordinate identities are
\[
 c_H+\frac{\widehat q^S_{\boldsymbol{\beta},c,h_{\boldsymbol{\beta}}}(s)}
 {\widehat d^S_{\boldsymbol{\beta},h_{\boldsymbol{\beta}}}(s)\vee w^S_{\boldsymbol{\beta}}}
 =c_H+\frac{\widehat q^{\mathrm{std}}_{\boldsymbol{\beta},c,\bar h}(z)}
 {\widehat d^{\mathrm{std}}_{\boldsymbol{\beta},\bar h}(z)\vee\bar w},
\]
\[
 c_H+\frac{\widehat q^S_{\boldsymbol{\beta},c,h_{\boldsymbol{\beta}}}(s)}
 {\rho^S_{\boldsymbol{\beta}}\{\widehat d^S_{\boldsymbol{\beta},h_{\boldsymbol{\beta}}}(s)\}}
 =c_H+\frac{\widehat q^{\mathrm{std}}_{\boldsymbol{\beta},c,\bar h}(z)}
 {\rho_{\bar w}\{\widehat d^{\mathrm{std}}_{\boldsymbol{\beta},\bar h}(z)\}}.
\]

Let $\boldsymbol{\beta}_2$ and $\boldsymbol{\beta}_1$ denote the arm components, so that
\[
 \boldsymbol{\beta}=(\boldsymbol{\beta}_2^\top,\boldsymbol{\beta}_1^\top)^\top,
\]
and write
\[
 S_{ij}(\boldsymbol{\beta})=\mathbf{X}_{2i}^\top\boldsymbol{\beta}_2+\mathbf{X}_{1j}^\top\boldsymbol{\beta}_1.
\]
If $\widehat\sigma_{\boldsymbol{\beta}}^2$ is the variance of the $n_2n_1$ labelled pair
scores computed with divisor $n_2n_1$, then, exactly,
\[
 \widehat\sigma_{\boldsymbol{\beta}}^2
 =\frac1{n_2}\sum_{i=1}^{n_2}
   (\mathbf{X}_{2i}^\top\boldsymbol{\beta}_2-\overline{\mathbf{X}}_2^\top\boldsymbol{\beta}_2)^2
 +\frac1{n_1}\sum_{j=1}^{n_1}
   (\mathbf{X}_{1j}^\top\boldsymbol{\beta}_1-\overline{\mathbf{X}}_1^\top\boldsymbol{\beta}_1)^2.
\]
Let $a_{\boldsymbol{\vartheta}}=\widehat\sigma_{\boldsymbol{\vartheta}}/\sigma_{\boldsymbol{\vartheta}}$. Under the
boundedness and sampling clauses of Assumption~\ref{ass:A}, the
single-index identification of Assumption~\ref{ass:C}, and the global
score-scale clause of Assumption~\ref{ass:E},
\[
 \sup_{\boldsymbol{\vartheta}\in\Theta_\delta}
 \left\|\partial_{\boldsymbol{\vartheta}}^j(a_{\boldsymbol{\vartheta}}-1)\right\|
 =O_p(n_L^{-1/2}),\qquad j=0,1,2.
\]
Moreover, for every fitting/evaluation pair of scores,
\[
 \frac{S_u-\widehat\mu_{\boldsymbol{\vartheta}}}{\widehat\sigma_{\boldsymbol{\vartheta}}}
 -\frac{S_v-\widehat\mu_{\boldsymbol{\vartheta}}}{\widehat\sigma_{\boldsymbol{\vartheta}}}
 =\frac{S_u-S_v}{\widehat\sigma_{\boldsymbol{\vartheta}}}.
\]
Thus the random centring mean has no separate first-order effect on the
kernel fit.

\item For the scale-extended Gaussian kernel blocks, let
$\mathcal K_{\boldsymbol{\vartheta},a}$ denote a primitive block in
population-standardized coordinates. For
$\mathcal I=\{(0,0),(1,0),(0,1),(0,2),(1,1)\}$, put
\[
 G_{q,r}=\partial_{\boldsymbol{\vartheta}}^q\partial_a^r
 \mathcal K_{\boldsymbol{\vartheta},a},
 \qquad (q,r)\in\mathcal I.
\]
Uniformly in its indices,
\[
 P\{|G_{q,r}(U;v)|^2\}
 \le C\bar h^{-(2q+1)}D_{\boldsymbol{\vartheta},a}(v).
\]
The same bound holds for the corresponding centred numerator blocks and
their once- and twice-deleted versions. With
\[
 r_{q,n}=\sqrt{\ell_n/n_2}+\sqrt{\ell_n/n_1}
 +\sqrt{\ell_n/(n_1n_2\bar h^{2q+1})},
\]
the associated scale-extended empirical derivative processes for $q=0,1$
obey, for every numerator or denominator block needed below,
\[
 \sup_{\boldsymbol{\vartheta},a,v}
 \left|\{\mathbb P_{21}-P\}G_{q,r}(\boldsymbol{\vartheta},a;v)\right|
 =O_p(r_{q,n}),
\]
where $(q,r)\in\mathcal I$. Under
\eqref{eq:admissible-rate-region}, $r_{1,n}/\bar w\to0$.
Consequently their errors are $o_p(\bar w)$ for $q=0,1$.
\item Uniformly for
$\boldsymbol{\vartheta}\in\Theta_\delta$ and $a\in[1/2,2]$,
\[
 \|\widehat m_a^{\rho,c}\|_{L_2(\mathbb P_{21})}=o_p(1),\qquad
 \|\widehat m_{\boldsymbol{\vartheta} a}^{\rho,c}\|_{L_2(\mathbb P_{21})}=o_p(1),
\]
and
\[
 \|\widehat m_{aa}^{\rho,c}\|_{L_2(\mathbb P_{21})}=O_p(1).
\]
Accordingly,
\[
 \sup_{\boldsymbol{\vartheta},a}|F_{n,a}(\boldsymbol{\vartheta},a)|=o_p(1),\qquad
 \sup_{\boldsymbol{\vartheta},a}\|F_{n,\boldsymbol{\vartheta} a}(\boldsymbol{\vartheta},a)\|=o_p(1),
\]
and
\[
 \sup_{\boldsymbol{\vartheta},a}|F_{n,aa}(\boldsymbol{\vartheta},a)|=O_p(1).
\]
\item For the sample-standardized criterion
$F_n^{\mathrm{ss}}(\boldsymbol{\vartheta})=F_n(\boldsymbol{\vartheta},a_{\boldsymbol{\vartheta}})$,
\[
 \nabla F_n^{\mathrm{ss}}(0)-F_{n,\boldsymbol{\vartheta}}(0,1)
 =o_p(n_L^{-1/2}),
\]
and, uniformly along every consistent local sequence
$\widetilde{\boldsymbol{\vartheta}}=o_p(1)$,
\[
 \nabla^2F_n^{\mathrm{ss}}(\widetilde{\boldsymbol{\vartheta}})=2\mathbf{H}+o_p(1).
\]
\item Let $\widehat{\boldsymbol{\beta}}^G$ be the exact global minimizer after it has
entered the locally strongly convex neighbourhood and let $\widetilde{\boldsymbol{\beta}}$ be a measurable
numerical solution in the same locally strongly convex basin. If either
\[
 F_n^{\mathrm{ss}}(\widetilde{\boldsymbol{\vartheta}})
 -F_n^{\mathrm{ss}}(\widehat{\boldsymbol{\vartheta}}^G)=o_p(b_n^2)
\]
or the projected profile score at $\widetilde{\boldsymbol{\vartheta}}$ is $o_p(b_n)$, then
\[
 \|\widetilde{\boldsymbol{\beta}}-\widehat{\boldsymbol{\beta}}^G\|=o_p(b_n).
\]
\end{enumerate}
\begin{proof}
The Gaussian scaling and centred-score identities give the two pointwise
coordinate identities. Writing the centred pair score as $A_i+B_j$, with
$\sum_iA_i=\sum_jB_j=0$, gives
\[
 \widehat\sigma_{\boldsymbol{\vartheta}}^2
 =\boldsymbol{\beta}_2(\boldsymbol{\vartheta})^\top\widehat{\boldsymbol{\Sigma}}_2\boldsymbol{\beta}_2(\boldsymbol{\vartheta})
  +\boldsymbol{\beta}_1(\boldsymbol{\vartheta})^\top\widehat{\boldsymbol{\Sigma}}_1\boldsymbol{\beta}_1(\boldsymbol{\vartheta}).
\]
Uniformly on the local neighbourhood, the difference between this empirical
quadratic form and its population counterpart, together with the corresponding
first two derivative differences, is
\[
 O_p(n_1^{-1/2}+n_2^{-1/2})=O_p(n_L^{-1/2}).
\]
Nondegeneracy of $\sigma_{\boldsymbol{\vartheta}}$ and the square-root chain rule then give
the stated scale-ratio rate. Set $a_0=a_{\boldsymbol{\vartheta}}|_{\boldsymbol{\vartheta}=0}$.

Put
\[
 u=\{T_{\boldsymbol{\vartheta},a}(U)-T_{\boldsymbol{\vartheta},a}(v)\}/\bar h.
\]
Since $\partial_a u=-u/a$, direct differentiation gives
\[
 \partial_a\{\bar h^{-1}K(u)\}
 =-a^{-1}\bar h^{-1}uK'(u),
\]
\[
 \partial_{aa}\{\bar h^{-1}K(u)\}
 =a^{-2}\bar h^{-1}\{2uK'(u)+u^2K''(u)\}.
\]
The $a$-derivatives are polynomial--Gaussian factors and introduce no
additional inverse power of $\bar h$; the technical entropy/envelope
conditions give the stated moment and process bounds uniformly over indices
and deletions, with (R) implying $r_{1,n}/\bar w\to0$. On the empirical
transition band, the resulting $o_p(\bar w)$ primitive errors imply
$d\asymp\bar w$, and endpoint compatibility gives the corresponding
$O_p(\bar w)$ derivative blocks. For example, with
$R=\rho_{\bar w}(\widehat d)$,
\[
 \frac{\widehat q\rho''(\widehat d)
       \widehat d_{\boldsymbol{\vartheta}}\widehat d_a}{R^2}
 =O_p(\bar w\bar w^{-1}\bar w^2/\bar w^2)=O_p(1).
\]
Lemma~\ref{lem:b4-profile-control} supplies the transition-band,
ordinary-ratio, below-floor, and deletion bounds for these blocks, giving
\[
 \|\widehat m_a-m_a\|_{L_2(\mathbb P_{21})}=O_p(\kappa^\rho_{0,n}),
 \qquad
 \|\widehat m_{\boldsymbol{\vartheta} a}-m_{\boldsymbol{\vartheta} a}\|_{L_2(\mathbb P_{21})}
 =O_p(\kappa^\rho_{1,n}),
\]
and \(\|\widehat m_{aa}\|_{L_2(\mathbb P_{21})}=O_p(1)\). Positive score
rescaling leaves the ideal profile invariant, and Lemma~\ref{lem:b4-profile-control}
gives, uniformly in $a$,
\[
 \|m_a^{\rho,c}\|_{L_2}=O(\bar h^2+\bar w^{\gamma_*/2}),\qquad
 \|m_{\boldsymbol{\vartheta} a}^{\rho,c}\|_{L_2}
 =O(\bar h+\bar w^{\gamma_*/2}),
\]
and $\|m_{aa}^{\rho,c}\|_{L_2}=O(1)$. Boundedness and empirical
Cauchy--Schwarz in the derivative formulas for
$F_n=\mathbb P_{21}(H-\widehat m)^2$ give part (iii).

By the score-coordinate identities in Lemma~\ref{lem:b6-candidate-scale}(i),
$F_n(\boldsymbol{\vartheta},a_{\boldsymbol{\vartheta}})$ is the sample-standardized criterion and
$F_n(\boldsymbol{\vartheta},1)$ its population-scale reference.
The total score satisfies
\[
 \nabla F_n^{\mathrm{ss}}(\boldsymbol{\vartheta})
 =F_{n,\boldsymbol{\vartheta}}(\boldsymbol{\vartheta},a_{\boldsymbol{\vartheta}})
  +F_{n,a}(\boldsymbol{\vartheta},a_{\boldsymbol{\vartheta}})\dot a_{\boldsymbol{\vartheta}}.
\]
At zero, integration in $a$ and parts (i)--(iii) give
\[
 F_{n,\boldsymbol{\vartheta}}(0,a_0)-F_{n,\boldsymbol{\vartheta}}(0,1)
 =(a_0-1)\int_0^1
 F_{n,\boldsymbol{\vartheta} a}\{0,1+t(a_0-1)\}\,dt
 =o_p(n_L^{-1/2}),
\]
and $F_{n,a}(0,a_0)\dot a_0=o_p(n_L^{-1/2})$.

The exact total Hessian is
\begin{align*}
 \nabla^2F_n^{\mathrm{ss}}
 ={}&F_{n,\boldsymbol{\vartheta}\boldsymbol{\vartheta}}
 +F_{n,\boldsymbol{\vartheta} a}\dot a^\top
 +\dot aF_{n,a\boldsymbol{\vartheta}}
 +F_{n,aa}\dot a\dot a^\top
 +F_{n,a}\ddot a,
\end{align*}
with all terms evaluated at $(\boldsymbol{\vartheta},a_{\boldsymbol{\vartheta}})$. By parts (i)--(iii)
and Lemma~\ref{lem:b5-curvature-transfer}, the leading term is
$2\mathbf{H}+o_p(1)$, while the remaining four terms are $o_p(1)$.

The curvature bound in part (iv) gives, with probability tending to one,
\[
 \frac{c_{\mathrm{curv}}}{2}
 \|\widetilde{\boldsymbol{\vartheta}}-\widehat{\boldsymbol{\vartheta}}^G\|^2
 \le F_n^{\mathrm{ss}}(\widetilde{\boldsymbol{\vartheta}})
      -F_n^{\mathrm{ss}}(\widehat{\boldsymbol{\vartheta}}^G),
\]
and the projected-score alternative follows from the mean-value expansion and
uniform Hessian invertibility. The local parametrisation transfers both conclusions
to $\boldsymbol{\beta}$.
\end{proof}
\end{lemma}

\subsection{Hard-Floor Estimation and Target Approximation}\label{app:B:hard-floor}
\phantomsection\label{sec:idx-activation}
For the point estimator, let
\[
 T^{H,c}_{h,w}(\boldsymbol{\vartheta})
 =Q\left\{c_H+\frac{q_{\boldsymbol{\vartheta},c}(\mathbf{V}^M)}
 {d_{\boldsymbol{\vartheta}}(\mathbf{V}^M)\vee w}\right\}.
\]
Let $T(\boldsymbol{\vartheta})=E\{E(H\mid S_{\boldsymbol{\vartheta}})\}$. By iterated expectations,
$T(\boldsymbol{\vartheta})=EH$ for every $\boldsymbol{\vartheta}$ for which the conditional mean is
defined; this identity does not require the single-index model to be correct
away from zero.
\begin{lemma}[Projection Stability under Index and Scale Perturbations]\label{lem:b7-projection-stability}
\begin{enumerate}[label=(\roman*)]
\item \emph{Reference convergence.}
Let
\[
 \phi_{X2}(\mathbf{x}_2)=E\{m(S^q(\mathbf{x}_2,\mathbf{X}_1))\}-\Delta,\qquad
 \phi_{X1}(\mathbf{x}_1)=E\{m(S^q(\mathbf{X}_2,\mathbf{x}_1))\}-\Delta,
\]
and, with $\epsilon=H-m(S)=H_c-\{m(S)-c_H\}$,
\[
 \phi_{Y2}(u_2)=E(\epsilon\mid O_2=u_2),\qquad
 \phi_{Y1}(u_1)=E(\epsilon\mid O_1=u_1).
\]
Under the low-density score-geometry clause of Assumption~\ref{ass:E},
\[
 \|\phi_{Xr,h,w}-\phi_{Xr}\|_2\to0,\qquad
 \|\phi_{Yr,h,w}-\phi_{Yr}\|_2\to0,\qquad r=1,2.
\]

\item \emph{Beta-direction stability.} A change in floor status satisfies
\[
 P\left[
 \mathbf 1\{d_{\boldsymbol{\vartheta},h}(S_{\boldsymbol{\vartheta}})\le w\}
 \ne\mathbf 1\{d_{0,h}(S_0)\le w\}
 \right]
 \le Crw^{\gamma_*-1}+o(rw^{\gamma_*-1}),
\]
uniformly for $\|\boldsymbol{\vartheta}\|\le r$. Moreover, uniformly for
$\|\boldsymbol{\vartheta}\|\le r$,
\[
 \|\phi_{\rho,h,w,\boldsymbol{\vartheta}}-\phi_{\rho,h,w,0}\|_{L_2}
 \le C\left[r+
 \{rw^{\gamma_*-1}\}^{1/2}+h^2+w^{\gamma_*/2}\right]
 =:\omega_n(r),
\]
for each $\rho\in\{X2,X1,Y2,Y1\}$. In particular,
$\omega_n(Cb_n)=o(1)$.

\item \emph{Beta-direction empirical equicontinuity.}
Let $L_n^H(\boldsymbol{\vartheta})$ be the sum of the four centred original-arm empirical
averages in the oracle hard-floor expansion. Then, for each fixed $C$,
\[
 \sup_{\|\boldsymbol{\vartheta}\|\le Cb_n}
 |L_n^H(\boldsymbol{\vartheta})-L_n^H(0)|
 =o_p(\nu^{-1/2}).
\]

\item \emph{Scale-perturbation stability.}
For a candidate local coordinate $\boldsymbol{\vartheta}$ and a positive scale multiplier
$\lambda$, let $\phi_{\rho,n}(\boldsymbol{\vartheta},\lambda)$ denote the finite-$(h,w)$
centred first projection for role
$\rho\in\{X2,X1,Y2,Y1\}$ obtained from the exact hard-floor four-component
decomposition on the raw-score path
\[
 h_{\boldsymbol{\vartheta},\lambda}=\lambda h_{\boldsymbol{\vartheta},1},\qquad
 w_{\boldsymbol{\vartheta},\lambda}=w_{\boldsymbol{\vartheta},1}/\lambda,
 \qquad
 h_{\boldsymbol{\vartheta},1}=\sigma_{\boldsymbol{\vartheta}}\bar h,\qquad
 w_{\boldsymbol{\vartheta},1}=\bar w/\sigma_{\boldsymbol{\vartheta}}.
\]
Here $\lambda=1$ is the population-SD reference scale and
$\lambda=a_{\boldsymbol{\vartheta}}=\widehat\sigma_{\boldsymbol{\vartheta}}/\sigma_{\boldsymbol{\vartheta}}$ is the
candidate-dependent sample scale. Let $P_{\rho,N_\rho}$ and $P_\rho$ be the empirical
and population laws of the original observations used by role $\rho$, and put
$b_n=n_L^{-1/2}$. Uniformly on the set below,
\[
 P\left[
 \mathbf 1\{d_{\boldsymbol{\vartheta},\lambda}\le w_{\boldsymbol{\vartheta},\lambda}\}
 \ne
 \mathbf 1\{d_{\boldsymbol{\vartheta},1}\le w_{\boldsymbol{\vartheta},1}\}
 \right]
 \le C|\lambda-1|w^{\gamma_*-1},
\]
and, for every role $\rho$,
\[
 \|\phi_{\rho,n}(\boldsymbol{\vartheta},\lambda)
       -\phi_{\rho,n}(\boldsymbol{\vartheta},1)\|_2
 \le C\left[|\lambda-1|
 +\{|\lambda-1|w^{\gamma_*-1}\}^{1/2}
 +h^2+w^{\gamma_*/2}\right]
 =:\omega_n^{\mathrm{sc}}(|\lambda-1|).
\]
In particular, $\omega_n^{\mathrm{sc}}(C b_n)\to0$ and
$\sqrt{\ell_n}\,\omega_n^{\mathrm{sc}}(C b_n)\to0$. Moreover, for every fixed
$C_{\boldsymbol{\vartheta}},C_\lambda<\infty$,
\[
 \sup_{\substack{\|\boldsymbol{\vartheta}\|\le C_{\boldsymbol{\vartheta}} b_n\\
                   |\lambda-1|\le C_\lambda b_n}}
 \left|
 (P_{\rho,N_\rho}-P_\rho)
 \{\phi_{\rho,n}(\boldsymbol{\vartheta},\lambda)
       -\phi_{\rho,n}(\boldsymbol{\vartheta},1)\}
 \right|
 =o_p(\nu^{-1/2}).
\]
\end{enumerate}
\begin{proof}
The population floored fit is uniformly bounded and converges to $m$ on
compact interior subintervals; the endpoint complement has vanishing
probability by Lemma~\ref{lem:b1-low-density}. Dominated convergence and
conditional-expectation contraction give the two outer limits. On the
interior, $f/d_h\to1$ and Gaussian approximate-identity arguments give the
labelled limit $H(u_2,u_1)-m\{S(u_2,u_1)\}$; Lemma~\ref{lem:b1-low-density}(v)
controls the endpoint complement.

For $\|\boldsymbol{\vartheta}\|\le r$, a floor crossing implies
\[
 |d_{0,h}(S_0)-w|
 \le |d_{\boldsymbol{\vartheta},h}(S_{\boldsymbol{\vartheta}})-d_{0,h}(S_0)|.
\]
The moving-denominator Lipschitz condition and the margin bound in
Lemma~\ref{lem:b1-low-density} therefore give the crossing probability in
part (ii). Density-cancelled differentiation gives an $O(r)$ increment on
the common high-density region, while the below-floor and crossing
regions contribute $O(h^2+w^{\gamma_*/2})$ and
$O\{(rw^{\gamma_*-1})^{1/2}\}$ in $L_2$, respectively. The local VC
maximal inequality then yields, for each armwise role,
\[
 O_p\left\{
 \frac{\sqrt{\ell_n}\,\omega_n(Cb_n)}{\sqrt{N_\rho}}
 \right\}=o_p(\nu^{-1/2}),
\]
because $N_\rho\ge cn_L$; summing the four roles proves part (iii).

For the scale path, the bounds in Lemma~\ref{lem:b6-candidate-scale} and the
endpoint-compatible score derivatives give, uniformly on the displayed local
set,
\[
 \sup_v|d_{\boldsymbol{\vartheta},\lambda}(v)-d_{\boldsymbol{\vartheta},1}(v)|
 \le C|\lambda-1|,
 \qquad
 |w_{\boldsymbol{\vartheta},\lambda}-w_{\boldsymbol{\vartheta},1}|
 \le Cw|\lambda-1|.
\]
Since $b_n/w\to0$ by \eqref{eq:admissible-rate-region}, Lemma~\ref{lem:b1-low-density}'s
margin bound yields
\[
 P\left[
 \mathbf 1\{d_{\boldsymbol{\vartheta},\lambda}\le w_{\boldsymbol{\vartheta},\lambda}\}
 \ne
 \mathbf 1\{d_{\boldsymbol{\vartheta},1}\le w_{\boldsymbol{\vartheta},1}\}
 \right]
 \le C|\lambda-1|w^{\gamma_*-1}.
\]

The noncrossing, crossing, and endpoint regions contribute, respectively,
$O(|\lambda-1|)$, $O\{(|\lambda-1|w^{\gamma_*-1})^{1/2}\}$, and
$O(h^2+w^{\gamma_*/2})$ in $L_2$, yielding the modulus in part (iv); (R)
gives $\sqrt{\ell_n}\,\omega_n^{\mathrm{sc}}(C_\lambda b_n)\to0$.

Let
\[
 \mathcal G_{\rho,n}=\left\{
 \phi_{\rho,n}(\boldsymbol{\vartheta},\lambda)
 -\phi_{\rho,n}(\boldsymbol{\vartheta},1):
 \|\boldsymbol{\vartheta}\|\le C_{\boldsymbol{\vartheta}} b_n,
 |\lambda-1|\le C_\lambda b_n\right\}.
\]
The bounded scalar scale coordinate and one-dimensional floor threshold preserve
the finite-dimensional/VC-type complexity of the parent classes in
Assumption~\ref{ass:E}; conditional-expectation contraction transfers their
entropy and $L_2$ modulus to the projection increments. With
$\mathbb G_{\rho,N_\rho}=\sqrt{N_\rho}(P_{\rho,N_\rho}-P_\rho)$, the local
maximal inequality gives
\[
 \sup_{g\in\mathcal G_{\rho,n}}
 |\mathbb G_{\rho,N_\rho}g|
 =O_p\left\{
 \sqrt{\ell_n}\,\omega_n^{\mathrm{sc}}(C_\lambda b_n)
 +\frac{\ell_n}{\sqrt{N_\rho}}
 \right\}
 =o_p(1).
\]
Dividing by $\sqrt{N_\rho}$ and summing the four roles gives the final
$o_p(\nu^{-1/2})$ bound.
\end{proof}

\end{lemma}

\begin{lemma}[Uniform Oracle Expansion for the Hard-Floor Estimator]\label{lem:b8-local-hard-floor}
\begin{enumerate}[label=(\roman*),leftmargin=*]
\item \emph{Population flatness.}
Under the relevant low-density score-geometry, kernel-process, and global-profile
clauses of Assumption~\ref{ass:E},
\[
 \sup_{\boldsymbol{\vartheta}\in\Theta_\delta}
 |T^{H,c}_{h,w}(\boldsymbol{\vartheta})-EH|
 \le C(h^2+w^{\gamma_*}).
\]
Therefore, for any random $\widetilde{\boldsymbol{\vartheta}}=O_p(b_n)$,
\[
 T^{H,c}_{h,w}(\widetilde{\boldsymbol{\vartheta}})-T^{H,c}_{h,w}(0)
 =o_p(\nu^{-1/2}).
\]

\item \emph{Local-uniform oracle remainder.}
Under the relevant low-density score-geometry, kernel-process, and global-profile
clauses of Assumption~\ref{ass:E}, for every fixed
$C<\infty$,
\[
 \sup_{\boldsymbol{\vartheta}\in\Theta_n(C)}
 \left|
 \widehat T^{H,c}_{h,w}(\boldsymbol{\vartheta})-T^{H,c}_{h,w}(\boldsymbol{\vartheta})
 -L_n^H(\boldsymbol{\vartheta})
 \right|
 =o_p(\nu^{-1/2}).
\]

\end{enumerate}

\begin{proof}
For each fixed $\boldsymbol{\vartheta}$, the unfloored conditional profile integrates to
$EH$. Uniform endpoint-compatible kernel bias contributes $O(h^2)$, while
the hard floor contributes $O(w^{\gamma_*})$ on the low-density set. Uniform
local conditions give the bound in (i), and (R) makes both terms
$o(\nu^{-1/2})$.

The exact four-component decomposition holds pointwise on \(\Theta_n(C)\).
Lemma~\ref{lem:b2-floor-expansion}, together with the uniform inverse-moment
and margin consequences in Lemma~\ref{lem:b1-low-density} and the primitive
process classes in the kernel-process clause of Assumption~\ref{ass:E}, gives
\[
 \sup_{\Theta_n(C)}\mathbb Q_M|R_{\boldsymbol{\vartheta},h,w}|
 =O_p(a_n^2w^{\gamma_*-2})=o_p(\nu^{-1/2}).
\]
For the bounded outer and integrated labelled kernels,
Lemma~\ref{lem:b3-collision-canonical}, with the primitive uniform process
classes in the kernel-process clause of Assumption~\ref{ass:E}, gives
\[
 \sup_{\Theta_n(C)}(|R_X^d|+|R_Y^d|)
 =O_p\!\left[\sqrt{\ell_n}
 \{(M_2M_1)^{-1/2}+(n_2n_1)^{-1/2}\}\right]
 =o_p(\nu^{-1/2}).
\]
The shared-index-overlap bounds of Lemma~\ref{lem:b3-collision-canonical},
combined with the same primitive entropy/envelope conditions, give
\[
 \max_{j\le9}\sup_{\Theta_n(C)}|R^{\mathrm{lin}}_{4,j,\boldsymbol{\vartheta}}|
 =o_p(\nu^{-1/2}).
\]
Summing these three remainder bounds proves the expansion.
\end{proof}
\end{lemma}

\begin{lemma}[Approximation to the Target Functional]\label{lem:b9-target-transfer}
\begin{enumerate}[label=(\roman*),leftmargin=*]
\item \emph{Integrated smoothing bias.}
Under the low-density and target-transfer conditions in
Assumptions~\ref{ass:E} and~\ref{ass:D},
\[
 |\Delta_h-\Delta|\le Ch^2.
\]
The bound is uniform for $h'$ in any local multiplicative neighbourhood
$|h'/h-1|\le c<1$.

\item \emph{Floor and smoothing transfer.}
Let
\[
 \Delta^c_{h,w}
 =
 E\left\{
 c_H+\frac{g_h(S)-c_Hd_h(S)}{d_h(S)\vee w}
 \right\}.
\]
Under the true-index low-density and scientific-target transfer assumptions,
\[
 |\Delta^c_{h,w}-\Delta|
 \le C\{w^{\gamma_0}+h^2\}.
\]
If
\[
 h\asymp n_{\min}^{-\alpha},\quad
 w\asymp n_L^{-\eta},\quad
 \alpha>1/4,\quad \eta\gamma_0>1/2,
\]
then
\[
 \Delta^c_{h,w}-\Delta=o(\nu^{-1/2}).
\]

\end{enumerate}

\begin{proof}
Under Assumptions~\ref{ass:E} and~\ref{ass:D}, uniformly for sufficiently small $h$,
\[
 d_h(s)\asymp \{r(s)+h\}^{p-1}
\]
on fixed endpoint neighbourhoods, and
\[
 |d_h'(s)|\le C\{r(s)+h\}^{p-2}.
\]
On the remaining compact interior, $d_h$ is bounded away from zero and
$d_h'$ is bounded. Since the zero extension of $f$ is absolutely
continuous and $f(a)=f(b)=0$, convolution differentiation gives
\[
 d_h'=K_h*f'.
\]
Near the lower endpoint, split the convolution into a fixed lower-endpoint
neighbourhood and its complement. On the first part,
$|f'(a+x)|\le Cx^{p-2}$, and the Gaussian moment inequality gives
\[
 \int K_h(a+x-a-r)x^{p-2}\,dx\le C(r+h)^{p-2}.
\]
The complement is separated from the evaluation point by a fixed distance and is
exponentially small in $h^{-2}$, hence is absorbed by the displayed
polynomial bound. The upper endpoint is identical. Interior boundedness
follows from $f'\in L_\infty$ locally and the unit mass of $K_h$.

These endpoint bounds also give
\[
 \sup_{0<h<h_0}
 \int_a^b f(s)\frac{|d_h'(s)|}{d_h(s)}\,ds<\infty.
\]
On a fixed interior this follows from the positive lower bound for $d_h$ and
boundedness of $f$ and $d_h'$. Near an endpoint,
\[
 f(a+r)\frac{|d_h'(a+r)|}{d_h(a+r)}
 \le C\frac{r^{p-1}}{r+h}\le Cr^{p-2},
\]
which is integrable because $p\ge4$; the upper endpoint is the same.

For the standard Gaussian kernel,
\[
 \partial_sK_h(u-s)=\frac{u-s}{h^2}K_h(u-s).
\]
Differentiation under the integral therefore gives
\[
 \int_a^b K_h(u-s)(u-s)f(u)\,du=h^2d_h'(s).
\]

The exact identity is
\[
 \Delta_h-\Delta
 =
 \int_a^b\frac{f(s)}{d_h(s)}
 \int_a^bK_h(u-s)\{m(u)-m(s)\}f(u)\,du\,ds.
\]
Write $\Delta_h-\Delta=I_{1,h}+I_{2,h}$. Taylor's theorem gives
\[
 m(u)-m(s)=m'(s)(u-s)+R_m(u,s),
 \qquad |R_m(u,s)|\le C(u-s)^2.
\]
For the quadratic remainder, the true-index endpoint bound $f(s)/d_h(s)\le C$
implies
\begin{align*}
 |I_{2,h}|
 &\le C\int_a^b\int_a^b
 K_h(u-s)(u-s)^2f(u)\,du\,ds\\
 &\le Ch^2\int_a^bf(u)\,du
 =O(h^2),
\end{align*}
because the truncated inner Gaussian second moment is no larger than the
full second moment.

For the linear term,
\[
 I_{1,h}
 =h^2\int_a^bm'(s)f(s)\frac{d_h'(s)}{d_h(s)}\,ds.
\]
The integrated log-derivative bound and boundedness of $m'$ give
$|I_{1,h}|\le Ch^2$. Hence (i) holds uniformly on a fixed multiplicative
bandwidth neighbourhood. Lemma~\ref{lem:b1-low-density}'s population-floor
bound gives $|\Delta^c_{h,w}-\Delta_h|\le Cw^{\gamma_0}$; combining it with
the smoothing bound proves (ii), including its root-scale assertion.
\end{proof}
\end{lemma}

\begin{lemma}[Effect of Sample Standardization on the Hard-Floor Estimator]\label{lem:b10-hard-floor-transfer}
\textnormal{(i) Consistency level.}
For $a\in[1/2,2]$, let
\[
 \widehat m^{H,c}_{\boldsymbol{\vartheta},a}(v)
 =c_H+
 \frac{\widehat q_{\boldsymbol{\vartheta},a,c}(v)}
      {\widehat d_{\boldsymbol{\vartheta},a}(v)\vee\bar w},
 \qquad
 \widehat T^{H,c}_{h,w}(\boldsymbol{\vartheta},a)
 =\mathbb Q_M\widehat m^{H,c}_{\boldsymbol{\vartheta},a}(\mathbf{V}^M).
\]
Here $a=1$ is the population-standard-deviation reference estimator of
Lemma~\ref{lem:b8-local-hard-floor},
whereas $a=a_{\boldsymbol{\vartheta}}=\widehat\sigma_{\boldsymbol{\vartheta}}/\sigma_{\boldsymbol{\vartheta}}$ gives the
hard-floor estimator with candidate-specific sample standardization.
For every fixed $C<\infty$,
\[
 \sup_{\|\boldsymbol{\vartheta}\|\le Cb_n}
 \left|
 \widehat T^{H,c}_{h,w}(\boldsymbol{\vartheta},a_{\boldsymbol{\vartheta}})
 -\widehat T^{H,c}_{h,w}(\boldsymbol{\vartheta},1)
 \right|=o_p(1),
 \qquad b_n=n_L^{-1/2}.
\]
Part (i) gives consistency; parts (ii)--(iii) give the root-scale transfers.
\begin{proof}[Proof of part (i)]
Lemma~\ref{lem:b6-candidate-scale} gives
\[
 \sup_{\boldsymbol{\vartheta}\in\Theta_\delta}|a_{\boldsymbol{\vartheta}}-1|=O_p(b_n).
\]
On the event that this supremum is at most $1/2$, the mean-value path lies in
$[1/2,2]$, and the $(q=0,r=1)$ scale-derivative bound in
Lemma~\ref{lem:b6-candidate-scale} gives
\[
 \sup_{\boldsymbol{\vartheta}\in\Theta_\delta,\,a\in[1/2,2],\,v}
 \left\{
 |\partial_a\widehat q_{\boldsymbol{\vartheta},a,c}(v)|
 +|\partial_a\widehat d_{\boldsymbol{\vartheta},a}(v)|
 \right\}=O_p(1).
\]
The mean-value theorem yields, uniformly for
$\|\boldsymbol{\vartheta}\|\le Cb_n$,
\[
 \sup_v\left\{
 |\widehat q_{\boldsymbol{\vartheta},a_{\boldsymbol{\vartheta}},c}(v)
       -\widehat q_{\boldsymbol{\vartheta},1,c}(v)|
 +|\widehat d_{\boldsymbol{\vartheta},a_{\boldsymbol{\vartheta}}}(v)
       -\widehat d_{\boldsymbol{\vartheta},1}(v)|
 \right\}=O_p(b_n).
\]
Because the centred mark is bounded, the elementary hard-floor inequality
\[
 \left|\frac{q_1}{d_1\vee\bar w}
       -\frac{q_2}{d_2\vee\bar w}\right|
 \le C\frac{|q_1-q_2|+|d_1-d_2|}{\bar w}
\]
holds without any lower bound on $d_1$ or $d_2$. Hence
\[
 \sup_{\|\boldsymbol{\vartheta}\|\le Cb_n}\sup_v
 \left|\widehat m^{H,c}_{\boldsymbol{\vartheta},a_{\boldsymbol{\vartheta}}}(v)
       -\widehat m^{H,c}_{\boldsymbol{\vartheta},1}(v)\right|
 =O_p(b_n/\bar w)=o_p(1),
\]
by \eqref{eq:admissible-rate-region}; the full outer average is bounded by the
same supremum.
\end{proof}
\textnormal{(ii) Oracle first-order equivalence under sample standardization.}
Suppress the group index and let $\sigma_S$ and $\widehat\sigma_S$ denote,
respectively, the population and empirical labelled-sample standard deviations
of the pair score at the true index. Put
\[
 a_0=\widehat\sigma_S/\sigma_S,
 \qquad \bar h\asymp n_{\min}^{-\alpha},
 \qquad \bar w\asymp n_L^{-\eta},
\]
and, in population-score coordinates, define
\[
 h_0=\sigma_S\bar h,\qquad w_0=\bar w/\sigma_S,
 \qquad h_\lambda=\lambda h_0,\qquad w_\lambda=w_0/\lambda.
\]
Let $\widehat\Delta^{\mathrm{or},c}(\lambda)$ be the all-pair true-index
oracle statistic formed with the centred hard floor and the full outer average
at $(h_\lambda,w_\lambda)$, and let $\Delta^c_{h_\lambda,w_\lambda}$ be its
finite-$(h,w)$ population target. Thus $\lambda=1$ is the
population-scale reference statistic $\widehat\Delta_{h_0,w_0}^{\mathrm{or},c,\mathrm{ref}}$,
whereas $\lambda=a_0$ is exactly the Section~\ref{sec:methodology} sample-standardized
oracle estimator by Lemma~\ref{lem:b6-candidate-scale}(i).

For fixed $C<\infty$, let
\[
 \mathcal A_n(C)=\{\lambda>0:|\lambda-1|\le Cn_L^{-1/2}\}.
\]
For every $\varepsilon>0$, there exists a finite $C_\varepsilon$ such that
\[
 \limsup_{n\to\infty}P\{a_0\notin\mathcal A_n(C_\varepsilon)\}
 \le\varepsilon.
\]
For every fixed $C<\infty$, with
\[
 L_n=M_2^{-1}\sum_a\phi_{X2}(\mathbf{X}_{2a})
    +M_1^{-1}\sum_b\phi_{X1}(\mathbf{X}_{1b})
    +n_2^{-1}\sum_i\phi_{Y2}(O_{2i})
    +n_1^{-1}\sum_j\phi_{Y1}(O_{1j}),
\]
the uniform representation and target transfer are
\[
 \sup_{\lambda\in\mathcal A_n(C)}
 \left|\widehat\Delta^{\mathrm{or},c}(\lambda)
       -\Delta^c_{h_\lambda,w_\lambda}-L_n\right|
 =o_p(\nu^{-1/2}),
 \qquad
 \sup_{\lambda\in\mathcal A_n(C)}
 |\Delta^c_{h_\lambda,w_\lambda}-\Delta|
 =o(\nu^{-1/2}).
\]
Consequently,
\[
 \widehat\Delta^{\mathrm{or},c}
 -\widehat\Delta^{\mathrm{or},c,\mathrm{ref}}
 =o_p(\nu^{-1/2}).
\]
\textnormal{(iii) Local candidate root scale.}
Suppress the group index and put
$b_n=n_L^{-1/2}$. For each candidate $\boldsymbol{\vartheta}$, let
$a_{\boldsymbol{\vartheta}}=\widehat\sigma_{\boldsymbol{\vartheta}}/\sigma_{\boldsymbol{\vartheta}}$ and write
\[
 h_{\boldsymbol{\vartheta},1}=\sigma_{\boldsymbol{\vartheta}}\bar h,\qquad
 w_{\boldsymbol{\vartheta},1}=\bar w/\sigma_{\boldsymbol{\vartheta}},\qquad
 h_{\boldsymbol{\vartheta},\lambda}=\lambda h_{\boldsymbol{\vartheta},1},\qquad
 w_{\boldsymbol{\vartheta},\lambda}=w_{\boldsymbol{\vartheta},1}/\lambda.
\]
For every fixed $C<\infty$,
\[
 \sup_{\|\boldsymbol{\vartheta}\|\le Cb_n}
 \left|
 \widehat T^{H,c}_{h,w}(\boldsymbol{\vartheta},a_{\boldsymbol{\vartheta}})
 -\widehat T^{H,c}_{h,w}(\boldsymbol{\vartheta},1)
 \right|=o_p(\nu^{-1/2}).
\]
Here $\lambda=1$ is the population-SD reference statistic and
$\lambda=a_{\boldsymbol{\vartheta}}$ is the labelled-sample-SD statistic by
Lemma~\ref{lem:b6-candidate-scale}(i). Write
$T^{H,c}_{h,w}(\boldsymbol{\vartheta},\lambda)$ for the corresponding population hard-floor
average and $L_n^H(\boldsymbol{\vartheta},\lambda)$ for the sum of its four finite-$(h,w)$
first projections.
\begin{proof}[Proof of parts (ii)--(iii)]
For fixed $C_{\boldsymbol{\vartheta}},C_\lambda<\infty$, let
\[
 \mathcal J_n=\{(\boldsymbol{\vartheta},\lambda):
 \|\boldsymbol{\vartheta}\|\le C_{\boldsymbol{\vartheta}} b_n,\ |\lambda-1|\le C_\lambda b_n\}.
\]
 The scale path has ratios $h_{\boldsymbol{\vartheta},\lambda}/h_{\boldsymbol{\vartheta},1}=\lambda$
 and $w_{\boldsymbol{\vartheta},\lambda}/w_{\boldsymbol{\vartheta},1}=\lambda^{-1}$.
 Lemmas~\ref{lem:b2-floor-expansion}--\ref{lem:b3-collision-canonical} and
 Lemma~\ref{lem:b8-local-hard-floor} give, uniformly on $\mathcal J_n$,
\[
 \sup_{\mathcal J_n}
 |\widehat T^{H,c}_{h,w}(\boldsymbol{\vartheta},\lambda)
 -T^{H,c}_{h,w}(\boldsymbol{\vartheta},\lambda)-L_n^H(\boldsymbol{\vartheta},\lambda)|
 =o_p(\nu^{-1/2}).
\]

For the projections, use the exact decomposition
\[
 \begin{aligned}
 L_n^H(\boldsymbol{\vartheta},\lambda)-L_n
 ={}&[L_n^H(\boldsymbol{\vartheta},\lambda)-L_n^H(\boldsymbol{\vartheta},1)]\\
 &+[L_n^H(\boldsymbol{\vartheta},1)-L_n^H(0,1)]
   +[L_n^H(0,1)-L_n].
 \end{aligned}
\]
The three brackets are $o_p(\nu^{-1/2})$ by
Lemma~\ref{lem:b7-projection-stability}(iv), parts (ii)--(iii), and part (i),
respectively.

At $\boldsymbol{\vartheta}=0$, Lemma~\ref{lem:b9-target-transfer} and
Assumption~\ref{ass:D} give the true-index bound
$C(h_\lambda^2+w_\lambda^{\gamma_0})$. Positive score rescaling and
Lemma~\ref{lem:b8-local-hard-floor} give
\[
 \sup_{\mathcal J_n}|T^{H,c}_{h,w}(\boldsymbol{\vartheta},\lambda)-EH|
 \le C\sup_{\mathcal J_n}
 \{h_{\boldsymbol{\vartheta},\lambda}^2+w_{\boldsymbol{\vartheta},\lambda}^{\gamma_*}\}
 =o(\nu^{-1/2}).
\]
The shared representation then gives
\[
 \sup_{\mathcal J_n}
 |\widehat T^{H,c}_{h,w}(\boldsymbol{\vartheta},\lambda)-EH-L_n|
 =o_p(\nu^{-1/2}),
\]
whose restriction at the true direction gives part (ii).

For localization at the true direction, every \(\varepsilon>0\) admits, by
Lemma~\ref{lem:b6-candidate-scale}, a finite
$C_\varepsilon$ with
\(\limsup_nP\{a_0\notin\mathcal A_n(C_\varepsilon)\}\le\varepsilon\).
For \(\delta>0\),
\[
\begin{split}
 P\{\sqrt\nu|\widehat\Delta^{\mathrm{or},c}(a_0)
                 -\widehat\Delta^{\mathrm{or},c}(1)|>\delta\}
 &\le P\{a_0\notin\mathcal A_n(C_\varepsilon)\}\\
 &\quad+P\{\sqrt\nu\sup_{\mathcal A_n(C_\varepsilon)}
 |\widehat\Delta^{\mathrm{or},c}(\lambda)
       -\widehat\Delta^{\mathrm{or},c}(1)|>\delta\}.
\end{split}
\]
The second probability vanishes by the uniform representation; letting
$\varepsilon\downarrow0$ proves part (ii).

The same localization argument, uniformly over
\(\|\boldsymbol{\vartheta}\|\le Cb_n\), gives, for every \(\varepsilon>0\), a finite
$C_{\lambda,\varepsilon}$ such that
\[
 \limsup_nP\{\sup_{\Theta_\delta}|a_{\boldsymbol{\vartheta}}-1|
                  >C_{\lambda,\varepsilon}b_n\}\le\varepsilon.
\]
For \(\mathcal J_{n,\varepsilon}=\mathcal J_n(C,C_{\lambda,\varepsilon})\),
\[
\begin{split}
 P\{\sqrt\nu\sup_{\|\boldsymbol{\vartheta}\|\le Cb_n}
 |\widehat T^{H,c}(\boldsymbol{\vartheta},a_{\boldsymbol{\vartheta}})
       -\widehat T^{H,c}(\boldsymbol{\vartheta},1)|>\delta\}
 &\le P\{\sup_{\Theta_\delta}|a_{\boldsymbol{\vartheta}}-1|
                 >C_{\lambda,\varepsilon}b_n\}\\
 &\quad+P\{\sqrt\nu\sup_{\mathcal J_{n,\varepsilon}}
 |\widehat T^{H,c}(\boldsymbol{\vartheta},\lambda)
       -\widehat T^{H,c}(\boldsymbol{\vartheta},1)|>\delta\}.
\end{split}
\]
The second probability vanishes by the uniform representation, and the
localization bound proves part (iii); Lemma~\ref{lem:b6-candidate-scale}(i)
identifies the random-scale statistics with the feasible estimators.
\end{proof}
\end{lemma}

\endgroup
\makeatletter
\begingroup
\def\@seccntformat#1{Appendix~\csname the#1\endcsname.\quad}
\section{Proofs of Main Results}\label{app:C}
\endgroup
\makeatother

The proofs below use the notation and auxiliary results of Appendices~\ref{app:A}--\ref{app:B}.

\subsection{Point-estimation results}\label{app:C-point}

\paragraph{Proof of Proposition~\ref{prop:profile-consistency}}
\label{proof:profile-consistency}

\begin{proof}
For part (i), suppress $k$. Boundedness in Assumption~\ref{ass:A}, $|H_c|\le C_c$, and
the stabilizer inequality
\[
 \rho_w(d)\ge c_\rho(d\vee w),
\]
give, by nonnegativity of the kernel,
\[
 |\widehat m^{\rho,c}_{\boldsymbol{\beta}}(s)|
 \le |c_H|+C_c\frac{\widehat d_{\boldsymbol{\beta}}(s)}
 {c_\rho\{\widehat d_{\boldsymbol{\beta}}(s)\vee w\}}
 \le |c_H|+C_c/c_\rho.
\]
The squared-residual map is uniformly Lipschitz. The global
Glivenko--Cantelli/profile condition and the finite-$(h,w)$ population
approximation in Assumption~\ref{ass:E}(c) give
\begin{equation}\label{eq:proof-global-uniform}
\begin{aligned}
 \sup_{\boldsymbol{\beta}\in\mathcal B}|\widehat Q^{\rho,c}(\boldsymbol{\beta})-Q(\boldsymbol{\beta})|
 &\le \sup_{\boldsymbol{\beta}\in\mathcal B}
 |\widehat Q^{\rho,c}(\boldsymbol{\beta})-Q^{\rho,c}_{h,w}(\boldsymbol{\beta})|\\
 &\quad+\sup_{\boldsymbol{\beta}\in\mathcal B}
 |Q^{\rho,c}_{h,w}(\boldsymbol{\beta})-Q(\boldsymbol{\beta})|
 =o_p(1).
\end{aligned}
\end{equation}
The bounded-covariate, fixed-dimensional law also gives
\(\Pr\{\inf_{\boldsymbol{\beta}\in\mathcal B}\widehat\sigma_{\boldsymbol{\beta}}>c_\sigma/2\}\to1\),
so the sample-standardized criterion is well defined with probability tending
to one. For every candidate $\boldsymbol{\beta}$,
\[
 q_{\boldsymbol{\beta},c}(s)=d_{\boldsymbol{\beta}}(s)\{m_{\boldsymbol{\beta}}(s)-c_H\},
 \qquad
 Q_c(\boldsymbol{\beta})=Q(\boldsymbol{\beta}).
\]
Let $\widehat{\boldsymbol{\beta}}^G$ be an exact global minimizer. Assumption~\ref{ass:C}
and compactness give, for each $\epsilon>0$, a $c_\epsilon>0$ with
$\inf_{\|\boldsymbol{\beta}-\boldsymbol{\beta}_0\|\ge\epsilon}\{Q(\boldsymbol{\beta})-Q(\boldsymbol{\beta}_0)\}\ge c_\epsilon$.
On the event in \eqref{eq:proof-global-uniform} with error below
$c_\epsilon/3$, every exact global minimizer lies in the $\epsilon$-ball;
hence $\widehat{\boldsymbol{\beta}}^G\xrightarrow{p}\boldsymbol{\beta}_0$.

For an approximate numerical solution $\widetilde{\boldsymbol{\beta}}$ satisfying the
stated tolerance, Lemma~\ref{lem:b6-candidate-scale}(v) gives local curvature
and numerical transfer. If the objective gap is
$o_p(b_n^2)$, local strong convexity yields
\[
 c\|\widetilde{\boldsymbol{\vartheta}}-\widehat{\boldsymbol{\vartheta}}^G\|^2
 \le
 \widehat Q^{\rho,c}(\widetilde{\boldsymbol{\vartheta}})
 -\widehat Q^{\rho,c}(\widehat{\boldsymbol{\vartheta}}^G)
 =o_p(b_n^2).
\]
If instead the projected score is $o_p(b_n)$, the mean-value expansion of the
score and invertibility of the same Hessian give the identical conclusion.
Thus $\|\widetilde{\boldsymbol{\beta}}-\widehat{\boldsymbol{\beta}}^G\|=o_p(b_n)$, in particular
$\widetilde{\boldsymbol{\beta}}\xrightarrow{p}\boldsymbol{\beta}_0$.
\label{proof:beta-expansion}
For part (ii), we first derive the profile-score expansion at the population
reference scale. At $\boldsymbol{\vartheta}=0$, write
\[
 m_0(\mathbf{V})=m(S_0),\qquad \epsilon=H-m(S_0),\qquad
 \mathbf{d}(\mathbf{V})=m'(S_0)\mathbf{R}^\top\{\mathbf{V}-E(\mathbf{V}\mid S_0)\}.
\]
Equivalently,
$\epsilon=H_c-\{m(S_0)-c_H\}=H-m(S_0)$, so centring does not alter
the leading residual.
Let
\[
 \eta_0=\widehat m_0^{\rho,c}-m_0^{\rho,c},\qquad
 \eta_1=\dot{\widehat m}_0^{\rho,c}-\dot m_0^{\rho,c}.
\]
For an evaluation pair $q=(i,j)$, let superscript `$-$' denote the fit with row
$i$ and column $j$ deleted, and write
\[
 L_r=\eta_r-\eta_r^-,\qquad r=0,1.
\]
Using these deleted fits, the smoothly stabilized profile score satisfies
\begin{align}
 -\frac12\dot{\widehat Q}^{\rho,c}(0)
 ={}&\mathbb P_{21}(\epsilon \mathbf{d})+B_{h,w}
 +T_{1n}-T_{2n}-T_{3n},                                  \label{eq:idx-score-split}
\end{align}
where
\[
 |B_{h,w}|=O(h^2+w^{\gamma_*})+o_p(b_n),
\]
and, up to terms of that deterministic order,
\begin{align*}
 T_{1n}&=\mathbb P_{21}(\epsilon\eta_1^-)
          +\mathbb P_{21}(\epsilon L_1),\\
 T_{2n}&=\mathbb P_{21}(\mathbf{d}\eta_0^-)
          +\mathbb P_{21}(\mathbf{d}L_0),\\
 T_{3n}&=\mathbb P_{21}(\eta_0^-\eta_1^-)
 +\mathbb P_{21}(\eta_0^-L_1)
 +\mathbb P_{21}(L_0\eta_1^-)
 +\mathbb P_{21}(L_0L_1).
\end{align*}

To control the resulting covariance terms, let
\[
 e_{q,n}=\kappa^\rho_{q,n}+\lambda^\rho_{q,n},\qquad
 \Lambda_{q,n}=\lambda^\rho_{q,n},\qquad q=0,1.
\]
Then the two deleted averages have mean zero and
\begin{align*}
 E\|\mathbb P_{21}(\epsilon\eta_1^-)\|^2
 \le C\{&b_n^2e_{1,n}^2+\Lambda_{1,n}e_{1,n}
 +\Lambda_{1,n}^2+(n_2n_1)^{-1}e_{1,n}^2\},\\
 E\|\mathbb P_{21}(\mathbf{d}\eta_0^-)\|^2
 \le C\{&b_n^2e_{0,n}^2+\Lambda_{0,n}e_{0,n}
 +\Lambda_{0,n}^2+(n_2n_1)^{-1}e_{0,n}^2\}.
\end{align*}

Applying the two-arm Hoeffding decomposition to the leading score term, define
\[
 \boldsymbol{\chi}(O_2,O_1)=\epsilon \mathbf{d}(\mathbf{V}),\qquad
 \boldsymbol{\zeta}_2(O_2)=E(\boldsymbol{\chi}\mid O_2),\qquad
 \boldsymbol{\zeta}_1(O_1)=E(\boldsymbol{\chi}\mid O_1).
\]
Then
\[
 \dot{\widehat Q}^{\rho,c}(0)
 =-2\left\{
 \frac1{n_2}\sum_{i=1}^{n_2}\boldsymbol{\zeta}_2(O_{2i})
 +\frac1{n_1}\sum_{j=1}^{n_1}\boldsymbol{\zeta}_1(O_{1j})
 \right\}+o_p(b_n).
\]

For the deleted score, start from the exact identity
\[
 -\frac12\dot{\widehat Q}^{\rho,c}(0)
 =\mathbb P_{21}\{(H-\widehat m_0^{\rho,c})
                    \dot{\widehat m}_0^{\rho,c}\}.
\]
Add and subtract $m(S_0)$, $m_0^{\rho,c}$, $\mathbf{d}$, and
$\dot m_0^{\rho,c}$, then substitute $\eta_r=\eta_r^-+L_r$.
Population approximation terms are absorbed in $B_{h,w}$, and conditioning on
$\mathbf{V}$ removes cross terms containing $\epsilon$ times a population function of
$\mathbf{V}$; in particular,
\[
 \|m_0^{\rho,c}-m\|_2\|\dot m_0^{\rho,c}-\mathbf{d}\|_2
 =O\{(h^2+w^{\gamma_0/2})(h+w^{\gamma_0/2})\},
\]
and the direct integrated profile drift is $O(h^2+w^{\gamma_0})$, hence also
$O(h^2+w^{\gamma_*})$. This gives the displayed decomposition.

Expand each covariance double sum into the same pair, a shared treatment
unit only, a shared control unit only, and four distinct original units. The
first three classes have normalized proportions
$O((n_2n_1)^{-1})$, $O(n_2^{-1})$, and $O(n_1^{-1})$ and are bounded by
conditional Cauchy--Schwarz.

For four distinct units, replace both fits by the common double-deleted fit.
For the $\epsilon\eta_1^-$ term, conditional mean zero follows from
$E(\epsilon\mid \mathbf{V})=0$. For the $\mathbf{d}\eta_0^-$ term, the deleted level fit is a
function of the deleted fitting array and the scalar evaluation score, whereas
$E\{\mathbf{d}(\mathbf{V})\mid S_0\}=0$. The common double-deleted product therefore has mean
zero. Replacing a double-deleted fit by a once-deleted fit costs at least one
$L_2$ difference of order $\Lambda_{q,n}$. The stated bounds follow.

Under \eqref{eq:admissible-rate-region},
\begin{align*}
 &\kappa^\rho_{2,n}=o(1),\qquad
 \lambda^\rho_{0,n}+\lambda^\rho_{1,n}=o(b_n),\\
 &\lambda^\rho_{1,n}\kappa^\rho_{1,n}=o(b_n^2),\\
 &\kappa^\rho_{0,n}\kappa^\rho_{1,n}=o(b_n),\\
 &\kappa^\rho_{0,n}\lambda^\rho_{1,n}
 +\kappa^\rho_{1,n}\lambda^\rho_{0,n}
 +\lambda^\rho_{0,n}\lambda^\rho_{1,n}=o(b_n),\\
 &h^2+w^{\gamma_*}=o(b_n).
\end{align*}

The exact two-arm Hoeffding decomposition gives
\[
 \mathbb P_{21}\boldsymbol{\chi}
 =n_2^{-1}\sum_i\boldsymbol{\zeta}_2(O_{2i})
 +n_1^{-1}\sum_j\boldsymbol{\zeta}_1(O_{1j})
 +\mathbb P_{21}\boldsymbol{\chi}^d,
\]
with $E\|\mathbb P_{21}\boldsymbol{\chi}^d\|^2=O((n_2n_1)^{-1})=o(b_n^2)$.
These bounds make the deleted linear nuisance averages, reinsertion terms,
nuisance products, and profile drift all $o_p(b_n)$; the canonical pair term is
$o_p(b_n)$ by the two-arm Hoeffding variance. Substitution into
\eqref{eq:idx-score-split} gives the stated score expansion.

Lemma~\ref{lem:b6-candidate-scale} transfers the reference-scale score
expansion and the curvature result in Lemma~\ref{lem:b5-curvature-transfer}
to the criterion standardized by the candidate-specific sample scale. Its
exact first-order condition is
\begin{equation}\label{eq:proof-beta-foc}
 0=\dot{\widehat Q}^{\rho,c}(0)
 +\ddot{\widehat Q}^{\rho,c}(\widetilde{\boldsymbol{\vartheta}})
  \widehat{\boldsymbol{\vartheta}}_G
\end{equation}
for an intermediate local point \(\widetilde{\boldsymbol{\vartheta}}\). Hence
\[
 \dot{\widehat Q}^{\rho,c}(0)
 =-2\left\{n_2^{-1}\sum_i\boldsymbol{\zeta}_2(O_{2i})
          +n_1^{-1}\sum_j\boldsymbol{\zeta}_1(O_{1j})\right\}+o_p(b_n)
\]
and
\[
 \ddot{\widehat Q}^{\rho,c}(\widetilde{\boldsymbol{\vartheta}})
 =2\mathbf{H}+o_p(1).
\]

Uniform curvature then yields the local solution expansion. Solving
\eqref{eq:proof-beta-foc} gives
\begin{equation}\label{eq:proof-theta-expansion}
 \widehat{\boldsymbol{\vartheta}}_G
 =\mathbf{H}^{-1}\left\{n_2^{-1}\sum_i\boldsymbol{\zeta}_2(O_{2i})
          +n_1^{-1}\sum_j\boldsymbol{\zeta}_1(O_{1j})\right\}+o_p(b_n).
\end{equation}
The two displayed sums are over independent original observations within
their respective arms; the canonical pair component is negligible by the
Hoeffding decomposition above. The ordinary two-arm triangular-array CLT gives the stated tangent-contrast
normality. Lemma~\ref{lem:b6-candidate-scale} transfers the same expansion
to any numerical solution satisfying the stated
objective-gap or projected-score condition.

The local-coordinate expansion transfers to the unit-sphere parameter through
\[
 \boldsymbol{\beta}(\boldsymbol{\vartheta})=\boldsymbol{\beta}_0+\mathbf{R}\boldsymbol{\vartheta}
 -\tfrac12\boldsymbol{\beta}_0\|\boldsymbol{\vartheta}\|^2+O(\|\boldsymbol{\vartheta}\|^3).
\]
Substitution of \eqref{eq:proof-theta-expansion} gives
\eqref{eq:beta-expansion}; its normal component is only \(O_p(b_n^2)\).
The leading residual is unchanged because
\(H_c-(m-c_H)=H-m\), so centring changes neither \(\boldsymbol{\zeta}_r\) nor
\(\mathbf{H}\).
\end{proof}

\paragraph{Proof of Theorem~\ref{thm:consistency}}
\label{proof:consistency}

\begin{proof}
Fix $k$ and suppress the group index in the local rates. By
Proposition~\ref{prop:profile-consistency},
$\widehat{\boldsymbol{\vartheta}}_k=O_p(b_{n,k})$. Write
$\widehat T_k(\boldsymbol{\vartheta},a)=\widehat T^{H,c}_{h_k,w_k}(\boldsymbol{\vartheta},a)$ and
$T_k(\boldsymbol{\vartheta},a)$ for its population counterpart. The feasible estimator
satisfies the single decomposition
\[
 \widehat\Delta_k-\Delta_k=A_{1k}+A_{2k}+A_{3k},
\]
where
\begin{align*}
 A_{1k}&=\widehat T_k(\widehat{\boldsymbol{\vartheta}}_k,a_{\widehat{\boldsymbol{\vartheta}}_k})
          -\widehat T_k(\widehat{\boldsymbol{\vartheta}}_k,1),\\
 A_{2k}&=\widehat T_k(\widehat{\boldsymbol{\vartheta}}_k,1)
          -T_k(\widehat{\boldsymbol{\vartheta}}_k,1),\\
 A_{3k}&=T_k(\widehat{\boldsymbol{\vartheta}}_k,1)-\Delta_k.
\end{align*}
By Lemma~\ref{lem:b10-hard-floor-transfer}(i), $A_{1k}=o_p(1)$;
the local-uniform expansion in Lemma~\ref{lem:b8-local-hard-floor}
and its bounded four-component leading term give $A_{2k}=o_p(1)$. Its population
flatness statement gives $A_{3k}=o_p(1)$, since $EH=\Delta_k$ and
$h_k^2+w_k^{\gamma_*}=o(1)$. Hence
$\widehat\Delta_k\xrightarrow{p}\Delta_k$. Fixed-$K$ deterministic aggregation gives
the same conclusion for Ave and W. The sample-size conditions in
Assumption~\ref{ass:A} also permit proportional unequal allocation on the
group-specific scales.
\end{proof}

\paragraph{Proof of Proposition~\ref{prop:oracle-decomp}}
\label{proof:oracle-decomp}

\begin{proof}
All fixed-$(h,w)$ quantities in this proof use the population-SD reference
score scale of Proposition~\ref{prop:oracle-decomp}.
Suppress $k$ and define, for a fixed evaluation score $s$,
\begin{align*}
 H^c_{h,w}(u_2,u_1;s)={}&
 \frac{K_h\{S(u_2,u_1)-s\}H_c(u_2,u_1)-q_c(s)}{T_w\{d(s)\}}\\
 &-J(s)\frac{q_c(s)
  [K_h\{S(u_2,u_1)-s\}-d(s)]}{d(s)^2}.
\end{align*}
It is centred under the product law of the fitting pair. The add-back
constant $c_H$ has no stochastic part. The exact fit identity in
Lemma~\ref{lem:b2-floor-expansion} is, for each evaluation score,
\[
 \widehat m^c_{h,w}(s)-m^c_{h,w}(s)
 =(\mathbb P_{21}-P)H^c_{h,w}(\cdot;s)+R^{ND,c}_{h,w}(s),
\]
where $R^{ND,c}_{h,w}$ denotes the exact nonlinear remainder in that
pointwise identity. Consequently,
\begin{equation}\label{eq:proof-oracle-start}
 \widehat\Delta^{\mathrm{or},c,\mathrm{ref}}_{h,w}-\Delta^c_{h,w}
 =(\mathbb Q_M-Q)m^c_{h,w}
 +\mathbb Q_M(\mathbb P_{21}-P)H^c_{h,w}
 +\mathbb Q_MR^{ND,c}.
\end{equation}

Apply Lemma~\ref{lem:b3-collision-canonical} to the first term. It gives the two
outer projections $\phi_{X2,h,w}$ and $\phi_{X1,h,w}$ and the outer
canonical remainder $R_X^d$. For the second term write
\[
 \mathbb P_{r,n_r}=P_r+\mathbb A_r,\qquad
 \mathbb Q_{r,M_r}=Q_r+\mathbb B_r,\qquad r=1,2.
\]
Because $P_2P_1H^c_{h,w}(\cdot;s)=0$ for every fixed evaluation point, all terms without
a fitting empirical operator vanish. The two terms with exactly one fitting
operator and no outer operator are
\[
 \mathbb A_2P_1Q_2Q_1H^c_{h,w},\qquad
 P_2\mathbb A_1Q_2Q_1H^c_{h,w},
\]
which are the labelled first projections
$\phi_{Y2,h,w}$ and $\phi_{Y1,h,w}$. The term with both fitting operators
and no outer operator is the labelled canonical remainder $R_Y^d$.

Every remaining operator product contains at least one $\mathbb A_r$ and at
least one $\mathbb B_r$. There are nine such terms: three with only
$\mathbb A_2$, three with only $\mathbb A_1$, and three with
$\mathbb A_2\mathbb A_1$; within each row the outer operator is
$\mathbb B_2$, $\mathbb B_1$, or $\mathbb B_2\mathbb B_1$. Their sum is
$R_4^{\mathrm{lin}}$. The fitting--outer index configurations fall into
$L_{00},L_{10},L_{01},L_{11}$ according to whether neither, one, or both
original-unit indices are shared; no independence is imposed between the
nested labelled and outer arrays. The last term in
\eqref{eq:proof-oracle-start} is
$R_4^{ND,c}=\mathbb Q_MR^{ND,c}$. Collecting the two outer
projections, two labelled projections, and four remainder blocks gives
\eqref{eq:exact-four-role} exactly.
\end{proof}

\paragraph{Proof of Theorem~\ref{thm:oracle-adapt-practical}}
\label{proof:oracle-adapt-practical}

\begin{proof}
Start from the exact reference-scale decomposition \eqref{eq:exact-four-role}.
Lemma~\ref{lem:b2-floor-expansion}, Lemma~\ref{lem:b3-collision-canonical},
Lemma~\ref{lem:b7-projection-stability}, and
Lemma~\ref{lem:b9-target-transfer} respectively control the nonlinear,
index-overlap, projection, and target-approximation terms at the required root scale.
Consequently, before the score-scale substitution,
\begin{equation*}
 \widehat\Delta_{k,h,w}^{\mathrm{or},c,\mathrm{ref}}-\Delta_k
 =L_{n,k}+o_p(\nu_k^{-1/2}),
\end{equation*}
where
\[
 L_{n,k}=M_{2k}^{-1}\sum_a\phi_{X2,k}(\mathbf{X}_{2ka})
       +M_{1k}^{-1}\sum_b\phi_{X1,k}(\mathbf{X}_{1kb})
       +n_{2k}^{-1}\sum_i\phi_{Y2,k}(O_{2ki})
       +n_{1k}^{-1}\sum_j\phi_{Y1,k}(O_{1kj}).
\]
This is the population-scale reference ALR. Lemma~\ref{lem:b10-hard-floor-transfer} and the exact hard-floor score-coordinate
identity in Lemma~\ref{lem:b6-candidate-scale}(i), with
$(h,w)=(h_0,w_0)$ in the reference raw-score coordinates, then give
\begin{equation*}
 \widehat\Delta_k^{\mathrm{or},c}
 -\widehat\Delta_{k,h,w}^{\mathrm{or},c,\mathrm{ref}}
 =o_p(\nu_k^{-1/2}).
\end{equation*}
Here $\widehat\Delta_k^{\mathrm{or},c}$ is the oracle statistic in
Section~\ref{sec:methodology} with sample score standardization. Together,
the two displays give \eqref{eq:oracle-alr}.

Proposition~\ref{prop:profile-consistency} gives
$\widehat{\boldsymbol{\vartheta}}_k=O_p(b_{n,k})$. On
$\|\widehat{\boldsymbol{\vartheta}}_k\|\le Cb_{n,k}$, replacing the population score scale
by its candidate-specific sample estimate gives the three-bracket decomposition
\begin{align*}
&\widehat T^{H,c}_{h,w}(\widehat{\boldsymbol{\vartheta}}_k,a_{\widehat{\boldsymbol{\vartheta}}_k})
 -\widehat T^{H,c}_{h,w}(0,a_0)\\
={}&
 [\widehat T^{H,c}_{h,w}(\widehat{\boldsymbol{\vartheta}}_k,a_{\widehat{\boldsymbol{\vartheta}}_k})
  -\widehat T^{H,c}_{h,w}(\widehat{\boldsymbol{\vartheta}}_k,1)]\\
&+
 [\widehat T^{H,c}_{h,w}(\widehat{\boldsymbol{\vartheta}}_k,1)
  -\widehat T^{H,c}_{h,w}(0,1)]\\
&+
 [\widehat T^{H,c}_{h,w}(0,1)
  -\widehat T^{H,c}_{h,w}(0,a_0)].
\end{align*}
The first bracket is $o_p(\nu_k^{-1/2})$ by part~(iii) of
Lemma~\ref{lem:b10-hard-floor-transfer}; the second is controlled by
Lemma~\ref{lem:b8-local-hard-floor} together with the beta-direction
parts of Lemma~\ref{lem:b7-projection-stability}; and the third is
$o_p(\nu_k^{-1/2})$ by part~(ii) of Lemma~\ref{lem:b10-hard-floor-transfer}.
The localization bounds in Lemma~\ref{lem:b10-hard-floor-transfer} apply to
all three brackets, whose sum gives \eqref{eq:beta-adaptivity}.

Equations \eqref{eq:oracle-alr} and \eqref{eq:beta-adaptivity} give
\eqref{eq:practical-alr}; centring leaves all four limiting influence
functions unchanged because $H_c-(m-c_H)=H-m$.
\end{proof}

\paragraph{Proof of Theorem~\ref{thm:oracle-adapt-practical} (parts (iv)--(v))}
\label{proof:group-clt}

\begin{proof}
Regroup the four leading averages in \eqref{eq:practical-alr} by original
unit. For \(a\le n_{rk}\), the summand is
\[
 \eta_{rka,k}=M_{rk}^{-1}\phi_{Xr,k}(\mathbf{X}_{rka})
              +n_{rk}^{-1}\phi_{Yr,k}(O_{rka}),
\]
whereas for \(a>n_{rk}\) it is
\(\eta_{rka,k}=M_{rk}^{-1}\phi_{Xr,k}(\mathbf{X}_{rka})\).
The original units are independent across arms and across distinct indices;
the first \(n_{rk}\) units alone carry both contributions. Therefore
\[
 \Var(L_{n,k})
 =\sum_{r=1}^2\left\{
 \frac{\Var(\phi_{Xr,k})}{M_{rk}}
 +\frac{\Var(\phi_{Yr,k})}{n_{rk}}
  +\frac{2}{M_{rk}}\Cov(\phi_{Xr,k},\phi_{Yr,k})\right\},
\]
which is \eqref{eq:group-var}. The covariance contribution is
\[
 2n_{rk}(M_{rk}^{-1})(n_{rk}^{-1})
 \Cov(\phi_{Xr,k},\phi_{Yr,k})
 =\frac{2}{M_{rk}}\Cov(\phi_{Xr,k},\phi_{Yr,k}).
\]
The coefficient \(2/M_{rk}\) is a nesting identity, not an independence
approximation.

Under correct specification, with
\(\epsilon_k=H_k-m_k(S_k)\),
\[
 E\{\phi_{Yr,k}(O_{rk})\mid \mathbf{X}_{rk}\}
 =E(\epsilon_k\mid \mathbf{X}_{rk})=0.
\]
Since \(\phi_{Xr,k}\) is \(\mathbf{X}_{rk}\)-measurable, this gives
\(\Cov(\phi_{Xr,k},\phi_{Yr,k})=0\). Thus the displayed variance is the
sum of the armwise variance contributions.

Bounded projections and the standing sample-ratio conditions give the original
unit Lyapunov bound
\[
 \sum_{r=1}^2\sum_{a=1}^{M_{rk}}E|\eta_{rka,k}|^3
 \le C\underline n_k^{-2},\qquad
 \underline n_k=\min(n_{1k},n_{2k}),
\]
and, by variance nondegeneracy,
\[
 \sigma_{\mathrm{SSL},k}^{-3}
 \sum_{r,a}E|\eta_{rka,k}|^3
 =O(\underline n_k^{-1/2})\longrightarrow0.
\]
This is a Lyapunov argument over independent original observations, never over
the \(n_{1k}n_{2k}\) labelled pairs. The feasible-estimator ALR and Slutsky's theorem
therefore prove group normality.

For the supervised two-sample estimator, conditioning
\[
 H_k-\Delta_k=\{m_k(S_k)-\Delta_k\}+\epsilon_k
\]
on an observation from arm \(r\) gives the first projection
\(\phi_{Xr,k}+\phi_{Yr,k}\). The same orthogonality removes its cross term,
so subtracting the semi-supervised variance from the supervised variance leaves
\[
 \sum_{r=1}^2
 \left(\frac1{n_{rk}}-\frac1{M_{rk}}\right)
 \Var\{\phi_{Xr,k}(\mathbf{X}_{rk})\}\ge0,
\]
because \(M_{rk}\ge n_{rk}\), proving \eqref{eq:group-variance-gain}.
\end{proof}

\paragraph{Proof of Corollary~\ref{cor:aggregate-clt}}
\label{proof:aggregate-clt}

\begin{proof}
Group independence and the group-normality statement in Theorem~\ref{thm:oracle-adapt-practical}(iv)
imply the joint Gaussian
limit for fixed $K$. Multiply each feasible group ALR by $a_{k,n}$ and sum. The
leading term is $\sum_ka_{k,n}L_{n,k}$. If
$r_{n,k}=o_p(\sigma_k)$ denotes the group remainder, put
$q_{k,n}=|a_{k,n}|\sigma_k/\sigma_{a,n}$. Then, because $K$ is fixed,
\[
 \frac{\left|\sum_k a_{k,n}r_{n,k}\right|}{\sigma_{a,n}}
 \le \sum_k q_{k,n}\frac{|r_{n,k}|}{\sigma_k}=o_p(1).
\]
Thus the aggregate remainder is $o_p(\sigma_{a,n})$.

The weighted original-unit summands are independent across groups and across
distinct original units, and their total variance is
$\sigma_{\mathrm{SSL},a,n}^2$. Put
\[
 q_{k,n}=\frac{|a_{k,n}|\sigma_{\mathrm{SSL},k}}
                 {\sigma_{\mathrm{SSL},a,n}},
 \qquad \sum_kq_{k,n}^2=1.
\]
By the original-unit Lyapunov bound established in the proof of the
group-normality part of Theorem~\ref{thm:oracle-adapt-practical} and
$\sum_kq_{k,n}^2=1$,
\[
 \frac{\sum_{k,r,a}E|a_{k,n}\eta_{rka,k}|^3}
      {\sigma_{\mathrm{SSL},a,n}^3}
 \le C\sum_kq_{k,n}^3\underline n_k^{-1/2}
 \le C\max_{k:a_{k,n}\ne0}\underline n_k^{-1/2}\to0.
\]
The Lyapunov theorem and Slutsky's theorem give the aggregate normal limit.

Setting \(a_{k,n}=1/K\) yields Ave; taking \(a_{k,n}\) to be the
prespecified scientific weights yields W. Group
independence makes both aggregate
variances sums of squared-weight group variances.
Multiplying \eqref{eq:group-variance-gain} by $a_{k,n}^2$ and summing proves
the aggregate comparison.

Finally, under the proportional unequal-allocation regime preceding
Corollary~\ref{cor:aggregate-clt}, the group ALRs retain their original-unit
normalizations on the group-specific scales. The same fixed-$K$
remainder and Lyapunov arguments therefore apply.
\end{proof}

\subsection{Perturbation inference}\label{app:C-perturbation}

\paragraph{Proof of Theorem~\ref{thm:ch4-validity}.}

\begin{proof}
We work under the conditions of Theorem~\ref{thm:ch4-validity}, including
Assumption~\ref{ass:perturbation} and the numerical-minimization condition
when applicable, and suppress the group index. We first obtain the
conditional two-arm decomposition and hard-floor linearisation. We then
control score standardization and index refitting, the fitting--outer overlap
generated by shared perturbation counts, and the replacement of feasible
projections by their population limits. These ingredients yield the
conditional linear representation in part (i), from which parts (ii)--(iii)
follow.

For a fixed group, write \(\mathcal D\) for the observed data and, for
\(r\in\{1,2\}\), write
\[
 P_{r,n_r}=\frac{1}{n_r}\sum_{i=1}^{n_r}\delta_{O_{ri}},
 \qquad
 U_{r,N_r}=\frac{1}{N_r}\sum_{j=1}^{N_r}\delta_{\mathbf{X}^U_{rj}},
\]
and
\[
 Q_{r,M_r}=\frac{n_r}{M_r}(P_{r,n_r})^X+
            \frac{N_r}{M_r}U_{r,N_r},
 \qquad M_r=n_r+N_r.
\]
The superscript \(^{*}\) denotes a conditional multinomial perturbation and
\(\xi=W-1\) a centred count. All conditional probability and expectation
statements below are with respect to the perturbation law given \(\mathcal D\).
The group-specific scales and conservative endpoint exponent are those in
\eqref{eq:bandwidth-rule}--\eqref{eq:floor-rule} and
Section~\ref{sec:theory}, with \((\alpha,\eta)\) satisfying
\eqref{eq:admissible-rate-region}. Under comparable allocation,
\(n_1\asymp n_2\asymp M_1\asymp M_2\asymp\nu\asymp n_L\); the proportional
unequal-allocation extension uses the regime stated in
Section~\ref{sec:aggregation}.

The four arm--stratum multinomial vectors are conditionally independent:
\[
 (W_{r1}^{L*},\ldots,W_{rn_r}^{L*})
 \sim\operatorname{Multinomial}(n_r;n_r^{-1},\ldots,n_r^{-1}),
\]
\[
 (W_{r1}^{U*},\ldots,W_{rN_r}^{U*})
 \sim\operatorname{Multinomial}(N_r;N_r^{-1},\ldots,N_r^{-1}).
\]
The same labelled count enters the labelled fitting operator and the labelled
component of the full outer covariate law; each replicate therefore fully
refits the score standardization, profile/index, hard-floor fit, and outer
average. The selected local minimizer is used; a numerical solution is
covered when it lies in the same locally strongly convex basin and satisfies
either
\[
 Q_n^*(\widetilde{\boldsymbol{\vartheta}}^*)-Q_n^*(\widehat{\boldsymbol{\vartheta}}^*)
 =o_{P^*,2}(n_L^{-1})
\]
or the projected-score condition \(o_{P^*,2}(n_L^{-1/2})\).

For a conditional multinomial representation, let
\(Z_1,\ldots,Z_m\) be one observed stratum and let
\[
 (W_1^*,\ldots,W_m^*)\sim\operatorname{Multinomial}
 (m;m^{-1},\ldots,m^{-1}).
\]
For every fixed function \(f\),
\[
 \sqrt m(P_m^*-P_m)f
 \ \stackrel{d^*}{=}\
 \frac{1}{\sqrt m}\sum_{b=1}^m
 \{f(Z_{I_b})-P_mf\},
\]
where \(I_1,\ldots,I_m\) are conditionally iid uniform on
\(\{1,\ldots,m\}\). Assumption~\ref{ass:perturbation} transfers the
VC/envelope conditions to these discrete conditional laws, yielding the
required fixed-order conditional maximal rates
\citep{PraestgaardWellner1993,vdVWellner1996,vdVWellner2011,Cheng2015}.

For a labelled pair kernel \(f(o_2,o_1)\), write
\(R_n=P_{2,n_2}P_{1,n_1}\), \(R_n^*=P_{2,n_2}^*P_{1,n_1}^*\), and set
\[
 f_{2,n}=P_{1,n_1}f-R_nf,\qquad
 f_{1,n}=P_{2,n_2}f-R_nf,
\]
\[
 f_n^\circ=f-P_{1,n_1}f-P_{2,n_2}f+R_nf.
\]
Expanding the two-arm multinomial product gives the exact identity
\[
 (R_n^*-R_n)f
 =(P_{2,n_2}^*-P_{2,n_2})f_{2,n}
 +(P_{1,n_1}^*-P_{1,n_1})f_{1,n}
 +C_n^*(f),
\]
where
\[
 C_n^*(f)=\frac{1}{n_2n_1}\sum_{i=1}^{n_2}\sum_{j=1}^{n_1}
 \xi_{2i}^{L*}\xi_{1j}^{L*}f_{n,ij}^\circ
\]
and
\[
 E^*\{C_n^*(f)^2\mid\mathcal D\}
 =\frac{1}{n_2^2n_1^2}\sum_{i,j}(f_{n,ij}^\circ)^2.
\]
For the full outer empirical law, the same product expansion gives
\[
 (Q_2^*Q_1^*-Q_2Q_1)m
 =\mathbb H_2^*Q_1m+Q_2\mathbb H_1^*m+
   \mathbb H_2^*\mathbb H_1^*m,
\]
where
\[
 \mathbb H_r^*=
 \frac{1}{M_r}\sum_{i=1}^{n_r}\xi_{ri}^{L*}\delta_{\mathbf{X}_{ri}}+
 \frac{1}{M_r}\sum_{j=1}^{N_r}\xi_{rj}^{U*}\delta_{\mathbf{X}_{rj}^U}.
\]
The zero totals of the centred counts and conditional independence give the
displayed first projections and doubly canonical remainder.

For conditional derivative-process control, let \(\mathcal G_{q,n}\),
\(q=0,1,2\), denote the local classes from Section~\ref{sec:theory} of
Gaussian translate and local-coordinate derivative kernel blocks, including
their deleted-row and deleted-column versions. Define
\[
 a_{q,n}=
 \left(\frac{\ell_n}{n_2}\right)^{1/2}
+\left(\frac{\ell_n}{n_1}\right)^{1/2}
+\left(\frac{\ell_n}{n_2n_1h^{2q+1}}\right)^{1/2},
\qquad \ell_n=1+\log(n_L/h).
\]
The conditional analogue of the Appendix B process bounds under
Assumption~\ref{ass:perturbation} yields, for every fixed finite \(p\ge2\),
\[
 \left(
 E^*\sup_{g\in\mathcal G_{q,n}}|(R_n^*-R_n)g|^p
 \,\middle|\,\mathcal D\right)^{1/p}=O_p(a_{q,n}),
\qquad q=0,1,2.
\]
The same orders hold for the finite collection of once- and twice-deleted
arrays used in the profile score, Hessian, and index-overlap arguments; the
two armwise first projections and the doubly canonical term give the three
components of \(a_{q,n}\).

To control conditional refitting and the floor, put
\(F_w(q,d)=c_H+q/(d\vee w)\), and, at an observed evaluation point, let
\(A=\widehat q_c^*-\widehat q_c\),
\(B=\widehat d^*-\widehat d\), \(q=\widehat q_c\), and \(d=\widehat d\). The
linear part is
\[
 \mathcal L_w(A,B;q,d)
 =\frac{A}{d\vee w}
 -1(d>w)\frac{qB}{d^2}.
\]
On the event
\(\|\widehat q_c^*-\widehat q_c\|_\infty
 \vee\|\widehat d^*-\widehat d\|_\infty\le\delta\), with
\(0<\delta\le w/4\), the remainder
\[
 \mathcal R_w=F_w(q+A,d+B)-F_w(q,d)-\mathcal L_w(A,B;q,d)
\]
is zero when \(d\le w-2\delta\), bounded by
\(C\delta/w\) when \(|d-w|\le2\delta\), and bounded by
\(C\delta^2/d^2\) when \(d\ge w+2\delta\). The margin and inverse-square
bounds therefore give, after averaging over the observed full-covariate law,
\[
 E^*\left[\left\{\widehat Q\,\mathcal R_w\right\}^2
 \,\middle|\,\mathcal D\right]
 =O_p\{a_{0,n}^4w^{2\gamma_*-4}\}+o_p(\nu^{-1}),
\]
where \(a_{0,n}\) is the \(q=0\) rate. Hence
\[
 E^*\{\nu(\widehat Q\,\mathcal R_w)^2\mid\mathcal D\}\xrightarrow{p}0.
\]
Indeed, above the threshold ordinary ratio algebra gives
\[
 -\frac{AB}{d(d+B)}+\frac{qB^2}{d^2(d+B)},
\]
while below the threshold both denominators are floored and the transition
band is controlled by its margin probability. The conditional
derivative-process bound controls \(\delta\), and (R) gives
\(a_{0,n}/w=o_p(1)\) and
\(\nu a_{0,n}^4w^{2\gamma_*-4}=o_p(1)\). The exponential-tail maximal
inequality and \(p_B(1/2-\eta)>1\) make the complement event negligible at
root scale.

Under the additive arm-score representation, the replicate score mean and
variance reduce to armwise empirical laws. Let \(\widehat\mu_s\) and
\(\widehat\sigma_s\) be the observed weighted-score mean and empirical
standard deviation. Uniformly on the local neighbourhood,
\begin{align*}
 \widehat\mu_s^*-\widehat\mu_s
 ={}&\frac{1}{n_2}\sum_{i=1}^{n_2}\xi_{2i}^{L*}
      \widehat\psi_{\mu,2}(O_{2i})
    +\frac{1}{n_1}\sum_{j=1}^{n_1}\xi_{1j}^{L*}
      \widehat\psi_{\mu,1}(O_{1j})
    +o_{P^*,2}(n_L^{-1/2}),\\
 \widehat v_s^*-\widehat v_s
 ={}&\frac{1}{n_2}\sum_{i=1}^{n_2}\xi_{2i}^{L*}
      \widehat\psi_{v,2}(O_{2i})
    +\frac{1}{n_1}\sum_{j=1}^{n_1}\xi_{1j}^{L*}
      \widehat\psi_{v,1}(O_{1j})
    +o_{P^*,2}(n_L^{-1/2}),
\end{align*}
where \(\widehat v_s=\widehat\sigma_s^2\) is the empirical variance and each
\(\widehat\psi\) is the corresponding armwise centred first projection.
Conditional on the high-probability local identification event established in
Section~\ref{sec:theory}, \(\widehat\sigma_s\) is bounded away from zero.
The square-root expansion is therefore
\begin{equation*}
 \widehat\sigma_s^*-\widehat\sigma_s
 =\frac{\widehat v_s^*-\widehat v_s}{2\widehat\sigma_s}
   +o_{P^*,2}(n_L^{-1/2}).
\end{equation*}
The original-score bandwidth and floor move with
\(\widehat\sigma_s^*\). With
\(\lambda=\log\widehat\sigma_s\), the relevant scale derivatives are
\[
 \partial_\lambda K(u)=-uK'(u),\qquad
 \partial_{\lambda\lambda}K(u)=uK'(u)+u^2K''(u).
\]
The score-mean increment cancels from score differences; the scale increment
enters the armwise projections through these derivative blocks and is not a new
target influence role.

The conditional profile score satisfies
\[
 \dot Q_n^*(\widehat{\boldsymbol{\vartheta}})
 =\frac{1}{n_2}\sum_{i=1}^{n_2}\xi_{2i}^{L*}\widehat{\boldsymbol{\zeta}}_{2i}
  +\frac{1}{n_1}\sum_{j=1}^{n_1}\xi_{1j}^{L*}\widehat{\boldsymbol{\zeta}}_{1j}
  +r_{S,n}^*,
\qquad
 E^*\{n_L\|r_{S,n}^*\|^2\mid\mathcal D\}\xrightarrow{p}0.
\]
For every fixed \(C\),
\[
 \sup_{\|\boldsymbol{\vartheta}-\widehat{\boldsymbol{\vartheta}}\|\le Cn_L^{-1/2}}
 \|\ddot Q_n^*(\boldsymbol{\vartheta})-2\widehat{\mathbf H}\|=o_{P^*}(1)
\]
in outer probability. The perturbed selected local minimizer is localized at
\(\|\widehat{\boldsymbol{\vartheta}}^*
-\widehat{\boldsymbol{\vartheta}}\|=O_{P^*}(n_L^{-1/2})\), and
\[
 \widehat{\boldsymbol{\vartheta}}^*-\widehat{\boldsymbol{\vartheta}}
 =-(2\widehat{\mathbf H})^{-1}\dot Q_n^*(\widehat{\boldsymbol{\vartheta}})
  +r_{\boldsymbol{\beta},n}^*,
\qquad
 E^*\{n_L\|r_{\boldsymbol{\beta},n}^*\|^2\mid\mathcal D\}\xrightarrow{p}0.
\]
Consequently,
\[
 \|\widehat{\boldsymbol{\beta}}^*-\widehat{\boldsymbol{\beta}}\|=O_{P^*}(n_L^{-1/2}).
\]
The conditional analogues of the Appendix B derivative, deletion,
transition-band, and canonical bounds give the score and Hessian assertions.
Subtracting the two first-order conditions and integrating the perturbed
Hessian along the joining segment gives
\[
 0=\dot Q_n^*(\widehat{\boldsymbol{\vartheta}})
  +\left\{\int_0^1\ddot Q_n^*
  \bigl(\widehat{\boldsymbol{\vartheta}}+t(\widehat{\boldsymbol{\vartheta}}^*-\widehat{\boldsymbol{\vartheta}})\bigr)
  \,dt\right\}(\widehat{\boldsymbol{\vartheta}}^*-\widehat{\boldsymbol{\vartheta}}).
\]
Local strong convexity and the conditional moments in
Assumption~\ref{ass:perturbation} transfer the remainder to the beta
parametrisation.

For the fully fitted hard-floor functional, write
\[
 G^*(\boldsymbol{\beta})=F_w(q_{\boldsymbol{\beta},c}^*,d_{\boldsymbol{\beta}}^*),\qquad
 G(\boldsymbol{\beta})=F_w(q_{\boldsymbol{\beta},c},d_{\boldsymbol{\beta}}),
\]
and decompose
\[
 G^*(\widehat{\boldsymbol{\beta}}^*)-G(\widehat{\boldsymbol{\beta}})
 =A+B+C,
\]
where
\[
 A=G^*(\widehat{\boldsymbol{\beta}})-G(\widehat{\boldsymbol{\beta}}),\quad
 B=G(\widehat{\boldsymbol{\beta}}^*)-G(\widehat{\boldsymbol{\beta}}),
\]
\[
 C=\{G^*(\widehat{\boldsymbol{\beta}}^*)-G^*(\widehat{\boldsymbol{\beta}})\}
   -\{G(\widehat{\boldsymbol{\beta}}^*)-G(\widehat{\boldsymbol{\beta}})\}.
\]
The fixed-beta term \(A\) belongs to the leading hard-floor expansion, while
\[
 E^*\left[
 \nu\left\{Q_M^*(B+C)\right\}^2\mid\mathcal D\right]\xrightarrow{p}0.
\]
Away from the floor, the nuisance-adaptivity identity established in
Section~\ref{sec:theory} makes the integrated first derivative orthogonal to
the score-direction increment. Across the floor threshold, the observed
margin and score-scale expansion imply, with
\(r_{\boldsymbol{\beta}}=\|\widehat{\boldsymbol{\beta}}^*
-\widehat{\boldsymbol{\beta}}\|\),
\[
 \widehat Q\{\text{floor status changes}\}
 \le C r_{\boldsymbol{\beta}} w^{\gamma_*-1}+\text{empirically negligible terms}.
\]
Since \(r_{\boldsymbol{\beta}}=O_{P^*}(n_L^{-1/2})\), (R) makes the
smooth-region, crossing, and joint beta--fitting switching contributions
\(o_p(1)\) at the root scale, including the scale increment in the crossing
width.

To retain the dependence induced by shared labelled counts, for one stratum
of size \(m\) the exact multinomial covariance is
\[
 E^*\xi_i=0,\qquad
 E^*(\xi_i\xi_j)=1(i=j)-m^{-1}.
\]
Thus, for a fixed vector \(a\),
\[
 \Var^*\left(\sum_i\xi_i a_i\right)
 =\sum_i(a_i-\bar a)^2,
\]
and, for independent strata and a rectangular array \(A\),
\[
 E^*(\xi_r^\top A\xi_s)^2
 =\operatorname{tr}(\Pi_rA\Pi_sA^\top)\le\|A\|_F^2.
\]

For the shared-labelled overlap, let \(R_{\mathrm{cross},n}^*\) collect the
fitting-by-outer terms after their first projections have been removed. Then
\[
 E^*\{(R_{\mathrm{cross},n}^*)^2\mid\mathcal D\}
 =O_p(r_{4,h,w}^2),
\]
where
\begin{align}
 r_{4,h,w}^2={}&\frac{1}{h}\left(
 \frac{1}{n_2M_2}+\frac{1}{n_2M_1}
 +\frac{1}{n_1M_2}+\frac{1}{n_1M_1}\right)\notag\\
 &+\frac{w^{-\kappa_2}}{M_2^2h}
  +\frac{w^{-\kappa_1}}{M_1^2h}
  +\frac{w^{-\kappa_{12}}}{M_2^2M_1^2h^2}
  +(M_2^{-1}+M_1^{-1})^2.\label{eq:C4-collision-bound}
\end{align}
Over the conservative local parameter range,
\[
 E^*\{\nu(R_{\mathrm{cross},n}^*)^2\mid\mathcal D\}\xrightarrow{p}0.
\]
Delete the fitting row and column and the labelled outer row and column. The
common deleted kernel is then reinserted as
\[
 H=H^{-ijab}+L_{10}+L_{01}+L_{11}.
\]
The deleted array is rectangularly centred, so every contraction against a row
or column sum vanishes under \(\Pi_m\). The remaining terms fall into five
classes:
\[
\begin{array}{lll}
\text{class} & \text{active counts and normalizer} & \text{conditional squared order}\\
\hline
\text{no shared index} & \xi_{ri}\xi_{sa}/(n_rM_s) & (n_rM_sh)^{-1}\\
\text{same-arm single} & \xi_{ri}^2/(n_rM_r) & w^{-\kappa_r}/(M_r^2h)\\
\text{opposite-arm} & \xi_{2i}\xi_{1a}/(n_2M_1)\ \text{or reverse} & (n_rM_sh)^{-1}\\
\text{double labelled} & \xi_{2i}^2\xi_{1j}^2/(n_2n_1M_2M_1)
 & w^{-\kappa_{12}}/(M_2^2M_1^2h^2)\\
\text{two-arm canonical} & \xi_2^\top(\Pi_{21}A)\xi_1/(m_2m_1)
 & O((M_2M_1)^{-1}).
\end{array}
\]
The same-arm and double-labelled diagonal means
\(E^*\xi_i^2=1-m^{-1}\) are retained, while centred fluctuations and
reinsertions use the multinomial contraction and the
\(L_{10},L_{01},L_{11}\) envelopes. Summing the five classes gives
\eqref{eq:C4-collision-bound}; the calculation retains the within-stratum
negative covariance and shared labelled counts. The overlap inequalities
following \eqref{eq:admissible-rate-region} give
\(\nu r_{4,h,w}^2=o_p(1)\).

To replace the feasible projections, let \(\widehat\phi_{Xr}\) and
\(\widehat\phi_{Yr}\) be the observed feasible,
finite-\((h,w)\), random-\(\widehat{\boldsymbol{\beta}}\), hard-floor projections, centred
in their own original-unit empirical strata. Then
\[
 \sum_{r=1}^2\left(
 M_r^{-1}\|\widehat\phi_{Xr}-\phi_{Xr}\|_{Q_{r,M}}^2+
 n_r^{-1}\|\widehat\phi_{Yr}-\phi_{Yr}\|_{P_{r,n}}^2
  \right)=o_p(\nu^{-1}),
\]
and if \(\widehat L_n^*\) replaces \(\phi\) by \(\widehat\phi\) in the
multinomial leading term, then
\[
 E^*\{\nu(\widehat L_n^*-L_n^*)^2\mid\mathcal D\}\xrightarrow{p}0.
\]
Put \(N_{Xr}=M_r\), \(N_{Yr}=n_r\), and let \(P_{\rho,N_\rho}\) denote the
corresponding original-unit empirical law. For each
\(\rho\in\{X2,X1,Y2,Y1\}\), the replacement chain is
\[
\widehat\phi_\rho
\ \longrightarrow\
\phi_{\rho,n}(\widehat{\boldsymbol{\vartheta}},a_{\widehat{\boldsymbol{\vartheta}}})
\ \longrightarrow\
\phi_{\rho,n}(\widehat{\boldsymbol{\vartheta}},1)
\ \longrightarrow\
\phi_{\rho,n}(0,1)
\ \longrightarrow\
\phi_\rho.
\]
The four differences are controlled, respectively, by the conditional
process/deletion/overlap bounds, Lemma~\ref{lem:b7-projection-stability}(iv)
and Lemma~\ref{lem:b6-candidate-scale} localization, parts~(ii)--(iii) of
Lemma~\ref{lem:b7-projection-stability}, and part~(i) of that lemma. The
corresponding empirical original-unit norms follow from the uniform law and
\(L_2\) contraction; since
\(M_r^{-1}\asymp n_r^{-1}\asymp\nu^{-1}\), these bounds imply the first
displayed assertion.

For a labelled unit, the coefficient difference is
\[
 b_{ri}^L=
 \frac{\widehat\phi_{Xr}(\mathbf{X}_{ri})-\phi_{Xr}(\mathbf{X}_{ri})}{M_r}
+\frac{\widehat\phi_{Yr}(O_{ri})-\phi_{Yr}(O_{ri})}{n_r}.
\]
For an unlabelled unit it is
\(b_{rj}^U=\{\widehat\phi_{Xr}(\mathbf{X}_{rj}^U)-\phi_{Xr}(\mathbf{X}_{rj}^U)\}/M_r\).
For each multinomial stratum,
\[
 E^*\left[\left\{\sum_i\xi_i^*b_i\right\}^2\,\middle|\,\mathcal D\right]
 =\sum_i(b_i-\bar b)^2\le\sum_i b_i^2.
\]
Thus the conditional variance is bounded by
\(\sum_i(b_{ri}^L)^2+\sum_j(b_{rj}^U)^2\) after summing the independent strata.
The projection bound makes this \(o_p(\nu^{-1})\), and hence
\[
 E^*\{\nu(\widehat L_n^*-L_n^*)^2\mid\mathcal D\}
 \xrightarrow{p}0.
\]

Combining the preceding expansions yields the conditional first-order
representation in part (i):
\[
 \widehat\Delta_n^*-\widehat\Delta_n=L_n^*+R_n^*,
 \qquad E^*\{\nu(R_n^*)^2\mid\mathcal D\}\xrightarrow{p}0,
\]
where
\[
 L_n^*=\sum_{r=1}^2\left[
 \sum_{i=1}^{n_r}\xi_{ri}^{L*}
 \left\{\frac{\phi_{Xr}(\mathbf{X}_{ri})}{M_r}
 +\frac{\phi_{Yr}(O_{ri})}{n_r}\right\}
 +\sum_{j=1}^{N_r}\xi_{rj}^{U*}
 \frac{\phi_{Xr}(\mathbf{X}_{rj}^U)}{M_r}\right].
\]
At fixed observed \(\widehat{\boldsymbol{\beta}}\), the exact outer-product and
hard-floor expansions supply the feasible four-component term; the outer
canonical, labelled product-canonical, and fitting--outer remainders are
controlled by the bounds above, while the profile and beta-adaptivity
expansions control refitting and index movement. The projection-replacement
bound transfers the feasible projections to the limiting projections from
Section~\ref{sec:theory}; conditional Minkowski's inequality over the finite
collection gives the displayed remainder.

For part (ii), the conditional iid representation in the standing perturbation
setup gives, for each stratum \(s\) of size \(m_s\),
\[
 \sum_{i=1}^{m_s}\xi_{si}^*a_{si}
 =\sum_{b=1}^{m_s}\{a_{s,I_{sb}}-\bar a_s\}.
\]
The four stratum arrays are conditionally independent. The bounded
projections established in Section~\ref{sec:theory} imply
\(\max_{s,i}|a_{si}|=O_p(\nu^{-1})\), while
\(\sigma_n^2\asymp\nu^{-1}\). Hence
\[
 \max_{s,i}|a_{si}-\bar a_s|/\sigma_n=O_p(\nu^{-1/2})\to0,
\]
so the conditional Lindeberg sum is eventually zero on a
probability-tending-one event. Therefore
\[
 d_{BL}\left\{
 \mathcal L^*\left(\frac{L_n^*}{\sigma_n}\,\middle|\,\mathcal D\right),
 N(0,1)\right\}\xrightarrow{p}0,
\]
\[
 \Var^*(L_n^*\mid\mathcal D)=\sigma_n^2\{1+o_p(1)\}.
\]

For part (iii), the exact conditional variance of the leading term contains,
for labelled arm \(r\),
\[
 \frac{2}{M_rn_r}\sum_{i=1}^{n_r}
 (\phi_{Xr,i}-\bar\phi_{Xr,L})
 (\phi_{Yr,i}-\bar\phi_{Yr}),
\]
which converges to
\(2\Cov(\phi_{Xr},\phi_{Yr})/M_r\). Adding the labelled and unlabelled
outer components yields
\[
 \sigma_n^2=\sum_{r=1}^2\left\{
 \frac{\Var(\phi_{Xr})}{M_r}
 +\frac{\Var(\phi_{Yr})}{n_r}
  +\frac{2\Cov(\phi_{Xr},\phi_{Yr})}{M_r}\right\}.
\]
The conditional variance of the remainder is \(o_p(\nu^{-1})\), and conditional
Cauchy--Schwarz gives the same order for the leading--remainder covariance.
Thus the full-refit perturbation has the ideal variance ratio in part (iii).
The conditional Lindeberg--Feller theorem and Slutsky's theorem transfer the
Gaussian approximation from \(L_n^*\) to the full-refit perturbation,
completing parts (ii) and (iii).
\end{proof}

\paragraph{Proof of Corollary~\ref{cor:ch4-finiteB-wald}.}

\begin{proof}
For iid full perturbation replicates, conditional unbiasedness of the sample
variance and the exact iid sample-variance formula give
\[
 \Var^*(\widehat V_B\mid\mathcal D)
 \le \frac{1}{B}E^*\{(Z_n^*)^4\mid\mathcal D\}.
\]
Thus, for every \(\epsilon>0\),
\[
 P^*\left(\left|\frac{\widehat V_B}{V_n^*}-1\right|>\epsilon
 \,\middle|\,\mathcal D\right)
 \le\frac{1}{B\epsilon^2}
 \frac{E^*\{(Z_n^*)^4\mid\mathcal D\}}{(V_n^*)^2}.
\]
The displayed conditional fourth-moment condition and \(B_n\to\infty\) prove
the variance-ratio claim. Combining this result with the sampling CLT in
Theorem~\ref{thm:oracle-adapt-practical} and Slutsky's theorem yields the
stated Wald coverage.
\end{proof}

\paragraph{Proof of Corollary~\ref{cor:ch4-aggregation}.}
\begin{proof}
Apply Corollary~\ref{cor:ch4-finiteB-wald} groupwise and sum the finite number of
independent group ALRs. Cross-group sampling covariances vanish by
independence. Hence the aggregate variance is the sum of the group variances
weighted by $a_{k,n}^2$.
\end{proof}


\clearpage
\begin{thebibliography}{99}

\bibitem[Ahmed, Einmahl, and Zhou(2025)]{ahmed2025extreme}
Ahmed, H., Einmahl, J. H. J., and Zhou, C. (2025).
\newblock Extreme value statistics in semi-supervised models.
\newblock \emph{Journal of the American Statistical Association}, 120(549), 291--304.
\newblock \url{https://doi.org/10.1080/01621459.2024.2333582}

\bibitem[Arcones and Gin\'e(1992)]{ArconesGine1992Bootstrap}
Arcones, M. A., and Gin\'e, E. (1992). On the bootstrap of $U$ and $V$ statistics. \emph{The Annals of Statistics}, 20(2), 655--674.
\newblock \url{https://doi.org/10.1214/aos/1176348650}

\bibitem[Arcones and Gin\'e(1993)]{ArconesGine1993}
Arcones, M. A., and Gin\'e, E. (1993). Limit theorems for $U$-processes. \emph{The Annals of Probability}, 21(3), 1494--1542.
\newblock \url{https://doi.org/10.1214/aop/1176989128}

\bibitem[Azriel et al.(2022)]{azriel2022}
Azriel, D., Brown, L. D., Sklar, M., Berk, R., Buja, A., and Zhao, L. (2022).
\newblock Semi-supervised linear regression.
\newblock \emph{Journal of the American Statistical Association}, 117(540), 2238--2251.
\newblock \url{https://doi.org/10.1080/01621459.2021.1915320}

\bibitem[Bickel et al.(1991)]{bickel1991}
Bickel, P. J., Ritov, Y., and Wellner, J. A. (1991).
\newblock Efficient estimation of linear functionals of a probability measure $P$ with known marginal distributions.
\newblock \emph{The Annals of Statistics}, 19(3), 1316--1346.
\newblock \url{https://doi.org/10.1214/aos/1176348251}

\bibitem[Cattaneo and Jansson(2018)]{CattaneoJansson2018}
Cattaneo, M. D., and Jansson, M. (2018).
\newblock Kernel-based semiparametric estimators: small bandwidth asymptotics
and bootstrap consistency.
\newblock \emph{Econometrica}, 86(3), 955--995.
\newblock \url{https://doi.org/10.3982/ECTA12701}

\bibitem[Chakrabortty and Cai(2018)]{chakrabortty2018}
Chakrabortty, A., and Cai, T. (2018).
\newblock Efficient and adaptive linear regression in semi-supervised settings.
\newblock \emph{The Annals of Statistics}, 46(4), 1541--1572.
\newblock \url{https://doi.org/10.1214/17-AOS1594}

\bibitem[Chapelle et al.(2006)]{chapelle2006}
Chapelle, O., Sch\"olkopf, B., and Zien, A. (Eds.). (2006).
\newblock \emph{Semi-Supervised Learning}.
\newblock MIT Press.
\newblock \url{https://doi.org/10.7551/mitpress/9780262033589.001.0001}

\bibitem[Cheng(2015)]{Cheng2015}
Cheng, G. (2015). Moment consistency of the exchangeably weighted bootstrap for semiparametric $M$-estimation. \emph{Scandinavian Journal of Statistics}, 42(3), 665--684.
\newblock \url{https://doi.org/10.1111/sjos.12128}

\bibitem[Cheng and Huang(2010)]{ChengHuang2010}
Cheng, G., and Huang, J. Z. (2010). Bootstrap consistency for general semiparametric $M$-estimation. \emph{The Annals of Statistics}, 38(5), 2884--2915.
\newblock \url{https://doi.org/10.1214/10-AOS809}

\bibitem[de la Pe\~na and Gin\'e(1999)]{deLaPenaGine1999}
de la Pe\~na, V. H., and Gin\'e, E. (1999). \emph{Decoupling: From Dependence to Independence}. Springer.
\newblock \url{https://doi.org/10.1007/978-1-4612-0537-1}

\bibitem[Delecroix et al.(2006)]{DelecroixHristachePatilea2006}
Delecroix, M., Hristache, M., and Patilea, V. (2006).
\newblock On semiparametric $M$-estimation in single-index regression.
\newblock \emph{Journal of Statistical Planning and Inference}, 136(3), 730--769.
\newblock \url{https://doi.org/10.1016/j.jspi.2004.09.006}

\bibitem[Delyon and Portier(2016)]{DelyonPortier2016}
Delyon, B., and Portier, F. (2016).
\newblock Integral approximation by kernel smoothing.
\newblock \emph{Bernoulli}, 22(4), 2177--2208.
\newblock \url{https://doi.org/10.3150/15-BEJ725}

\bibitem[Gronsbell and Cai(2018)]{gronsbell2018}
Gronsbell, J. L., and Cai, T. (2018).
\newblock Semi-supervised approaches to efficient evaluation of model prediction performance.
\newblock \emph{Journal of the Royal Statistical Society: Series B (Statistical Methodology)}, 80(3), 579--594.
\newblock \url{https://doi.org/10.1111/rssb.12264}

\bibitem[H\"ardle et al.(1993)]{hardle1993}
H\"ardle, W., Hall, P., and Ichimura, H. (1993).
\newblock Optimal smoothing in single-index models.
\newblock \emph{The Annals of Statistics}, 21(1), 157--178.
\newblock \url{https://doi.org/10.1214/aos/1176349020}

\bibitem[Hoeffding(1948)]{Hoeffding1948}
Hoeffding, W. (1948). A class of statistics with asymptotically normal distribution. \emph{The Annals of Mathematical Statistics}, 19(3), 293--325.
\newblock \url{https://doi.org/10.1214/aoms/1177730196}

\bibitem[Huang, Liu, and Peng(2023)]{HuangLiuPeng2023}
Huang, B., Liu, Y., and Peng, L. (2023). Weighted bootstrap for two-sample $U$-statistics. \emph{Journal of Statistical Planning and Inference}, 226, 86--99.
\newblock \url{https://doi.org/10.1016/j.jspi.2023.02.004}

\bibitem[Ichimura(1993)]{ichimura1993}
Ichimura, H. (1993).
\newblock Semiparametric least squares (SLS) and weighted SLS estimation of single-index models.
\newblock \emph{Journal of Econometrics}, 58(1--2), 71--120.
\newblock \url{https://doi.org/10.1016/0304-4076(93)90114-K}

\bibitem[Ichimura and Lee(2010)]{IchimuraLee2010}
Ichimura, H., and Lee, S. (2010). Characterization of the asymptotic distribution of semiparametric $M$-estimators. \emph{Journal of Econometrics}, 159(2), 252--266.
\newblock \url{https://doi.org/10.1016/j.jeconom.2010.05.005}

\bibitem[Ichimura and Lee(2018)]{IchimuraLee2018Corrigendum}
Ichimura, H., and Lee, S. (2018). Corrigendum to ``Characterization of the asymptotic distribution of semiparametric $M$-estimators''. \emph{Journal of Econometrics}, 202(2), 306--307.
\newblock \url{https://doi.org/10.1016/j.jeconom.2017.07.003}

\bibitem[Jin et al.(2001)]{jin2001}
Jin, Z., Ying, Z., and Wei, L. J. (2001).
\newblock A simple resampling method by perturbing the minimand.
\newblock \emph{Biometrika}, 88(2), 381--390.
\newblock \url{https://doi.org/10.1093/biomet/88.2.381}

\bibitem[Kim et al.(2025)]{kim2025semisupervised}
Kim, I., Wasserman, L., Balakrishnan, S., and Neykov, M. (2025).
\newblock Semi-supervised $U$-statistics.
\newblock \emph{The Annals of Statistics}, 53(6), 2488--2515.
\newblock \url{https://doi.org/10.1214/25-AOS2550}

\bibitem[Mammen, Rothe, and Schienle(2012)]{MammenRotheSchienle2012}
Mammen, E., Rothe, C., and Schienle, M. (2012). Nonparametric regression with nonparametrically generated covariates. \emph{The Annals of Statistics}, 40(2), 1132--1170.
\newblock \url{https://doi.org/10.1214/12-AOS995}

\bibitem[Mammen, Rothe, and Schienle(2016)]{MammenRotheSchienle2016}
Mammen, E., Rothe, C., and Schienle, M. (2016).
\newblock Semiparametric estimation with generated covariates.
\newblock \emph{Econometric Theory}, 32(5), 1140--1177.
\newblock \url{https://doi.org/10.1017/S0266466615000134}

\bibitem[Nadaraya(1964)]{nadaraya1964}
Nadaraya, E. A. (1964).
\newblock On estimating regression.
\newblock \emph{Theory of Probability \& Its Applications}, 9(1), 141--142.
\newblock \url{https://doi.org/10.1137/1109020}

\bibitem[Newey(1994)]{Newey1994PartialMeans}
Newey, W. K. (1994). Kernel estimation of partial means and a general variance estimator. \emph{Econometric Theory}, 10(2), 233--253.
\newblock \url{https://doi.org/10.1017/S0266466600008409}

\bibitem[Newey and McFadden(1994)]{newey1994}
Newey, W. K., and McFadden, D. (1994).
\newblock Large sample estimation and hypothesis testing.
\newblock In R. F. Engle and D. L. McFadden (Eds.), \emph{Handbook of Econometrics}, Vol. 4, 2111--2245.
\newblock Elsevier.
\newblock \url{https://doi.org/10.1016/S1573-4412(05)80005-4}

\bibitem[Nolan and Pollard(1987)]{NolanPollard1987}
Nolan, D., and Pollard, D. (1987). $U$-processes: rates of convergence. \emph{The Annals of Statistics}, 15(2), 780--799.
\newblock \url{https://doi.org/10.1214/aos/1176350374}

\bibitem[Pakes and Pollard(1989)]{PakesPollard1989}
Pakes, A., and Pollard, D. (1989).
\newblock Simulation and the asymptotics of optimization estimators.
\newblock \emph{Econometrica}, 57(5), 1027--1057.
\newblock \url{https://doi.org/10.2307/1913622}

\bibitem[Pr\ae stgaard and Wellner(1993)]{PraestgaardWellner1993}
Pr\ae stgaard, J., and Wellner, J. A. (1993). Exchangeably weighted bootstraps of the general empirical process. \emph{The Annals of Probability}, 21(4), 2053--2086.
\newblock \url{https://doi.org/10.1214/aop/1176989011}

\bibitem[Serfling(1980)]{Serfling1980}
Serfling, R. J. (1980). \emph{Approximation Theorems of Mathematical Statistics}. Wiley.
\newblock \url{https://doi.org/10.1002/9780470316481}

\bibitem[Sherman(1994a)]{Sherman1994}
Sherman, R. P. (1994a). Maximal inequalities for degenerate $U$-processes with applications to optimization estimators. \emph{The Annals of Statistics}, 22(1), 439--459.
\newblock \url{https://doi.org/10.1214/aos/1176325377}

\bibitem[Sherman(1994b)]{Sherman1994Semiparametric}
Sherman, R. P. (1994b).
\newblock $U$-processes in the analysis of a generalized semiparametric regression estimator.
\newblock \emph{Econometric Theory}, 10(2), 372--395.
\newblock \url{https://doi.org/10.1017/S0266466600008458}

\bibitem[Song(2014)]{Song2014}
Song, K. (2014).
\newblock Semiparametric models with single-index nuisance parameters.
\newblock \emph{Journal of Econometrics}, 178(3), 471--483.
\newblock \url{https://doi.org/10.1016/j.jeconom.2013.07.004}

\bibitem[Song, Lin, and Zhou(2024)]{song2024}
Song, S., Lin, Y., and Zhou, Y. (2024).
\newblock A general M-estimation theory in semi-supervised framework.
\newblock \emph{Journal of the American Statistical Association}, 119(546), 1065--1075.
\newblock \url{https://doi.org/10.1080/01621459.2023.2169699}

\bibitem[Tan, Zhang, and Zhou(2025)]{tan2025}
Tan, T., Zhang, S., and Zhou, Y. (2025).
\newblock Efficient semiparametric estimation in two-sample comparison via semisupervised learning.
\newblock \emph{Canadian Journal of Statistics}, 53(2), e11813.
\newblock \url{https://doi.org/10.1002/cjs.11813}

\bibitem[Tian, Zhang, and Tan(2026)]{tian2026etm}
Tian, Y., Zhang, X., and Tan, Z. (2026).
\newblock On semi-supervised estimation using exponential tilt mixture models.
\newblock \emph{Journal of Statistical Planning and Inference}, 241, 106314.
\newblock \url{https://doi.org/10.1016/j.jspi.2025.106314}

\bibitem[van der Vaart and Wellner(1996)]{vdVWellner1996}
van der Vaart, A. W., and Wellner, J. A. (1996). \emph{Weak Convergence and Empirical Processes}. Springer.
\newblock \url{https://doi.org/10.1007/978-1-4757-2545-2}

\bibitem[van der Vaart and Wellner(2011)]{vdVWellner2011}
van der Vaart, A. W., and Wellner, J. A. (2011).
\newblock A local maximal inequality under uniform entropy.
\newblock \emph{Electronic Journal of Statistics}, 5, 192--203.
\newblock \url{https://doi.org/10.1214/11-EJS605}

\bibitem[van Engelen and Hoos(2020)]{vanengelen2020}
van Engelen, J. E., and Hoos, H. H. (2020).
\newblock A survey on semi-supervised learning.
\newblock \emph{Machine Learning}, 109(2), 373--440.
\newblock \url{https://doi.org/10.1007/s10994-019-05855-6}

\bibitem[Watson(1964)]{watson1964}
Watson, G. S. (1964).
\newblock Smooth regression analysis.
\newblock \emph{Sankhya Series A}, 26(4), 359--372.
\newblock \url{https://www.jstor.org/stable/25049340}

\bibitem[Wen et al.(2025)]{wen2025distribution}
Wen, M., Jia, Y., Ren, H., Wang, Z., and Zou, C. (2025).
\newblock Semi-supervised distribution learning.
\newblock \emph{Biometrika}, 112(1), asae056.
\newblock \url{https://doi.org/10.1093/biomet/asae056}

\bibitem[Zhang and Bradic(2022)]{zhangbradic2022}
Zhang, Y., and Bradic, J. (2022).
\newblock High-dimensional semi-supervised learning: in search of optimal inference of the mean.
\newblock \emph{Biometrika}, 109(2), 387--403.
\newblock \url{https://doi.org/10.1093/biomet/asab042}

\bibitem[Zhang, Brown, and Cai(2019)]{zhang2019}
Zhang, A., Brown, L. D., and Cai, T. T. (2019).
\newblock Semi-supervised inference: General theory and estimation of means.
\newblock \emph{The Annals of Statistics}, 47(5), 2538--2566.
\newblock \url{https://doi.org/10.1214/18-AOS1756}

\bibitem[Zhang, Peng, and Zhou(2025)]{zhangpeng2025}
Zhang, M., Peng, M., and Zhou, Y. (2025).
\newblock Semi-supervised learning for various comparison functions across two populations.
\newblock \emph{Statistical Papers}, 66, Article 18.
\newblock \url{https://doi.org/10.1007/s00362-024-01632-3}

\end{thebibliography}
\end{document}